\documentclass[onecolumn,twoside]{IEEEtran}
\usepackage{ifpdf}
\usepackage{cite}
\usepackage{tikz}
\usepackage{lipsum}
\usepackage{mathtools}
\usepackage{cuted}
\usepackage{balance}
\usetikzlibrary{decorations, decorations.text,}
\usepackage{float}
\usepackage{amsmath}
\usepackage{amsthm}
\usepackage{array}
\usepackage[caption=false,font=footnotesize]{subfig}
\usepackage{textcomp}
\usepackage{url}
\usepackage{paralist,esint,amssymb,cases,bm,galois}
\usepackage{cleveref}

\newcommand\Tx[1]{\mathrm{#1}}
\newcommand\Se[1]{\mathcal{#1}}
\newcommand\Db[1]{\mathbb{#1}}

\newcommand\MB[1]{\left[#1\right]}
\newcommand\LB[1]{\{#1\}}
\newcommand\SB[1]{\left(#1\right)}

\newcommand{\RN}[1]{\textup{\uppercase\expandafter{\romannumeral#1}}}

\newtheorem{theo}{Theorem}

\newtheorem{lemma}{Lemma}
\newtheorem{exam}{Example}
\newtheorem{defi}{Definition}
\newtheorem{rem}{Remark}
\newtheorem{cons}{Construction}

\usepackage{algorithm}
\usepackage{algpseudocode}
\makeatletter
\def\BState{\State\hskip-\ALG@thistlm}
\makeatother

\definecolor{apple green}{rgb}{0.17,0.75,0.13}
\definecolor{lara}{rgb}{0,0,0}
\definecolor{robert}{rgb}{0.1,0.1,0.8}
\definecolor{prisca}{rgb}{0.1,0.8,0.1}

\errorcontextlines\maxdimen
\makeatletter
\newcommand*{\algrule}[1][\algorithmicindent]{\makebox[#1][l]{\hspace*{.5em}\vrule height 0.9 \baselineskip depth 0.3\baselineskip}}%

\newcount\ALG@printindent@tempcnta
\def\ALG@printindent{%
    \ifnum \theALG@nested>0
        \ifx\ALG@text\ALG@x@notext
            \addvspace{0pt}
        \else
            \unskip
            \ALG@printindent@tempcnta=1
            \loop
                \algrule[\csname ALG@ind@\the\ALG@printindent@tempcnta\endcsname]%
                \advance \ALG@printindent@tempcnta 1
            \ifnum \ALG@printindent@tempcnta<\numexpr\theALG@nested+1\relax
            \repeat
        \fi
    \fi
    }%
\usepackage{etoolbox}
\patchcmd{\ALG@doentity}{\noindent\hskip\ALG@tlm}{\ALG@printindent}{}{\errmessage{failed to patch}}
\makeatother



\makeatletter
\def\old@comma{,}
\catcode`\,=13
\def,{%
  \ifmmode%
    \old@comma\discretionary{}{}{}%
  \else%
    \old@comma%
  \fi%
}
\makeatother

\begin{document}
\bstctlcite{IEEEexample:BSTcontrol}
\title{Hierarchical Coding for Cloud Storage: Topology-Adaptivity, Scalability, and Flexibility}

\author{Siyi~Yang,~\IEEEmembership{Student Member,~IEEE},
        Ahmed~Hareedy,~\IEEEmembership{Member,~IEEE},
        Robert~Calderbank,~\IEEEmembership{Fellow,~IEEE},
        and~Lara~Dolecek,~\IEEEmembership{Senior Member,~IEEE}
\thanks{Elements of this paper were presented in part at the IEEE Global Communications Conference, Waikoloa, Hawaii, USA, December 2019 \cite{Yang2019HC}, and at the IEEE International Symposium on Information Theory, Los Angeles, California, USA, June 2020 \cite{Yang2020DSN}.}
\thanks{Siyi Yang and Lara Dolecek are with the Department of Electrical and Computer Engineering, University of California, Los Angeles, Los Angeles, CA 90095 USA (email: siyiyang@ucla.edu; dolecek@ee.ucla.edu).} 
\thanks{Ahmed Hareedy and Robert Calderbank are with the Department of Electrical and Computer Engineering, Duke University, Durham, NC 27708 USA (email: ahmed.hareedy@duke.edu; robert.calderbank@duke.edu).}
}

\maketitle

\begin{abstract} In order to accommodate the ever-growing data from various, possibly independent, sources and the dynamic nature of data usage rates in practical applications, modern cloud data storage systems are required to be scalable, flexible, and heterogeneous. The recent rise of the blockchain technology is also moving various information systems towards decentralization to achieve high privacy at low costs. While codes with hierarchical locality have been intensively studied in the context of centralized cloud storage due to their effectiveness in reducing the average reading time, those for decentralized storage networks (DSNs) have not yet been discussed. In this paper, we propose a joint coding scheme where each node receives extra protection through the cooperation with nodes in its neighborhood in a heterogeneous DSN with any given topology. This work extends and subsumes our prior work on coding for centralized cloud storage. In particular, our proposed construction not only preserves desirable properties such as scalability and flexibility, which are critical in dynamic networks, but also adapts to arbitrary topologies, a property that is essential in DSNs but has been overlooked in existing works.

\end{abstract}

\IEEEpeerreviewmaketitle

\section{Introduction}
\label{sectoin: introduction}

\IEEEPARstart{I}{n} response to the rapidly growing demand of data management, cloud storage such as Microsoft Azure and Amazon Web Services \textcolor{lara}{have} become among the most widely deployed public cloud services. In these centralized cloud services, a tech giant takes full custodianship over data of all its customers\textcolor{lara}{; this} situation can result in expensive infrastructure maintenance and may lead to privacy violations. Decentralized storage networks (DSNs) such as Storj \cite{storj2018}, in which no entity is solely responsible for all data, have emerged as a secure and \textcolor{lara}{economic} alternative to centralized cloud services. DSNs are believed to be economically attractive since extra capacity can be afforded by utilizing idle storage space on devices at the edge of the network. Despite all advantages of decentralization, practical management of personal devices also faces challenges from component failures, high churn rates, heterogeneous bandwidths and link speeds, in addition to dynamic node balancing for content delivery of hot files \cite{storj2018}. While erasure correction (EC) codes are widely used to combat component failures, EC schemes that address the aforementioned issues are relatively overlooked. In this paper, we propose EC solutions that are tailored to tackle those challenges pertaining to DSNs.

Latency and reliability are among the most critical factors that customers care about \textcolor{lara}{in} cloud storage. However, DSNs naturally impose numerous challenges on simultaneously maintaining low latency and high reliability. EC solutions with large block lengths are more resilient to large weight errors, but they simultaneously slow down the recovery for the more frequent cases where only few erasures occur. To reach a better trade-off between data reliability and latency, codes enabling multi-level access are desired. In these codes, any node is allowed to access different sets of helper nodes to retrieve the data, where the sizes of the sets get reduced if the number of erasures to be recovered is small enough. This architecture is referred to as codes with hierarchical localities. While hierarchical coding in the context of centralized storage \cite{hassner2001integrated,huang2017multi,cassuto2017multi,wu2017generalized,ballentine2018codes,zhang2018generalized,blaum2018extended,balaji2019tight,Yang2019HC} has been intensively studied, codes for DSNs have been mostly discussed without considering localities \cite{dimakis2010distributed,kong2010decentralized,ye2018cooperative}.

More recently, codes with localities in multi-rack storage, a special case of DSNs, have also been investigated, where either the system is considered to be homogeneous \cite{tebbi2019multi,hou2019rack,chen2019explicit,prakash2018storage}, or the network topology has a simple structure \cite{li2010tree,pernas2013non}. However, DSNs typically have more sophisticated topologies characterized by heterogeneity among bandwidths of communication links and erasure statistics of nodes due to the arbitrary and dynamic nature of practical networks \cite{pernas2013non,wang2014heterogeneity,ibrahim2016green,sipos2018network,sipos2016erasure}. Instead of solutions for simplified models, schemes that fit into any topology (a property referred to as \textcolor{lara}{\textbf{topology-adaptivity}} later on) with customizable data lengths and redundancies are desired to exploit the existing resources.

Another major challenge for DSNs comes from the high churn rate, namely, participants join the network and leave without a predictable pattern. Therefore, it is essential for a DSN to enable its organic growth, i.e., enable expanding the backbone network to accommodate additional node operators, without rebuilding the entire infrastructure \cite{rimal2009taxonomy}; this property is referred to as \textcolor{lara}{\textbf{scalability}}. Moreover, data stored at certain nodes occasionally become \textcolor{lara}{hotter} than anticipated\textcolor{lara}{, and the download rate can thus exceed} the bandwidth limit. \textcolor{lara}{In such a scenario,} dynamic node balancing is required for content delivery \textcolor{lara}{to reach a lower latency}. In particular, the cloud (node) should be split into smaller clouds without worsening the global erasure correction capability or changing the remaining components. This property is referred to as \textcolor{lara}{\textbf{flexibility}} and has been firstly investigated for dynamic data storage systems under the discussion of \textcolor{lara}{the} so-called sum-rank codes \cite{martnez2018universal}. However, sum-rank codes require a \textcolor{lara}{Galois field} size that grows \textit{\textcolor{lara}{exponentially}} with the maximum local block length, which is a major obstacle to being implemented in real world applications \cite{martnez2018universal}.

In this paper, we strategically combine hierarchical locality and topological properties of a DSN. We develop a hierarchical coding scheme that is topologically-adaptive. The scheme is built upon our prior work on centralized cloud storage \cite{Yang2019HC} and preserves desirable properties including scalability and flexibility. The \textcolor{lara}{Galois field} size of this scheme grows \textit{\textcolor{lara}{linearly}} with the local block length. Our proposed coding scheme enables joint encoding and decoding of the data stored at all nodes such that nodes in \textcolor{lara}{a} neighborhood cooperatively protect and validate their stored data collectively in the DSN. Cooperation in DSNs further improves the reliability since information propagates from more reliable nodes to less reliable nodes through paths connecting them.

The rest of the paper is organized as follows. In \Cref{section notation and preliminaries}, we introduce the DSN model and necessary preliminaries. In \Cref{section coop}, we define \textcolor{lara}{erasure correction (EC)} hierarchies as well as their depth to systematically describe the \textcolor{lara}{maximal number of recoverable erasures} corresponding to different access levels. We present a coding scheme with depth $1$ that results in a better recovery speed \textcolor{lara}{compared} with existing schemes that are not topologically-adaptive \cite{li2010tree,pernas2013non}. We also discuss the recoverable erasure patterns of the proposed construction and show that our scheme enables correction of erasure patterns relevant to DSNs. In \Cref{section multi-level cooperation}, we extend the single-level construction (depth $1$) to have higher-level cooperation. In the hierarchical scheme, the cooperation between nodes in the DSN is described by the so-called \textcolor{lara}{\textbf{cooperation graphs}}. We also present sufficient conditions on any graph to be a cooperation graph, and refer to graphs satisfying these conditions as \textcolor{lara}{\textbf{compatible graphs}}. In \Cref{section desired properties}, we first present an algorithm that searches for a cooperation graph on any DSN with a given topology. Next, we show that our coding scheme supports scalability and flexibility. Finally, we summarize our results in \Cref{section conclusion}.

\section{Notation and Preliminaries}
\label{section notation and preliminaries}

In this section, we discuss the model and mathematical representation of a DSN, as well as necessary preliminaries. Throughout the remainder of this paper, $\MB{N}$ refers to $\{1,2,\dots,N\}$. 
For a vector $\bold{v}$ of length $n$, $v_i$, $1\leq i\leq n$, represents the $i$-th component of $\bold{v}$, and $\bold{v}\MB{a:b}=(v_{a},\dots,v_b)$. For a matrix $\bold{M}$ of size $a\times b$, $\bold{M}\MB{i_1:i_2,j_1:j_2}$ represents the sub-matrix $\bold{M}'$ of $\bold{M}$ such that $(\bold{M}')_{i-i_1+1,j-j_1+1}=(\bold{M})_{i,j}$, $i\in\MB{i_1:i_2}$, $j\in\MB{j_1:j_2}$. For vectors $\bold{u}$ and $\bold{v}$ of the same length $p$, $\bold{u}\succ\bold{v}$ and $\bold{u}\prec \bold{v}$ means $u_i>v_i$ and $u_i<v_i$, for all $i\in\left[p\right]$, respectively; $\bold{u}\succeq \bold{v}$ and $\bold{u}\preceq \bold{v}$ means $u_i \geq v_i$ and $u_i \leq v_i$, for all $i\in\left[p\right]$, respectively. For any $m,n\in\mathbb{N}$, an identity matrix of size $n\times n$ is denoted by $\bold{I}_n$, and a zero matrix of size $m \times n $ is denoted by $\bold{0}_{m\times n}$. For any $q\in \mathbb{N}$, $\textup{GF}(q)$ refers to a Galois field with size $q$. In this paper, we constrain $q$ to be a power of $2$.

\begin{figure}
\centering
\includegraphics[width=0.5\textwidth]{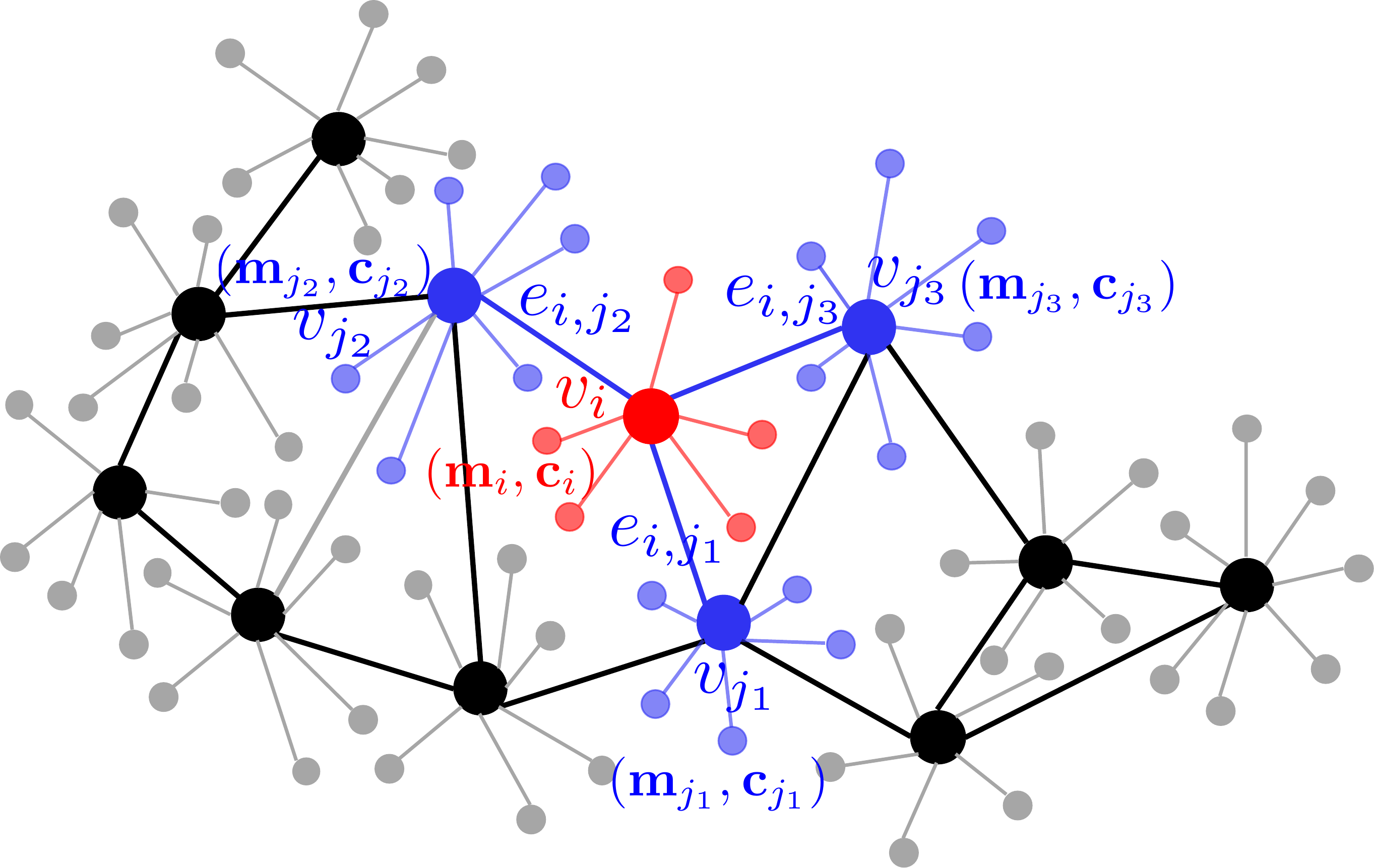}
\caption{Decentralized storage network (DSN). For the cluster with the master node $v_i$, message $\bold{m}_i$ is encoded to $\bold{c}_i$, and symbols of $\bold{c}_i$ are stored distributively among non-master nodes that are locally connected to $v_i$. In the figures after Fig.~\ref{fig: figmodel}, we omit the local non-master nodes for clarity of figures.}
\label{fig: figmodel}
\end{figure}

\subsection{Decentralized Storage Network}
\label{subsec: DSN model}
In a DSN, nodes are typically partitioned into distributed clusters of nodes, where each cluster has a ``master node'' that functions in this cluster similar to that of a central node in a centralized network, which is what the ``decentralization'' refers to. As shown in Fig.~\ref{fig: figmodel}, each master node, represented by big bold-colored nodes, communicates with both its neighboring master nodes and other nodes in the cluster it belongs to, whereas each non-master node, represented by small light-colored nodes, only communicates with the master node of the cluster it belongs to. For the cluster with the master node $v_i$, message $\bold{m}_i$ is encoded to $\bold{c}_i$, and symbols of $\bold{c}_i$ are stored distributively among non-master nodes that are locally connected to $v_i$. For clarity and simplicity of figures and notation, we omit the non-master nodes in figures after Fig.~\ref{fig: figmodel}. We then refer to master nodes and the communication links among them by ``nodes'' and ``edges'', respectively, in the remainder of this paper.

As shown in Fig.~\ref{fig: figmodel}, a DSN is modeled as a graph $G(V,E)$, where $V$ and $E$ denote the set of nodes (master only) and edges, respectively. Codewords are stored among the nodes in a cluster. A failed node in a cluster is regarded as an erased symbol in the codeword stored at this cluster. A cluster is represented in $G$ by its master node $v_i\in V$ solely. Each edge $e_{i,j}\in E$ represents a communication link connecting node $v_i$ and node $v_j$, through which $v_i$ and $v_j$ are allowed to exchange information. Denote the set of all neighbors of node $v_i$ by $\mathcal{N}_i$, e.g., $\mathcal{N}_i=\{v_{j_1},v_{j_2},v_{j_3}\}$ in Fig.~\ref{fig: figmodel}, and refer to it as the \textbf{neighborhood} of node $v_i$. Messages $\{\bold{m}_i\}_{v_i\in V}$ are jointly encoded as $\{\bold{c}_i\}_{v_i\in V}$, and $\bold{c}_i$ is stored at the cluster of nodes containing $v_i$. For simplicity, we instead say ``$\bold{c}_i$ is stored at node $v_i$'' in the rest of the paper.

For a DSN denoted by $G(V,E)$, let $p=|V|$. Suppose $G$ is associated with a tuple $(\bold{n},\bold{k},\bold{r})\in\left(\mathbb{N}^p\right)^3$, where $\bold{k},\bold{r}\succ\bold{0}$ and $\bold{n}=\bold{k}+\bold{r}$. Note that $k_i$ represents the length of the message $\bold{m}_i$ associated with $v_i\in V$; $n_i$ and $r_i$ denote the length of $\bold{c}_i$ stored at $v_i$ and its syndrome, respectively. Let $\bold{m}=(\bold{m}_1,\bold{m}_2,\dots,\bold{m}_p)$, and $\bold{c}=(\bold{c}_1,\bold{c}_2,\dots,\bold{c}_p)$. A set $\Se{C}$ is called an $(n,k,d)_q$-code if $\Se{C}\subset \textup{GF}(q)^n$, $\Tx{dim}(\Se{C})=k$, and $\min\limits_{\bold{c},\bold{c}'\in \Se{C}, \bold{c}\neq\bold{c}'} d_\textup{H}(\bold{c},\bold{c}')-1=d$, where $d_\textup{H}$ refers to the Hamming distance. We next define a family of codes with double-level access. Note that our discussion is restricted to linear block codes.

\subsection{Cauchy Matrices}
\label{subsec: cauchy matrices}
\textcolor{lara}{Before we describe the constructions in detail}, we first introduce the so-called \textbf{Cauchy matrices} that are used as major components in the generator matrices of our codes. Codes based on Cauchy matrices, the so-called Cauchy Reed-Solomon (CRS) codes, have been studied in \cite{van1986minimum,bloemer1995xor}. CRS codes present desirable properties, as discussed later, and have been proposed to be applied to distributed storage systems in \cite{plank2006optimizing,wu2015efficient}. In our work, we further exploit the scaling property of CRS codes, which makes them an ideal choice to accommodate hierarchical access on arbitrarily deployed nodes in DSNs.

\begin{defi} \emph{\textbf{(Cauchy matrix)}} \label{CauchyMatrix} Let $s,t\in\Db{N}$ and $\textup{GF}(q)$ be a Galois field of size $q$. Suppose $a_1,\dots,a_s,b_1,\dots,b_t$ are $s + t$ distinct elements in $\textup{GF}(q)$. The following matrix is known as a \textbf{Cauchy matrix},
\begin{equation*}\left[
\begin{array}{cccc}
\frac{1}{a_1-b_1} & \frac{1}{a_1-b_2} & \dots & \frac{1}{a_1-b_t}\\
\frac{1}{a_2-b_1} & \frac{1}{a_2-b_2} & \dots & \frac{1}{a_2-b_t}\\
\vdots & \vdots &\ddots & \vdots \\
\frac{1}{a_s-b_1} & \frac{1}{a_s-b_2} & \dots & \frac{1}{a_s-b_t}\\
\end{array}\right].
\label{defi: Cauchy matrix}
\end{equation*}
We denote this matrix by $\bold{Y}(a_1,\dots,a_s;b_1,\dots,b_t)$, and refer to sequences $(a_1,a_2,\dots,a_s)$, $(b_1,b_2,\dots,b_t)$ as the \textbf{row indicator} and the \textbf{column indicator} of the specified Cauchy matrix, respectively.
\end{defi}

Cauchy matrices \textcolor{lara}{possess desirable properties} that make them \textcolor{lara}{an ideal} alternative to Vandermonde matrices, the major components of the parity-check matrices of Reed-Solomon (RS) codes, as the parity-computing (non-systematic) components in systematic generator matrices of some maximum distance separable (MDS) codes with low encoding and decoding complexities \cite{plank2006optimizing}. 
Cauchy matrices are \textbf{totally invertible}, i.e., every square sub-matrix of a Cauchy matrix is invertible. Therefore, horizontally concatenating a Cauchy matrix with another Cauchy matrix having an identical row indicator but a non-overlapping column indicator results in a third Cauchy matrix. Similarly, vertically concatenating a Cauchy matrix with another Cauchy matrix having an identical column indicator but a non-overlapping row indicator also results in a third Cauchy matrix. This property, referred to as the scaling property previously, \textcolor{lara}{is} desirable for hierarchical access in topology-adaptive DSNs. Moreover, \Cref{lemma: Good matrix} presents another useful property about Cauchy matrices, which will be used repeatedly in this paper.

\begin{lemma} \label{lemma: Good matrix} Let $s,t,r\in\Db{N}$ such that $t-s<r\leq t$, and $\bold{A}\in \textup{GF}(q)^{s\times t}$. If $\bold{A}$ is a Cauchy matrix, then the following matrix $\bold{M}$ is a parity-check matrix of an $(s+r,s+r-t,t+1)_q$-code\footnote{Note that when $q$ is a power of $2$, the minus operand can be removed, as shown in \Cref{exam: CodeDL}, since subtraction and addition are equivalent on the Galois field $\textup{GF}(q)$ in this case.},
\begin{equation*}
\bold{M}=\left[
\begin{array}{c}
\bold{A}\\
-\bold{I}_r\ \bold{0}_{r\times(t-r)}\\
\end{array}
\right]^{\Tx{T}}.
\end{equation*}
\end{lemma} 

\begin{proof} The parity-check matrix of an $(s+r,s+r-t,t+1)_q$-code satisfies the property that every $t$ columns of this matrix are linearly independent. Therefore, we only need to prove that every $t$ rows of $\bold{M}^{\Tx{T}}$ are linearly independent. We prove Lemma~\ref{lemma: Good matrix} by contradiction. Suppose there exist $t$ rows from $\bold{M}^{\Tx{T}}$ that are linearly dependent. Suppose $a$ of these linearly dependent rows $\bold{r}_1,\bold{r}_2,\dots,\bold{r}_a$ are from $\bold{A}$, and the other $t-a$ rows $\bold{r}_{a+1},\bold{r}_{a+2},\dots,\bold{r}_{t}$ are from $\MB{-\bold{I}_r\ \bold{0}_{r\times(t-r)}}$, where $0\leq t-a\leq r$. Suppose the entries with $-1$ in $\bold{r}_{a+1},\bold{r}_{a+2},\dots,\bold{r}_{t}$ are located in the $i_1,i_2\dots,i_{t-a}$-th columns of $\bold{M}^{\textup{T}}$, then $i_p\leq r$ for all $1\leq p\leq t-a$. Observe that $\MB{t}$ is the set of indices of all columns in $\bold{M}^{\Tx{T}}$. Suppose $\MB{t}\setminus \LB{i_1,i_2,\dots,i_{t-a}}=\LB{j_1,j_2,\dots,j_a}$. \textcolor{lara}{Then,} the $a\times a$ sub-matrix of the intersection of the rows $\bold{r}_1,\bold{r}_{2},\dots,\bold{r}_a$ and the $j_1,j_2,\dots,j_a$-th columns of $\bold{A}$ is singular. A contradiction.
\end{proof}

\subsection{Locality of Interleaved Cauchy Reed Solomon Codes}
\label{subsection CRS codes}
A code is \textbf{systematic} if the codeword contains a segment that is identical to the message being encoded. For a linear block code, systematic encoding of messages with length $k$ is performed via a generator matrix containing a $k\times k$ submatrix being the identity matrix $\bold{I}_k$. Systematic codes are of interest because of their low complexity mapping from any valid codeword to the message it represents, as well as their low encoding complexity due to the fact that only parities need extra calculation steps. Based on the aforementioned notation, a systematic generator matrix of a code on $G(V,E)$ has the following structure:
\begin{equation}\label{eqn: GenMatDL}\bold{G}=\left[
\begin{array}{c|c|c|c|c|c|c}
\bold{I}_{k_1} & \bold{A}_{1,1} & \bold{0} & \bold{A}_{1,2} & \dots & \bold{0} & \bold{A}_{1,p}\\
\hline
\bold{0} & \bold{A}_{2,1} & \bold{I}_{k_2} & \bold{A}_{2,2} & \dots & \bold{0}& \bold{A}_{2,p}\\
\hline
\vdots & \vdots & \vdots & \vdots & \ddots & \vdots & \vdots \\
\hline
\bold{0} & \bold{A}_{p,1} & \bold{0} & \bold{A}_{p,2} & \dots & \bold{I}_{k_p} & \bold{A}_{p,p}\\
\end{array}\right],
\end{equation}
where all elements are from a Galois field $\textup{GF}(q)$, $q=2^\theta$ and $\theta \geq 2$. 

Following the notation in the previous subsections, the codeword at node $v_i$ is $\bold{c}_i=(\bold{m}_i,\sum\nolimits_{j\in\MB{p}}\bold{m}_j\bold{A}_{j,i})$, and it has two parts. We call $\bold{m}_i$ the systematic part, and $\sum\nolimits_{j\in\MB{p}}\bold{m}_j\bold{A}_{j,i}$ the \textbf{local parities} of $\bold{c}_i$. More specifically, we refer to $\sum\nolimits_{j\in\MB{p},j\neq i}\bold{m}_j\bold{A}_{j,i}$, $\bold{m}_i\bold{A}_{i,i}$ as the \textbf{additional local parities} and the \textbf{original local parities} at $v_i$, respectively. For any $j\in\MB{p}$, $j\neq i$, symbols in $\bold{m}_i\bold{A}_{i,j}$ are referred to as the \textbf{cross parities} of $v_i$ from node $v_j$. Note that by saying ``parities'' we actually mean ``parity symbols''. We use these two terms interchangeably in the remaining text.

The submatrices $\{\bold{A}_{i,j}\}_{i,j\in\MB{p}}$ in our construction are either zero matrices, Cauchy matrices, or products of Cauchy matrices. For this reason, we call codes represented by a generator matrix in (\ref{eqn: GenMatDL}) as \textbf{interleaved Cauchy Reed Solomon (CRS) codes}. The primary property of interleaved CRS codes is that each local message $\bold{m}_i$ is to be obtained locally by only accessing the codeword $\bold{c}_i$ stored at $v_i$ if the number of erasures in $\bold{c}_i$ does not exceed an upper bound that is determined by some local parameters. We next provide an exemplary construction, \Cref{cons: CRScons}, to illustrate the locality of interleaved CRS codes. 

\begin{cons}\label{cons: CRScons} (Interleaved CRS codes) Let $p\in\Db{N}$, $k_1,k_2,\dots,k_p\in \Db{N}$, $n_1,n_2,\dots,n_p\in\Db{N}$, $\delta_1,\delta_2,\dots,\delta_p\in \Db{N}$, with $r_x=n_x-k_x>0$ for all $x\in\MB{p}$. Let $P=(\MB{p}\times\MB{p})\setminus\{(i,i)\}_{i\in\MB{p}}$, and $I\subseteq P$ is such that for all $(x,y) \in I$, $\bold{A}_{x,y}$ is non-zero. {Let $I_x=\{i: (x,i)\in I\}$, for each $x\in \MB{p}$, and suppose $I_x=\{i_1,i_2,\cdots,i_{|I_x|}\}$.} Let $\delta'_x=\sum\nolimits_{y\in I_x}\delta_y$, for all $x\in\MB{p}$. Let $\textup{GF}(q)$ be a \textcolor{lara}{Galois field} such that $q\geq \max\nolimits_{x\in\MB{p}}\LB{n_x+\delta'_x}$. 

For each $x\in \MB{p}$, let $a_{x,i}$, $i\in\MB{k_x+\delta_x}$, and $b_{x,j}$, $j\in\MB{r_x-\delta_x+\delta'_x}$, be distinct elements of $\textup{GF}(q)$.
Consider the Cauchy matrix $\bold{T}_x\in \textup{GF}(q)^{(k_x+\delta_x)\times (r_x-\delta_x+\delta'_x)}$ such that $\bold{T}_x=\bold{Y}(a_{x,1}, \dots, a_{x,k_x+\delta_x};b_{x,1},\dots,b_{x,r_x-\delta_x+\delta'_x})$. For each $x\in\MB{p}$, we obtain $\{\bold{B}_{x,i}\}_{i\in{I_x}}$, $\bold{U}_x$, $\bold{A}_{x,x}$, according to the following partition of $\bold{T}_x$:

\begin{equation}\label{eqn: CRS}
\bold{T}_x=\left[
\begin{array}{c|c}
\bold{A}_{x,x} & \begin{array}{c|c|c}
\bold{B}_{x,i_1} & \dots & \bold{B}_{x,i_{|I_x|}}
\end{array}
\\
\hline
\bold{U}_x & \bold{Z}_{x}
\end{array}\right],
\end{equation}
where $\bold{A}_{x,x}\in \textup{GF}(q)^{k_x\times r_x}$, $\bold{B}_{x,i}\in \textup{GF}(q)^{k_x\times \delta_i}$, $\bold{U}_x\in \textup{GF}(q)^{\delta_x\times r_x}$. Moreover, let $\bold{A}_{x,y}=\bold{B}_{x,y}\bold{U}_y$, for $(x,y)\in I$; let $\bold{A}_{x,y}=\bold{0}_{k_x \times r_y}$, for $(x,y)\in P\setminus I$.

Matrices $\bold{A}_{x,x}$ and $\bold{A}_{x,y}$ are substituted in $\bold{G}$ specified in (\ref{eqn: GenMatDL}), for all $x,y\in\MB{p}$. Let $\Se{C}_1$ represent the code with \textcolor{lara}{the} generator matrix $\bold{G}$.
\end{cons}

Following the notation in Subsection~\ref{subsec: DSN model}, \textcolor{lara}{suppose in a DSN that is implemented with a code specified in \Cref{cons: CRScons}, all nodes} are able to communicate with each other. For all $x\in\MB{p}$, let $d_{x,1}$, $d_{x,2}$ represent the maximum number of erasures that node $v_x$ can tolerate with local access to the codeword $\bold{c}_x$, and global access to all the codewords $\{\bold{c}_x\}_{x\in\MB{p}}$, respectively. \Cref{lemma: DLcodedis} presents the value of the local and the global correction capabilities of codes proposed in \Cref{cons: CRScons}. Note that even though $\bold{m}_j\bold{A}_{j,i}=\bold{m}_j\bold{B}_{j,i}\bold{U}_i$ gives the explicit cross parities, symbols resulting from $\bold{m}_j\bold{B}_{j,i}$ can accurately be seen as the cross parities too since they constitute a set of independent linear combinations of message symbols, and they contain all the information node $v_j$ provides to node $v_i$, for all $v_i,v_j\in V, i \neq j$. Therefore, in the remainder of this paper, we also refer to $\bold{m}_j\bold{B}_{j,i}$ as the cross parities or the cross parity symbols for simplicity.

\begin{lemma}\label{lemma: DLcodedis} In code $\Se{C}_1$ specified in \Cref{cons: CRScons}, $d_{x,1}=r_x-\delta_x$, $d_{x,2}=r_x+\delta'_x$, for $x\in\MB{p}$.
\end{lemma}

\begin{proof} For each $x\in \MB{p}$, define $\bold{y}_x=\sum\nolimits_{y\in I_x} \bold{m}_y\bold{B}_{y,x}$. It follows from $\bold{m}\bold{G}=\bold{c}$ and (\ref{eqn: GenMatDL}) that for $x\in\MB{p}$, $\bold{c}_x=\MB{\bold{m}_x,\bold{m}_x\bold{A}_{x,x}+\bold{y}_x\bold{U}_x}$. Define the local parity-check matrix $\bold{H}^{\Tx{L}}_x$ and the global parity-check matrix $\bold{H}^{\Tx{G}}_x$, for each $x\in\MB{p}$, as follows:
\begin{equation*}
\bold{H}_x^{\Tx{G}}=\left[
\begin{array}{c|c}
\bold{A}_{x,x} & \begin{array}{c|c|c}
\bold{B}_{x,i_1} & \dots & \bold{B}_{x,i_{|I_x|}}
\end{array}
\\
\hline
-\bold{I}_{r_x} & \bold{0}_{r_x \times \delta'_x}
\end{array}\right]^{\Tx{T}}, \textup{ } \bold{H}^{\Tx{L}}_x=\left[\begin{array}{ccc}
\bold{A}_{x,x}\\
-\bold{I}_{r_x}\\
\bold{U}_x\\
\end{array}\right]^{\Tx{T}}.
\end{equation*}
We next prove the equations of the local correction capability $d_{x,1}=r_x-\delta_x$ and the global correction capability $d_{x,2}=r_x+\delta'_x$ using $\bold{H}^{\Tx{L}}_x$ and $\bold{H}^{\Tx{G}}_x$, $x\in\MB{p}$.

To prove the equation of the local correction capability, let $\tilde{\bold{c}}_x=\MB{\bold{c}_x,\bold{y}_x}$. Then, one can show that $\tilde{\bold{c}}_x$ belongs to a code $\Se{C}_x^{\Tx{L}}$ with the local parity-check matrix $\bold{H}^{\Tx{L}}_x$. From \Cref{lemma: Good matrix}, $\Se{C}_x^{\Tx{L}}$ is an $(n_x+\delta_x,k_x,r_x+1)_q$-code. Therefore, any $r_x$ erasures in $\tilde{\bold{c}}_x$ are correctable. Provided that $\bold{y}_x$ has length $\delta_x$, we can consider the entries of $\bold{y}_x$ as erasures, and thus any $(r_x-\delta_x)$ erasures in the remaining part of $\tilde{\bold{c}}_x$, i.e., $\bold{c}_x$, can be corrected. Therefore, $d_{x,1}=r_x-\delta_x$.

To prove the equation of the global correction capability, assume all the local codewords except for $\bold{c}_x$ are successfully decodable locally. For each $x\in\MB{p}$, let $\bold{s}_{x}=\MB{\bold{m}_x\bold{B}_{x,i_1},\dots,\bold{m}_x\bold{B}_{x,i_{|I_x|}}}$ and $\bar{\bold{c}}_x=\bold{c}_x-\MB{\bold{0}_{k_x},\bold{y}_x\bold{U}_x}$. Then, one can show that $\bold{H}^{\Tx{G}}_x\bar{\bold{c}}_x^{\Tx{T}}=\MB{\bold{0}_{r_x}, \bold{s}_{x}}^{\Tx{T}}$. From \Cref{lemma: Good matrix} and from the construction of $\bold{H}^{\Tx{G}}_x$, any $(r_x+\delta'_x)$ erasures in $\bar{\bold{c}}_x$ are correctable, and thus $(r_x+\delta'_x)$ erasures in $\bold{c}_x$ are also correctable. Therefore, $d_{x,2}=r_x+\delta'_x$.
\end{proof}

Next, we give an example of \Cref{cons: CRScons} where $I=P$. Namely, all the $\bold{A}_{x,y}$, $(x,y)\in P$, are non-zero matrices. \Cref{exam: CodeDL} also demonstrates how erasures are corrected.

\begin{table}
\centering
\caption{Polynomial and binary representation of $\textup{GF}(2^4)$}
\begin{tabular}{|c|c||c|c||c|c||c|c|}
\hline
$0$ & $0000$ & $\beta^4$ & $1100$ & $\beta^{8}$ & $1010$ & $\beta^{12}$ & $1111$\\
\hline
$\beta$ & $0100$ & $\beta^{5}$ & $0110$ & $\beta^{9}$ & $0101$ & $\beta^{13}$ & $1011$\\
\hline
$\beta^2$ & $0010$ & $\beta^{6}$ & $0011$ & $\beta^{10}$ & $1110$ & $\beta^{14}$ & $1001$\\
\hline
$\beta^3$ & $0001$ & $\beta^{7}$ & $1101$ & $\beta^{11}$ & $0111$ & $\beta^{15}=1$ & $1000$\\
\hline
\end{tabular}
\label{table: GF}
\end{table}

\begin{figure*}
\normalsize
\setcounter{equation}{2}
\begin{equation}\label{eqn: exam1}\small
\bold{T}_{1}=\bold{T}_{2}=\left[\begin{array}{c|c}
\bold{A}_{1,1} & \bold{B}_{1,2}\\
\hline
\bold{U}_1 & \bold{Z}_1
\end{array}\right]=\left[\begin{array}{c|c}
\bold{A}_{2,2} & \bold{B}_{2,1}\\
\hline
\bold{U}_2 & \bold{Z}_2
\end{array}\right]
=\left[\begin{array}{ccc|c}
\frac{1}{\beta^{}-\beta^{8}} & \frac{1}{\beta^{}-\beta^{9}} & \frac{1}{\beta^{}-\beta^{10}} & \frac{1}{\beta-\beta^{11}}\\
\frac{1}{\beta^{2}-\beta^{8}} & \frac{1}{\beta^{2}-\beta^{9}} & \frac{1}{\beta^{2}-\beta^{10}} &\frac{1}{\beta^{2}-\beta^{11}}\\
\frac{1}{\beta^{3}-\beta^{8}} & \frac{1}{\beta^{3}-\beta^{9}} & \frac{1}{\beta^{3}-\beta^{10}} &\frac{1}{\beta^{3}-\beta^{11}} \\
\hline
\frac{1}{\beta^{7}-\beta^{8}} & \frac{1}{\beta^{7}-\beta^{9}} & \frac{1}{\beta^{7}-\beta^{10}} & \frac{1}{\beta^{7}-\beta^{11}}
\end{array}\right]=\left[\begin{array}{ccc|c}
\beta^{5} &\beta^{12} & \beta^{7} & \beta^{9}\\
1 &\beta^{4} & \beta^{11} & \beta^{6}\\
\beta^{2} &\beta^{14} & \beta^{3} & \beta^{10}\\
\hline
\beta^{4} & 1 & \beta^{9} & \beta^{7} 
\end{array}
\right].
\end{equation}
\hrulefill
\setcounter{equation}{3}
\end{figure*}

\begin{exam} \label{exam: CodeDL} Let $q=2^4$, $p=2$, $r=r_1=r_2=3$, $\delta=\delta_1=\delta_2=1$, $k=k_1=k_2=3$, $n=n_1=n_2=k+r=6$, $\delta'=\delta'_1=\delta_2'=\delta_1+\delta_2-\delta=1$. Then, $d_1=r-\delta'=3-1=2$, $d_2=r+\delta'=3+1=4$. Choose a primitive polynomial of degree $4$ over $\textup{GF}(2)$: $g(X)=X^4+X+1$. Let $\beta$ be a root of $g(X)$. Then, $\beta$ is a primitive element of $\textup{GF}(2^4)$. The binary representation of all the symbols in $\textup{GF}(2^4)$ is provided in \Cref{table: GF}. 

Let $\bold{A}_{1,1}=\bold{A}_{2,2}$, $\bold{B}_{1,2}=\bold{B}_{2,1}$, $\bold{U}_1=\bold{U}_2$, and $\bold{T}_1=\bold{T}_2$ as specified in (\ref{eqn: exam1}).
Therefore, 
\begin{equation*}\bold{A}_{1,2}=\bold{A}_{2,1}=\bold{B}_{2,1}\bold{U}_1=\left[\begin{array}{ccc}
\beta^{13} &\beta^{9} & \beta^{3}\\
\beta^{10} &\beta^{6} & 1\\
\beta^{14} &\beta^{10} & \beta^{4}
\end{array}
\right].
\end{equation*}
Then, the generator matrix $\bold{G}$ is specified as follows,
\begin{equation*}\small
\left[
\begin{array}{ccc|ccc|ccc|ccc}
1 & 0 & 0 & \beta^{5} &\beta^{12} & \beta^{7} & 0 & 0 & 0 & \beta^{13} &\beta^{9} & \beta^{3}\\
0 & 1 & 0 & 1 &\beta^{4} & \beta^{11} & 0 & 0 & 0 & \beta^{10} &\beta^{6} & 1 \\
0 & 0 & 1 & \beta^{2} &\beta^{14} & \beta^{3} & 0 & 0 & 0 & \beta^{14} &\beta^{10} & \beta^{4} \\
\hline
0 & 0 & 0 & \beta^{13} &\beta^{9} & \beta^{3} & 1 & 0 & 0 & \beta^{5} &\beta^{12} & \beta^{7} \\
0 & 0 & 0 &\beta^{10} &\beta^{6} & 1 & 0 & 1 & 0 & 1 &\beta^{4} & \beta^{11} \\
0 & 0 & 0 & \beta^{14} &\beta^{10} & \beta^{4} & 0 & 0 & 1 & \beta^{2} &\beta^{14} & \beta^{3} \\
\end{array}\right]=\left[\begin{array}{cccc}
\bold{I}_{k_1}&\bold{A}_{1,1}&\bold{0}_{k_1\times k_2} &\bold{A}_{1,2}\\
\bold{0}_{k_2\times k_1}&\bold{A}_{2,1}&\bold{I}_{k_2}&\bold{A}_{2,2}
\end{array}\right].
\end{equation*}

Suppose $\bold{m}_1=(1,\beta,\beta^2)$ and $\bold{m}_2=(\beta,1,0)$. Thus, $\bold{c}_1=(1,\beta,\beta^2,\beta^{14},0,0)$ and $\bold{c}_2=(\beta,1,0,\beta^{6},0,\beta^{13})$. Moreover, $\bold{H}_1^{\Tx{L}}$ and $\bold{H}_1^{\Tx{G}}$ are specified as follows,
\begin{equation*}\small
\bold{H}_1^{\Tx{G}}=\left[\hspace{-0.1cm}\begin{array}{cccc}
\beta^5 & \beta^{12} & \beta^7 & \beta^9\\
1 & \beta^{4} & \beta^{11} & \beta^6 \\
\beta^2 & \beta^{14} & \beta^3 &\beta^{10}\\
1 & 0 & 0 & 0 \\
0 & 1 & 0 & 0 \\
0 & 0 & 1 & 0
\end{array}\hspace{-0.1cm}\right]^{\Tx{T}}, \textup{ } \bold{H}_1^{\Tx{L}}=\left[\hspace{-0.1cm}\begin{array}{cccc}
\beta^5 & \beta^{12} & \beta^7 \\
1 & \beta^{4} & \beta^{11}\\
\beta^2 & \beta^{14} & \beta^3\\
1 & 0 & 0 \\
0 & 1 & 0 \\
0 & 0 & 1 \\
\beta^4 & 1 & \beta^9
\end{array}\hspace{-0.1cm}\right]^{\Tx{T}}.
\end{equation*}
According to \Cref{cons: CRScons}, $\bold{G}$ is the generator matrix of a double-level accessible code that corrects $2$ local erasures by local access, and corrects $2$ extra erasures within a single cloud by global access. In the following, we denote the version of $\bold{c}_1$ having erasures by $\bold{c}'_1$, and erased symbols by $e_i$, $i\in\Db{N}$.

As an example of decoding by local access, suppose $\bold{c}'_1=(1,e_1,\beta^2,e_2,0,0)$. Then, the erased elements of $\tilde{\bold{c}}_1=(1,e_1,\beta^2,e_2,0,0,e_3)$ can be retrieved using $\bold{H}_1^{\Tx{L}}$ as its parity-check matrix. In particular, we solve $\bold{H}^{\Tx{L}}_1\tilde{\bold{c}}_1^{\Tx{T}}=(0,0,0)^{\Tx{T}}$ for $e_1,e_2,e_3$ and obtain $(e_1,e_2,e_3)=(\beta,\beta^{14},\beta^7)$. Therefore, we have decoded $\bold{c}_1$ successfully.

As an example of decoding by global access, suppose $\bold{c}'_1=(e_1,e_2,\beta^2,e_3,e_4,0)$, and suppose $\bold{c}_2$ has been decoded successfully as $\bold{c}_2=(\beta,1,0,\beta^{6},0,\beta^{13})$, which implies that $\bold{m}_1\bold{B}_{1,2}\bold{U}_2=(\beta^{6},0,\beta^{13})-\beta\cdot(\beta^5,\beta^{12},\beta^7)-1\cdot(1,\beta^4,\beta^{11})=(1,\beta^{11},\beta^{5})$. Since $\bold{U}_2=(\beta^{4},1,\beta^9)$, we obtain $\bold{m}_1\bold{B}_{1,2}=\beta^{11}$. Moreover, we compute $\bold{m}_2\bold{B}_{2,1}\bold{U}_1=(\beta^{11},\beta^{7},\beta)$. Let $\bar{\bold{c}}_1=\bold{c}'_1-(0,0,0,\beta^{11},\beta^{7},\beta)=(e'_1,e'_2,\beta^2,e'_3,e'_4,\beta)$. Then, we solve $\bold{H}^{\Tx{G}}_1\bar{\bold{c}}_1^{\Tx{T}}=(0,0,0,\beta^{11})^{\Tx{T}}$ and obtain $(e'_1,e'_2,e'_3,e'_4)=(1,\beta,\beta^{10},\beta^7)$. Therefore, $e_1=e'_1=1$, $e_2=e'_2=\beta$, $e_3=e'_3+\beta^{11}=\beta^{14}$, $e_4=e'_4+\beta^7=0$, and we have decoded $\bold{c}_1$ successfully.
\end{exam}

Note that \Cref{cons: CRScons} is proposed based on the assumption that any node is able to communicate with all the nodes in the network; namely, the underlying DSN has a specific topology embodied in a complete graph. However, as discussed in \Cref{sectoin: introduction} and shown in Fig.~\ref{fig: figmodel}, practical DSNs are not necessarily constrained into any specific structures. A major reason is that nodes are typically scattered in geographically separated locations and communicate with only a few nodes nearby. Even if connections of nodes are not determined by physical locations, their logical connections can still be of any topology tailored for particular requirements from users of the services those nodes provide. Therefore, it is important to generalize our previous construction into one that fits into arbitrary topologies. In the next section, we take network topology into account and focus on constructions that are topology-adaptive. 

\section{Cooperative Data Protection}
\label{section coop}
In this section, we first mathematically describe the EC hierarchy and its depth associated with the given DSN. EC hierarchy specifies the EC capabilities of nodes while cooperating with different sets of other nodes. We then propose a cooperation scheme where each node only cooperates with its single-hop neighbors.

\begin{figure}
\centering
\includegraphics[width=0.5\textwidth]{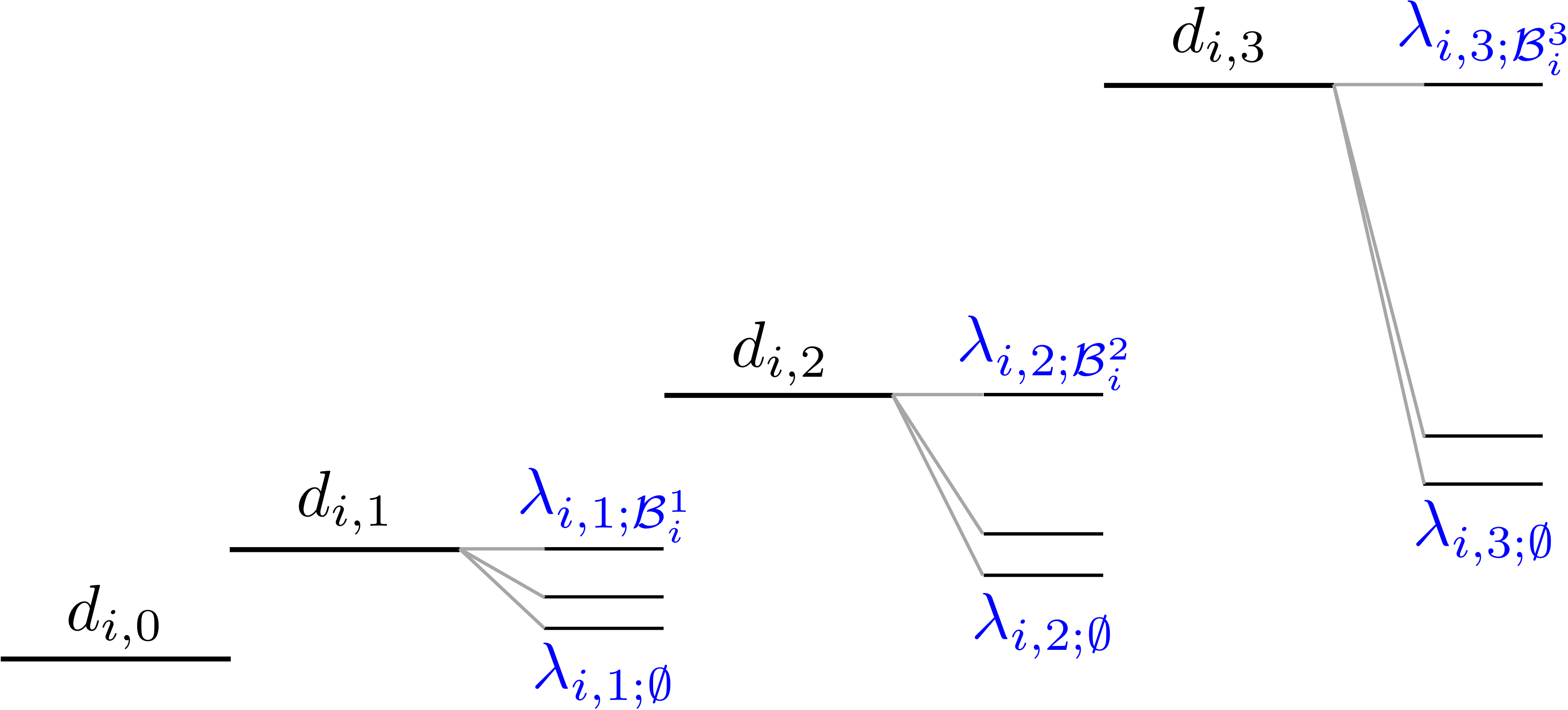}
\caption{EC hierarchy of node $v_i\in V$. The values $d_{i,0}$ and $d_{i,\ell}$, $1\leq \ell\leq 3$, represent the maximum number of erasures $v_i$ can tolerate in the local decoding, and decoding with the assistance of the $1$-st to the $\ell$-th level cooperations of the codeword $\bold{c}_i$, i.e., cooperating with nodes in $\Se{A}_i^{\ell}$, respectively. However, even for a fixed $\Se{A}_i^{\ell}$, different sets $\Se{W}$ of nodes that are recovered in $\Se{B}_i^{\ell}$ may also result in different EC capabilities; we refer to them as $(\lambda_{i,\ell;\Se{W}})$, $\varnothing\subseteq\Se{W}\subseteq\Se{B}_i^{\ell}$.}
\label{fig: EC Hierarchy}
\end{figure}

\begin{figure*}[!t]
\normalsize
\centering

\setcounter{equation}{3}
\setcounter{MaxMatrixCols}{12}
\begin{equation}\small
\scalebox{.7}{$
\begin{matrix}
\textcolor{robert}{1}&\textcolor{robert}{2}&\textcolor{robert}{3}&\textcolor{robert}{4}&\textcolor{robert}{5}&\textcolor{robert}{6}&\textcolor{robert}{7}&\textcolor{robert}{8}&\textcolor{robert}{9}&\textcolor{robert}{10}&\textcolor{robert}{11}&\textcolor{robert}{12}\\
\hline
\hline
\bold{A}_{1,1}&\bold{B}_{1,2}\bold{U}_2&\bold{0}&\bold{0}&\bold{0}&\bold{0}&\bold{0}&\bold{0}&\bold{0}&\bold{0}&\bold{0}&\bold{0}\\
\hline
\bold{B}_{2,1}\bold{U}_1&\bold{A}_{2,2}&\bold{B}_{2,3}\bold{U}_3&\bold{0}&\bold{B}_{2,5}\bold{U}_5&\bold{0}&\bold{0}&\bold{0}&\bold{0}&\bold{0}&\bold{0}&\bold{0}\\
\hline
\bold{0}&\bold{B}_{3,2}\bold{U}_2&\bold{A}_{3,3}&\bold{B}_{3,4}\bold{U}_4&\bold{0}&\bold{0}&\bold{0}&\bold{0}&\bold{0}&\bold{0}&\bold{0}&\bold{0}\\
\hline
\bold{0}&\bold{0}&\bold{B}_{4,3}\bold{U}_3&\bold{A}_{4,4}&\bold{B}_{4,5}\bold{U}_5&\bold{B}_{4,6}\bold{U}_6&\bold{0}&\bold{0}&\bold{0}&\bold{0}&\bold{0}&\bold{0}\\
\hline
\bold{0}&\bold{B}_{5,2}\bold{U}_2&\bold{0}&\bold{B}_{5,4}\bold{U}_4&\bold{A}_{5,5}&\bold{B}_{5,6}\bold{U}_6&\bold{0}&\bold{B}_{5,8}\bold{U}_8&\bold{0}&\bold{0}&\bold{0}&\bold{0}\\
\hline
\bold{0}&\bold{0}&\bold{0}&\bold{B}_{6,4}\bold{U}_4&\bold{B}_{6,5}\bold{U}_5&\bold{A}_{6,6}&\bold{B}_{6,7}\bold{U}_7&\bold{0}&\bold{0}&\bold{0}&\bold{0}&\bold{0}\\
\hline
\bold{0}&\bold{0}&\bold{0}&\bold{0}&\bold{0}&\bold{B}_{7,6}\bold{U}_6&\bold{A}_{7,7}&\bold{B}_{7,8}\bold{U}_8&\bold{B}_{7,9}\bold{U}_9&\bold{0}&\bold{B}_{7,11}\bold{U}_{11}&\bold{0}\\
\hline
\bold{0}&\bold{0}&\bold{0}&\bold{0}&\bold{B}_{8,5}\bold{U}_5&\bold{0}&\bold{B}_{8,7}\bold{U}_7&\bold{A}_{8,8}&\bold{B}_{8,9}\bold{U}_9&\bold{0}&\bold{0}&\bold{0}\\
\hline
\bold{0}&\bold{0}&\bold{0}&\bold{0}&\bold{0}&\bold{0}&\bold{B}_{9,7}\bold{U}_7&\bold{B}_{9,8}\bold{U}_8&\bold{A}_{9,9}&\bold{B}_{9,10}\bold{U}_{10}&\bold{0}&\bold{0}\\
\hline
\bold{0}&\bold{0}&\bold{0}&\bold{0}&\bold{0}&\bold{0}&\bold{0}&\bold{0}&\bold{B}_{10,9}\bold{U}_9&\bold{A}_{10,10}&\bold{B}_{10,11}\bold{U}_{11}&\bold{B}_{10,12}\bold{U}_{12}\\
\hline
\bold{0}&\bold{0}&\bold{0}&\bold{0}&\bold{0}&\bold{0}&\bold{B}_{11,7}\bold{U}_7&\bold{0}&\bold{0}&\bold{B}_{11,10}\bold{U}_{10}&\bold{A}_{11,11}&\bold{B}_{11,12}\bold{U}_{12}\\
\hline
\bold{0}&\bold{0}&\bold{0}&\bold{0}&\bold{0}&\bold{0}&\bold{0}&\bold{0}&\bold{0}&\bold{B}_{12,10}\bold{U}_{10}&\bold{B}_{12,11}\bold{U}_{11}&\bold{A}_{12,12}\\
\hline
\end{matrix}$}
\label{fig: genexample1}
\end{equation}
\hrulefill
\setcounter{equation}{4}
\end{figure*}

\subsection{EC Hierarchy}
\label{subsec EC hierarchy}

Denote the \textbf{EC hierarchy} of node $v_i\in V$ by a sequence $\bold{d}_i=(d_{i,0},d_{i,1},\dots,d_{i,L_i})$, where $L_i$ is called the \textbf{depth} of $\bold{d}_i$, and $d_{i,\ell}$ represents the maximum number of erased symbols $v_i$ can recover in its codeword $\bold{c}_i$ from the $\ell$-th level cooperation, for all ${\ell}\in \MB{L_i}$. The maximum number of erased symbols $v_i$ can recover in $\bold{c}_i$ locally, i.e., without communicating with neighboring nodes, is $d_{i,0}$.

For each $v_i\in V$ such that $L_i>0$, there exist two series of sets of nodes, denoted by $\varnothing\subset\Se{A}_i^1\subset\Se{A}_i^2\subset \dots \subset \Se{A}_i^{L_i}\subseteq V$ and $\{\Se{B}_i^{\ell}\}_{\ell=1}^{L_i}$, where $\Se{A}_i^{\ell}\cap \Se{B}_i^{\ell}=\varnothing$ for all ${\ell}\in \MB{L_i}$, and a series $\left(\lambda_{i,{\ell};\Se{W}}\right)_{\varnothing\subseteq\Se{W}\subseteq\Se{B}_i^{\ell}}$. In the ${\ell}$-th level cooperation, node $v_i\in V$ tolerates $\lambda_{i,{\ell};\Se{W}}$ ($\varnothing\subseteq\Se{W}\subseteq\Se{B}_i^{\ell}$) erasures if all nodes in $\Se{A}_i^{\ell}\cup\Se{W}$ are able to decode their own messages, where the maximum value is $\lambda_{i,{\ell};\Se{B}_i^{\ell}} = d_{i,{\ell}}$ and is reached when $\Se{W}=\Se{B}_i^{\ell}$; the minimum value is $\lambda_{i,{\ell};\varnothing}$ and is reached when $\Se{W}=\varnothing$. See Fig.~\ref{fig: EC Hierarchy} for illustration.

We first take a look at the cooperation schemes with the EC hierarchy of depth $1$. For the EC hierarchy of depth $1$, $\Se{A}_i^1$ is always a subset of the neighbors of $v_i$, while $\Se{B}_i^1$ is the set of all nodes in $\Se{A}_j^1$, for all $j$ such that $v_j$ is in $\Se{A}_i^1$, except the ones in $\{v_i\} \cup \Se{A}_i^1$.

\subsection{Single-Level Cooperation}
\label{subsec single-level cooperation}

We now discuss the case where each node only has cooperation of depth $1$. Consider a DSN represented by $G(V,E)$ that is associated with parameters $(\bold{n},\bold{k},\bold{r})$ and a class of sets $\{\Se{M}_i\}_{v_i\in V}$ such that $\varnothing\subset\Se{M}_i\subseteq \Se{N}_i$, for all $v_i\in V$. In \Cref{cons: 1}, we present a joint coding scheme where node $v_i$ only cooperates with nodes in $\Se{M}_i$, for all $v_i\in V$. Heterogeneity is obviously achieved since $n_i$, $k_i$, $r_i$, are not required to be identical for all $v_i\in V$. 

Our previous result in \cite{Yang2019HC} represents a special case of \Cref{cons: 1}, where the motivating application was in centralized cloud storage.
\Cref{cons: 1} extends that work to deal with arbitrary decentralized topologies, in contrast to the tree-like topology prevalent in. centralized networks. \Cref{exam: exam1} and \Cref{exam: speed} illustrate the efficacy of the proposed construction in decentralized storage.

\begin{cons} \label{cons: 1} Let $G(V,E)$ represent a DSN associated with parameters $(\bold{n},\bold{k},\bold{r})$ and a local EC parameter $\boldsymbol{\delta}$, where $\bold{r}\succ \boldsymbol{\delta}\succeq \bold{0}$. Let $p=|V|$ and $\textup{GF}(q)$ be a Galois field of size $q$, where $q>\max\limits_{v_i\in V} \left(n_i+\delta_i+\sum\nolimits_{v_j\in \Se{M}_i} \delta_j\right)$.   

For each $i\in \MB{p}$, let $a_{i,x}$, $x\in\MB{k_i+\delta_i}$, and $b_{i,y}$, $y\in\MB{r_i+\sum\nolimits_{v_j\in \Se{M}_i} \delta_j}$, be distinct elements of $\textup{GF}(q)$. Consider the Cauchy matrix $\bold{T}_i\in \textup{GF}(q)^{(k_i+\delta_i)\times (r_i+\sum\nolimits_{v_j\in \Se{M}_i} \delta_j)}$ such that $\bold{T}_i=\bold{Y}(a_{i,1}, \dots, a_{i,k_i+\delta_i};b_{i,1},\dots,b_{i,r_i+\sum\nolimits_{v_j\in \Se{M}_i} \delta_j})$. Matrix $\bold{G}$ in (\ref{eqn: GenMatDL}) is assembled as follows. For each $i\in\MB{p}$, we obtain $\{\bold{B}_{i,j}\}_{v_j\in \Se{M}_i}$, $\bold{U}_i$, $\bold{A}_{i,i}$, according to the following partition of $\bold{T}_i$:

\begin{equation}\label{eqn: CRS}
\bold{T}_i=\left[
\begin{array}{c|c}
\bold{A}_{i,i} & \begin{array}{c|c|c}
\bold{B}_{i,j_1} & \dots & \bold{B}_{i,j_{| \Se{M}_i|}}
\end{array}
\\
\hline
\bold{U}_i & \bold{Z}_{i}
\end{array}\right],
\end{equation}
where $\Se{M}_i=\{v_{j_1},v_{j_2},\dots,v_{j_{| \Se{M}_i|}}\}$, $\bold{A}_{i,i}\in \textup{GF}(q)^{k_i\times r_i}$,$\bold{U}_i\in \textup{GF}(q)^{\delta_i\times r_i}$, $\bold{B}_{i,j}\in \textup{GF}(q)^{k_i\times \delta_j}$, for $v_i\in V$ and $v_j\in \Se{M}_i$. Let $\bold{A}_{i,j}=\bold{B}_{i,j}\bold{U}_j$ if $v_j\in \Se{M}_i$, otherwise let it be a zero matrix.

Denote the code with generator matrix $\bold{G}$ by $\Se{C}_1$.
\end{cons}

\begin{theo} \label{theo: ECcons1} In a DSN with $\Se{C}_1$, $\bold{d}_i=(r_i-\delta_i,r_i+\sum\nolimits_{v_j\in \Se{M}_i}\delta_j)$, $\Se{A}_i^1=\Se{M}_i$, and $\Se{B}_i^1=\bigcup\nolimits_{v_j\in\Se{M}_i}\left(\Se{M}_j\setminus(\{v_i\}\cup \Se{M}_i)\right)$, for all $v_i\in V$. Furthermore, the EC hierarchy associated with $d_{i,1}$ is $(\lambda_{i,1;\Se{W}})_{\varnothing\subseteq\Se{W}\subseteq \Se{B}_i^1}$, where\\ $\lambda_{i,1;\Se{W}}= r_i+ \sum\nolimits_{j:v_j\in \Se{M}_i,(\Se{M}_j\setminus\{v_i\})\subseteq (\Se{M}_i\cup\Se{W})}\delta_j$.
\end{theo}

\begin{proof} It follows directly from \Cref{lemma: DLcodedis} that for all $i\in\MB{p}$, the $i$-th entry of the EC hierarchy at node $v_i$ is $\bold{d}_i=(r_i-\delta_i,r_i+\sum\nolimits_{v_j\in \Se{M}_i}\delta_j)$. The remaining task is to prove that $\lambda_{i,1;\Se{W}}=r_i+\sum\nolimits_{j:v_j\in \Se{M}_i,(\Se{M}_j\setminus\{v_i\})\subseteq (\Se{M}_i\cup\Se{W})}\delta_j$.


For all $v_i\in V$, let $\bold{s}_{i}=\sum\nolimits_{v_{j}\in\Se{M}_i}\bold{m}_{j}\bold{B}_{j,i}$. We first notice that for any $v_i\in V$, $v_j\in \Se{M}_i$, if $\bold{m}_j$ is recoverable, then the additional cross parities $\bold{s}_j\bold{U}_j$ and the original cross parities $\bold{m}_j\bold{A}_{j,j}$ of $v_j$ can be computed. Therefore, $\bold{s}_j$ can be computed, and if all the messages $\{\bold{m}_{j'}\}_{v_{j'}\in\Se{M}_j\setminus\{v_i\}}$ are further recoverable, then the cross parities $\bold{m}_i\bold{B}_{i,j}$ of $v_i$ from $v_j$ can be computed from $\bold{m}_i\bold{B}_{i,j}=\bold{s}_{j}-\sum\nolimits_{v_{j'}\in\Se{M}_j\setminus\{v_i\}}\bold{m}_{j'}\bold{B}_{j',j} $.  

Previous discussion implies that for any $\Se{W}$, $\varnothing\subseteq\Se{W}\subseteq \Se{B}_i^1$, if $(\Se{M}_j\setminus\{v_i\})\subseteq (\Se{M}_i\cup\Se{W})$, then the additional $\delta_j$ cross parities $\bold{m}_i\bold{B}_{i,j}$ of $\bold{m}_i$ can be obtained. Therefore, $\lambda_{i,1;\Se{W}}=r_i+\sum\nolimits_{j:v_j\in \Se{M}_i,(\Se{M}_j\setminus\{v_i\})\subseteq (\Se{M}_i\cup\Se{W})}\delta_j$. 
\end{proof}

\begin{figure}
\centering
\includegraphics[width=0.5\textwidth]{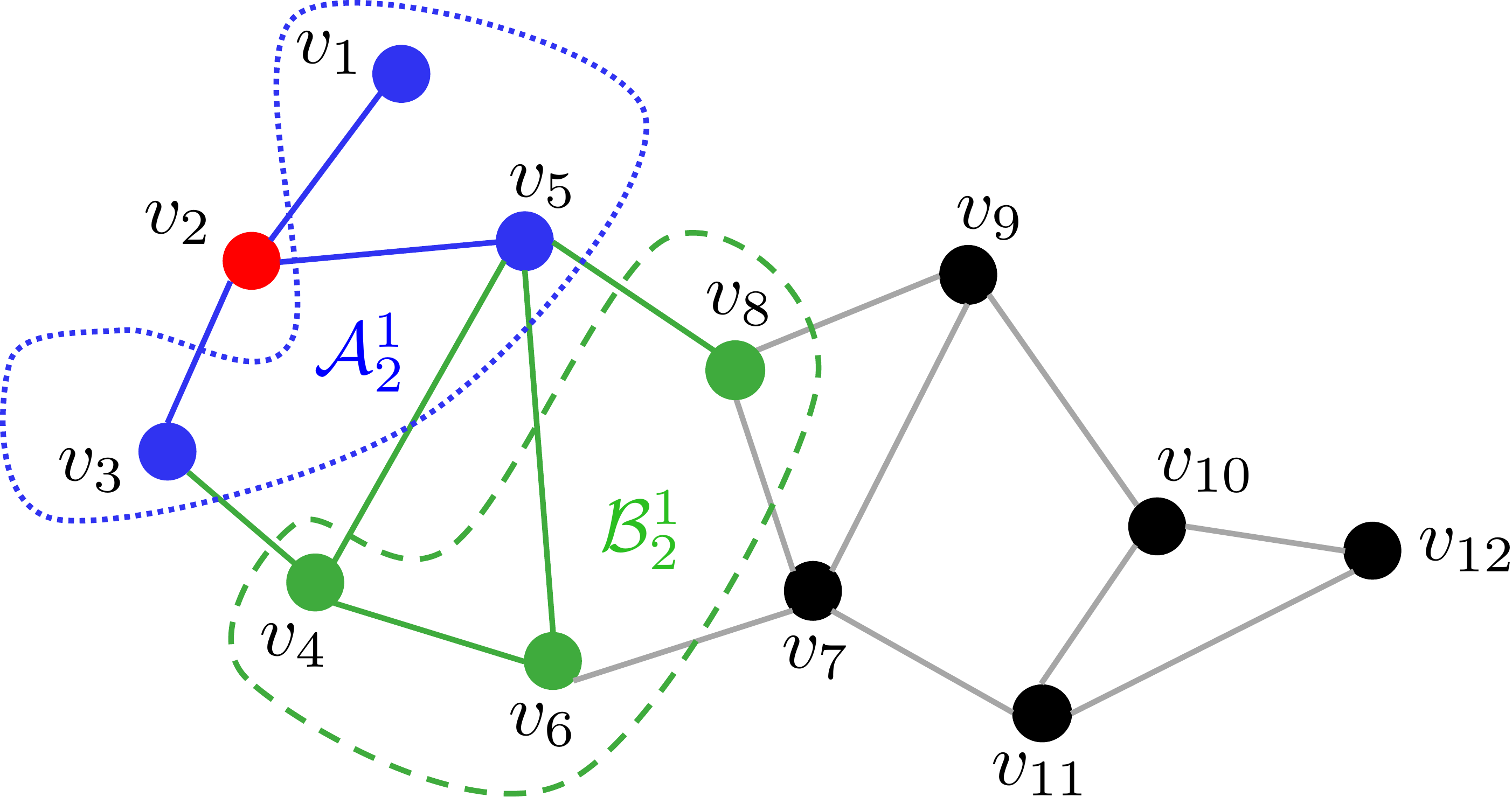}
\caption{DSN for \Cref{exam: exam1}. Nodes in $\Se{A}_2^{1}$ are neighbors of $v_2$ and are required to be locally-recoverable to remove cross parities from the parity part of $\bold{c}_2$. Nodes in $\Se{B}_2^{1}$ are neighbors of nodes in $\Se{A}_2^{1}$ except for $v_2$ and nodes in $\Se{A}_2^{1}$ themselves. For any node $v_j$ in $\Se{A}_2^{1}$, only when all the neighbors of $v_j$ in $\Se{B}_2^{1}$ are recovered, $v_j$ can provide extra parity symbols to $v_2$.}
\label{fig: example1}
\end{figure}

\begin{exam} \label{exam: exam1} Consider the DSN shown in Fig.~\ref{fig: example1}. Let $\Se{M}_i=\Se{N}_i$ in \Cref{cons: 1}, for all $i\in \MB{12}$. The matrix in (\ref{fig: genexample1}) is obtained by removing all the block columns of identity surrounded by zero matrices from the generator matrix (\ref{eqn: GenMatDL}) of $\Se{C}_1$, and is referred to as the \textbf{\emph{non-systematic component}} of the generator matrix. 

Take node $v_2$ as an example. Observe that $\Se{A}_2^1=\Se{M}_2=\{v_1,v_3,v_5\}$, $\Se{A}_1^1=\Se{M}_1=\{v_2\}$, $\Se{A}_3^1=\Se{M}_3=\{v_2,v_4\}$, and $\Se{A}_5^1=\Se{M}_5=\{v_2,v_4,v_6,v_8\}$. Therefore, $\Se{B}_2^1=\bigcup\nolimits_{j\in\{1,3,5\}} \Se{M}_j\setminus\{v_1,v_2,v_3,v_5\}=\{v_4,v_6,v_8\}$. Moreover, $\bold{d}_{2}=(r_2-\delta_2,r_2+\sum\nolimits_{j\in\{1,3,5\}}\delta_j)$, $\lambda_{2,1;\varnothing}=\lambda_{2,1;\{v_6\}}=\lambda_{2,1;\{v_8\}}=\lambda_{2,1;\{v_6,v_8\}}=r_2+\delta_1$, $\lambda_{2,1;\{v_4\}}=\lambda_{2,1;\{v_4,v_6\}}=\lambda_{2,1;\{v_4,v_8\}}=r_2+\delta_1+\delta_3$, and $\lambda_{2,1;\{v_4,v_6,v_8\}}=r_2+\delta_1+\delta_3+\delta_5$.

Consider the case where the $1$-st level cooperation of $v_2$ is initiated, i.e., the number of erasures lies within the interval $\left[r_2-\delta_2+1,r_2+\delta_1+\delta_3+\delta_5\right]$. Then, if $\bold{m}_1,\bold{m}_3,\bold{m}_5$ are all locally-recoverable, the cross parities $\bold{m}_1\bold{B}_{1,2}$, $\bold{m}_3\bold{B}_{3,2}$, $\bold{m}_5\bold{B}_{5,2}$ computed from the non-diagonal parts in the generator matrix can be subtracted from the parity part of $\bold{c}_2$ to get $\bold{m}_2\bold{A}_{2,2}$. Moreover, the successful decoding of $\bold{m}_1$ makes $\bold{m}_{2}\bold{B}_{2,1}$ known to $v_2$. This process provides $(r_2+\delta_1)$ parities for $\bold{m}_2$, and thus allows $v_2$ to tolerate $(r_2+\delta_1)$ erasures.

In order to correct more than $(r_2+\delta_1)$ erasures, we need extra cross parities generated from $\bold{B}_{2,3}\bold{U}_3$ and $\bold{B}_{2,5}\bold{U}_5$. However, local decoding only allows $v_3$, $v_5$ to know $\bold{m}_{2}\bold{B}_{2,3}+\bold{m}_{4}\bold{B}_{4,3}$ and $\bold{m}_{2}\bold{B}_{2,5}+\bold{m}_{4}\bold{B}_{4,5}+\bold{m}_{6}\bold{B}_{6,5}+\bold{m}_{8}\bold{B}_{8,5}$, respectively. Therefore, $v_3$ needs $\bold{m}_4$ to be recoverable to obtain the extra $\delta_3$ cross parities, and $v_5$ needs $\bold{m}_4$, $\bold{m}_6$, $\bold{m}_8$ to be recoverable to obtain the extra $\delta_5$ cross parities. 
\end{exam}

As shown in \Cref{exam: exam1}, instead of presenting a rigid~EC capability, our proposed scheme enables nodes to have correction of a growing number of erasures with bigger sets of neighboring nodes recovering their messages. Therefore, nodes automatically choose the shortest path to recover their messages, thus significantly increasing the average recovery speed, especially when the erasures are distributed non-uniformly and sparsely, which is important for blockchain-based DSNs \cite{zhu2019blockchain,underwood2016blockchain}. Moreover, nodes with higher reliabilities are utilized to help decode the data of less reliable nodes, enabling correction of erasure patterns that are not recoverable in our previous work in \cite{Yang2019HC}. We show these properties in \Cref{exam: speed} and \Cref{exam: erasure pattern}.

\begin{figure}
\centering
\includegraphics[width=0.5\textwidth]{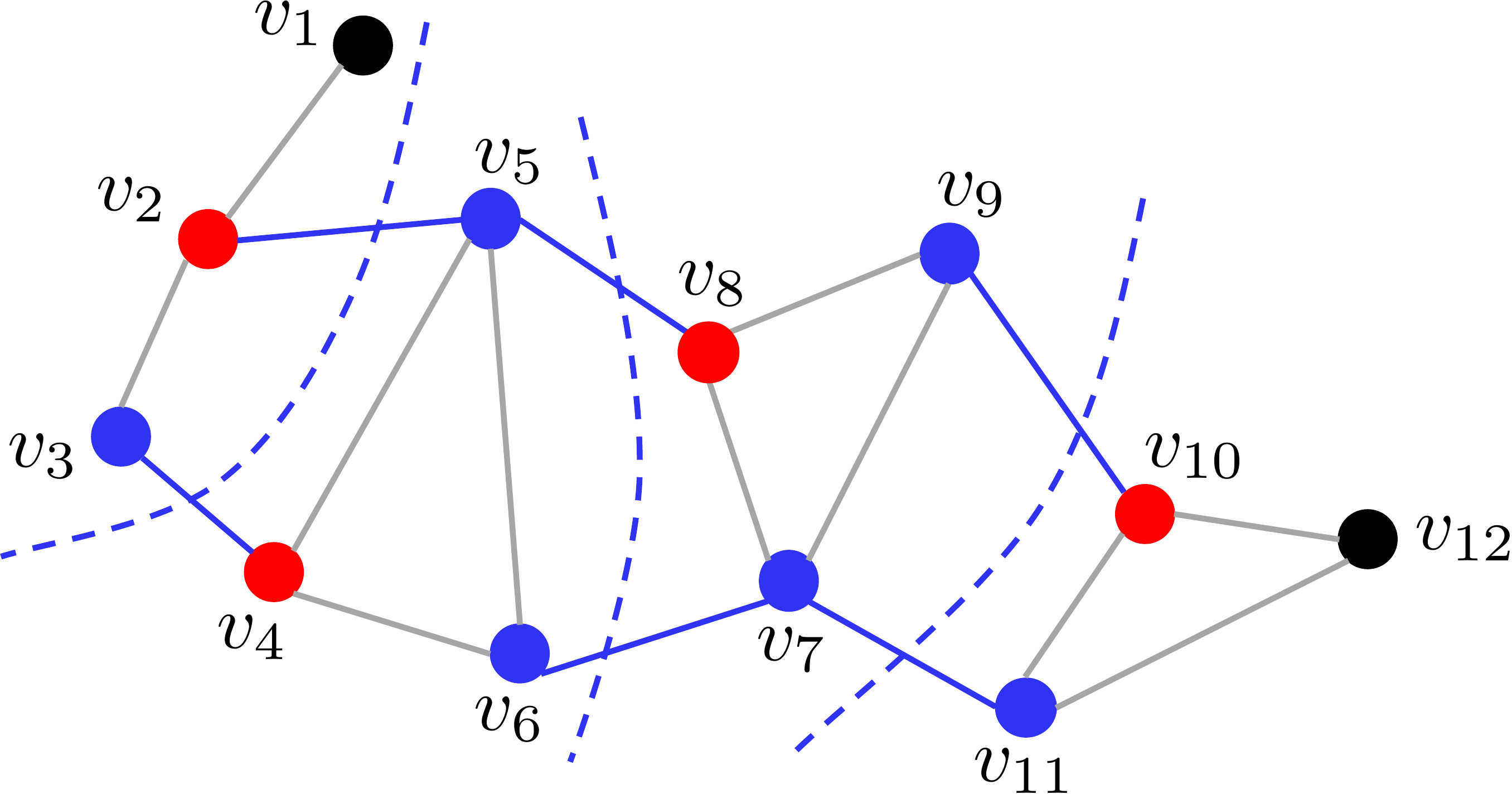}
\caption{The erasure pattern in \Cref{exam: erasure pattern}. Red and non-red nodes refer to nodes where the codewords stored at them are non-locally-recoverable and locally-recoverable, respectively.}
\label{fig: compare}
\end{figure}

\begin{exam} \label{exam: speed} \emph{\textbf{(Faster Recovery Speed)}} Consider a DSN with the cooperation scheme specified in \Cref{exam: exam1}. Suppose the time to be consumed on transferring information through the communication link $e_{i,j}$ is $t_{i,j}\in\Db{R}^{+}$, where $t_{i,j}=t_{j,i}$ for all $i,j\in\MB{12}$, $i\neq j$, and $\max\{t_{1,2},t_{2,5}\}<(t_{2,3}+t_{3,4})<t_{2,5}+\min\{t_{4,5},t_{5,6},t_{5,8}\}$.

Consider the case where $\bold{c}_2$ at node $v_2$ has $(r_2+1)$ erasures, which implies that in addition to the case of $\bold{m}_1$, $\bold{m}_3$, $\bold{m}_5$ being obtained locally, recovering $\bold{m}_4$ is sufficient for $v_2$ to successfully obtain its message. The time consumed for decoding is $(t_{2,3}+t_{3,4})$. Therefore, any system using network coding with the property that a node failure is recovered through accessing more than $4$ other nodes will need longer processing time for this case. 
\end{exam}

\begin{exam} \label{exam: erasure pattern} \emph{\textbf{(Flexible Erasure Patterns)}} Consider the DSN with the cooperation scheme specified in \Cref{exam: exam1}. Suppose $\{\bold{m}_i\}_{i\notin\{2,4,8,10\}}$ are all locally-recoverable. Then, consider the case where $\bold{m}_i$ has $(r_i+1)$ erasures for $i\in\{2,4,8,10\}$, which exemplifies a correctable erasure pattern for our proposed codes. 

The hierarchical coding scheme presented in \cite{Yang2019HC} can recover from this erasure pattern only if the code used adopts a partition of all nodes into $4$ disjoint groups, each of which contains exactly a node from $\{v_2,v_4,v_8,v_{10}\}$, as shown in Fig.~\ref{fig: compare}. Moreover, the partition of the code in \cite{Yang2019HC} results in a reduction of the EC capability of the $1$-st level cooperation at every node except for $v_1,v_{12}$ because the additional information \Cref{cons: 1} allows to flow through the edges marked in blue no longer exists.
\end{exam}

\subsection{Recoverable Erasure Patterns}
\label{subsec recoverable erasure patterns 1}
For a code specified for a DSN according to \Cref{cons: 1}, we next investigate the recoverable erasure patterns under the proposed EC solution. Throughout this paper, for any edge $(i,j)$ from $v_i$ to $v_j$ in a directed graph $G(V,E)$, we call $v_j$ a child of $v_i$, and $v_i$ a parent of $v_j$. 

In the DSN depicted in Fig.~\ref{decgraph}, suppose all codewords stored at black nodes are locally-recoverable; those stored at green nodes, e.g., $v_i$ with $i \in \{6,8,12\}$, are recoverable by accessing their neighboring nodes in $\Se{A}_i^1$ only; and those in blue nodes, e.g., $v_i$ with $i \in \{0,2,3,5,10\}$, need some nodes in $\Se{B}_i^1$ to be also recoverable since they need to obtain extra cross parities from at least one of their neighboring nodes in $\Se{A}_i^1$. As an example, assume that node $v_0$ needs to obtain extra parities from only one of its neighbors, say $v_1$. This condition requires codewords stored at $v_2,v_3,v_{13},v_{14},v_{15}$ all being recoverable. Since codewords in $v_{13},v_{14},v_{15}$ are already locally-recoverable, this \textcolor{lara}{case} essentially requires codewords in $v_2,v_3$ to be recovered. For simplicity, we just refer to this requirement as ``$v_0$ needs $v_2,v_3$''. Similarly, $v_2$ needs $v_5,v_6$, $v_3$ needs $v_6,v_{10}$, $v_5$ needs $v_6,v_8$, and $v_{10}$ needs $v_6,v_{12}$. Given that codewords in $v_6,v_8,v_{12}$ are recoverable, all the blue nodes are recoverable following the order $v_5$, $v_{10}$, $v_2$, $v_3$, $v_0$.

Note that in this paper, we suppose that in the protocol carrying out the decoding algorithm, each node, when receiving a request, either replies back with the required message, provided that the information gathered at this node suffices to provide the answer, or broadcasts a request to all its neighbors to ask for the information it needs. For now, we assume that nodes remain intact during decoding.

We next define the so-called \textbf{decoding graph} of each node where the codeword stored there is recoverable in a DSN. For a given node, this graph describes the aforementioned order of decoding non-locally-recoverable nodes involved in the process of decoding this particular node. Observe that connections between any two nodes having their codewords locally-recoverable are omitted for simplicity.

\begin{figure}
\centering
\includegraphics[width=0.4\textwidth]{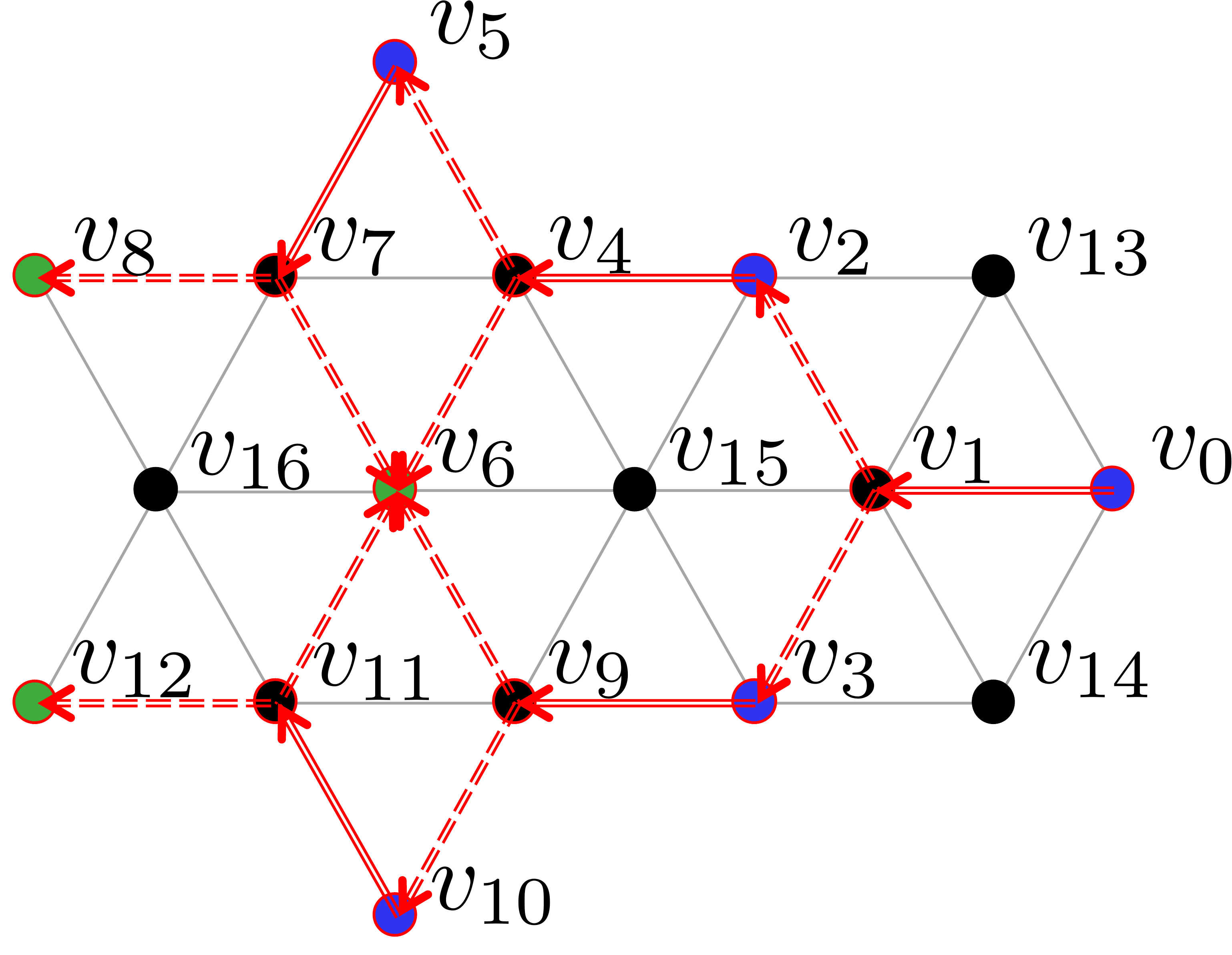}
\caption{Decoding graph in \Cref{defi decoding graph}. Black and non-black nodes refer to locally-recoverable and non-locally-recoverable nodes, respectively, where green nodes are recoverable without requiring any node that is not in their neighborhood to be recovered. The subgraph marked in red is a decoding graph at root node $v_0$. When there exists an edge pointing from $v_i$ to $v_j$, $v_i$ needs to obtain cross parities from $v_j$ while decoding $v_0$, where the edges are solid if and only if $v_i$ is non-locally-recoverable.}
\label{decgraph}
\end{figure}

\begin{defi} \label{defi decoding graph} \emph{\textbf{(Decoding Graph)}} Let $G(V,E)$ represent a DSN with $|V|=p$ and $i\in \MB{p}$. Let $\mathcal{T}(\mathcal{V},\mathcal{E})$ denote a directed subgraph of $G$ associated with $v_j$. For all $v_i\in\mathcal{V}$, denote the set containing all children of $v_i$ by $\Se{V}_i^{\textup{C}}$, and that containing all parents of $v_i$ by $\Se{V}_i^{\textup{P}}$. Suppose $v_j\in\Se{V}$ is the only node without parents. We call this node the \textbf{root} of $\Se{T}$. We call any node without children a \textbf{leaf}. Suppose then the codewords of all the leaves of $\Se{T}$ are not locally-recoverable, and any other $v_i\in \mathcal{V}$ satisfies one of the following conditions.
\begin{enumerate}
\item The codeword stored at $v_i$ is locally-recoverable; $\Se{V}_i^{\textup{P}}\cup\Se{V}_i^{\textup{C}}$ consists of all the nodes in $\Se{M}_i$ such that codewords stored at them are not locally-recoverable and $|\Se{V}_i^{\textup{P}}|=1$.
\item The codeword stored at $v_i$ is not locally-recoverable; codewords stored at nodes in $\Se{V}_i^{\textup{P}}\cup\Se{V}_i^{\textup{C}}$ are all locally-recoverable.
\end{enumerate}
We call $\Se{T}$ a decoding graph at its root node $v_j$ over $G(V,E)$.
\end{defi}

As shown in Fig.~\ref{decgraph}, the decoding graph $\Se{T}$ at node $v_0$ is marked in red. Nodes marked in green are leaves in $\Se{T}$. Nodes $v_1,v_4,v_7,v_9,v_{11}$ satisfy Condition 1. Take $v_1$ as an example, $\Se{V}_1^{\textup{P}}=\{v_0\}$, and $\Se{V}_1^{\textup{C}}=\{v_2,v_3\}$. We know that $\Se{M}_1=\{v_0,v_2,v_3,v_{13},v_{14},v_{15}\}$, where codewords stored at nodes in $\Se{V}_i^{\textup{P}}\cup\Se{V}_i^{\textup{C}}=\{v_0,v_2,v_3\}$ are not locally-recoverable. This local constraint enforces the node in $\Se{V}_1^{\textup{P}}$, i.e., $v_0$, to obtain the extra cross parities from $v_1$ after codewords stored at nodes in $\Se{V}_1^{\textup{C}}$ are recovered. Nodes $v_0,v_2,v_3,v_5,v_{6},v_{10},v_8,v_{12}$ satisfy Condition 2, and they are the nodes that need to recover their codewords in order that $v_0$ recovers its codeword.

Based on the definition of decoding graphs, \Cref{theo: erausre pattern single} describes recoverable erasure patterns in a DSN.

\begin{theo} \label{theo: erausre pattern single} Let $\Se{C}$ be a code with single-level cooperation on a DSN represented by $G(V,E)$, where $\Se{C}$ and all related parameters are specified according to \Cref{cons: 1}. Let $\bold{u}\in\mathbb{N}^p$ such that $\bold{u}\preceq\bold{n}$. Suppose $\Se{C}$ and $\bold{u}$ satisfy the following conditions:
\begin{enumerate}
\item Let $V^{\textup{NL}}$ represent the set that contains all the nodes $v_i$, $i\in\MB{p}$ such that $u_i>r_i-\delta_i$. Let $V^{\textup{L}}=V\setminus V^{\textup{NL}}$. Then, for any $v_i\in V^{\textup{NL}}$, $\Se{M}_i\subset V^{\textup{L}}$.
\item For any $v_i\in V^{\textup{NL}}$, there exists a decoding graph $\mathcal{T}_i(\Se{V}_i,\Se{E}_i)$ at root $v_i$ over $G$. Moreover, for any leaf $v_j$ of $\mathcal{T}_i$, $u_j\leq r_j$; for any node $v_j\in \Se{V}_i\cap V^{\textup{NL}}$, $u_j\leq r_j+\sum\nolimits_{v_k\in\Se{V}_j^{\textup{C}}}\delta_k$. 
\end{enumerate} 
Then, $\bold{u}$ is a \emph{\textbf{recoverable erasure pattern}} of $\Se{C}$ over $G(V,E)$.
\end{theo}

\begin{proof} For any node $v_i\in\Se{V}^{\textup{NL}}$, consider the decoding graph $\mathcal{T}_i(\Se{V}_i,\Se{E}_i)$ at root $v_i$. Denote the number of nodes contained in $V^{\textup{L}}$ on the longest directed path connecting node $v_i$ with a leaf in $\Se{T}_i$ by $l_i$, which is referred to as the \textbf{decoding depth} of $v_i$. We prove the statement ``any node in $G$ is recoverable'' by mathematical induction on the decoding depth of the node. 

The decoding graph of a node with decoding depth $0$ contains only the node itself. The first condition implies that all neighbors of any node that is not locally-recoverable (in $\Se{V}^{\textup{NL}}$) are locally-recoverable (in $\Se{V}^{\textup{L}}$). Therefore, any node $v_i\in\Se{V}^{\textup{NL}}$ tolerates at least $r_i$ erasures. This means that nodes with decoding depth $0$ are recoverable.

Suppose the statement is true for any node $v_j\in \Se{V}^{\textup{NL}}$ with decoding depth less than or equal to ${\ell}\in\Db{N}$. Then, for any node $v_i$ with decoding depth $\ell+1$, let $S_i$ denote the union of all sets $\Se{V}_j^{\textup{C}}$ such that $v_j\in\Se{V}_i^{\textup{C}}$. Since the subgraph of $\mathcal{T}_i$ rooted at any node $v_j\in S_i$ is a decoding graph of $v_j$ with length at most $\ell$, $v_j$ is recoverable. Condition 1) in \Cref{defi decoding graph} indicates that all neighbors of $v_j\in\Se{V}_i^{\textup{C}}$ except for $v_i$ are recoverable and $\bold{m}_i\bold{B}_{i,j}$ is known, which provides $\delta_j$ extra parities of $\bold{m}_i$. Therefore, node $v_i$ tolerates up to $r_i+\sum\nolimits_{v_j\in\Se{V}_i^{\textup{C}}}\delta_j$ erasures, thus is recoverable according to Condition 2). Consequently, the statement for $\ell+1$ is also true. 

By induction, the statement is true for all the nodes, and the theorem is true.
\end{proof}

\begin{figure}
\centering
\includegraphics[width=0.47\textwidth]{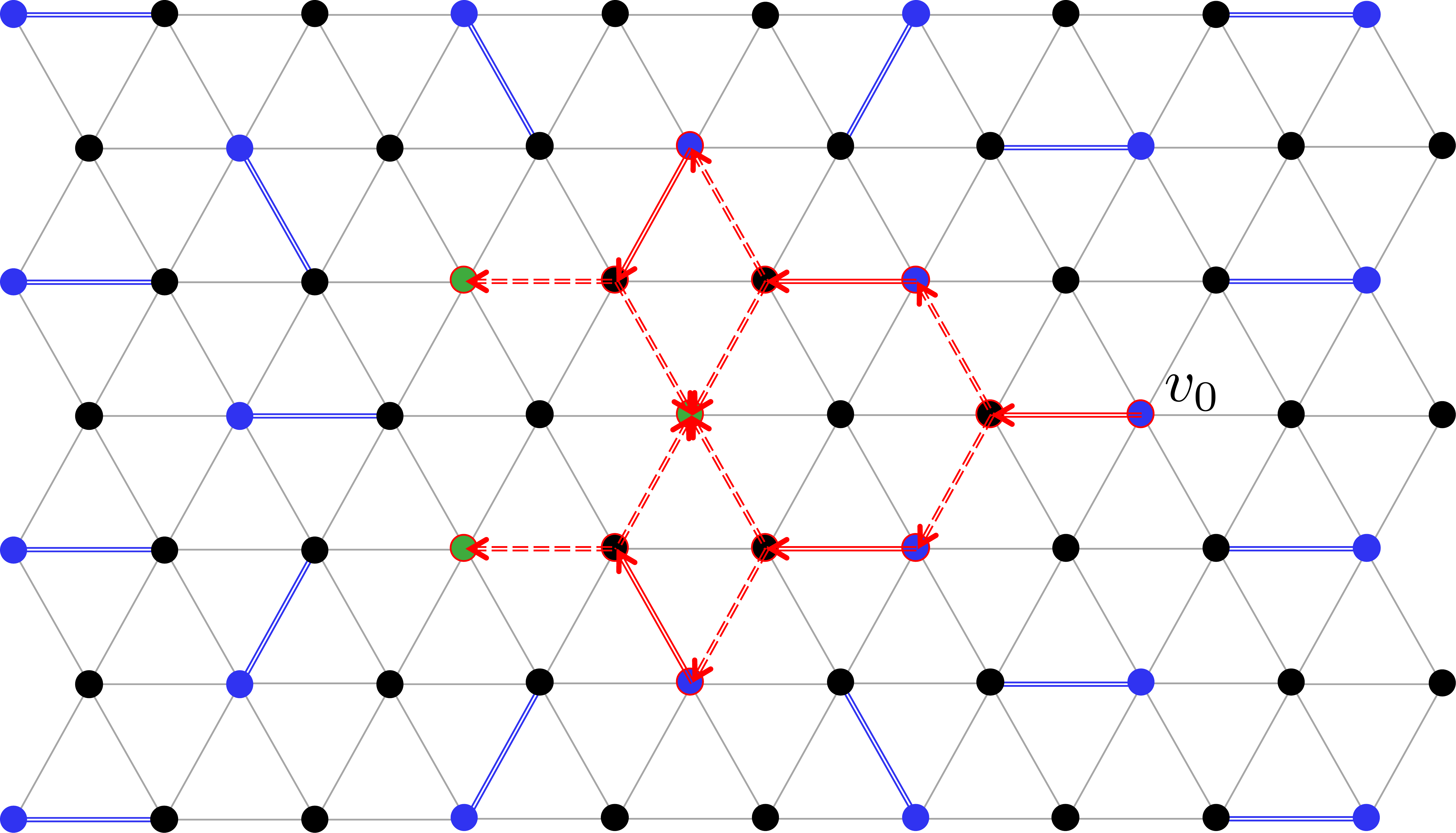}
\hspace{0.8cm}
\includegraphics[width=0.47\textwidth]{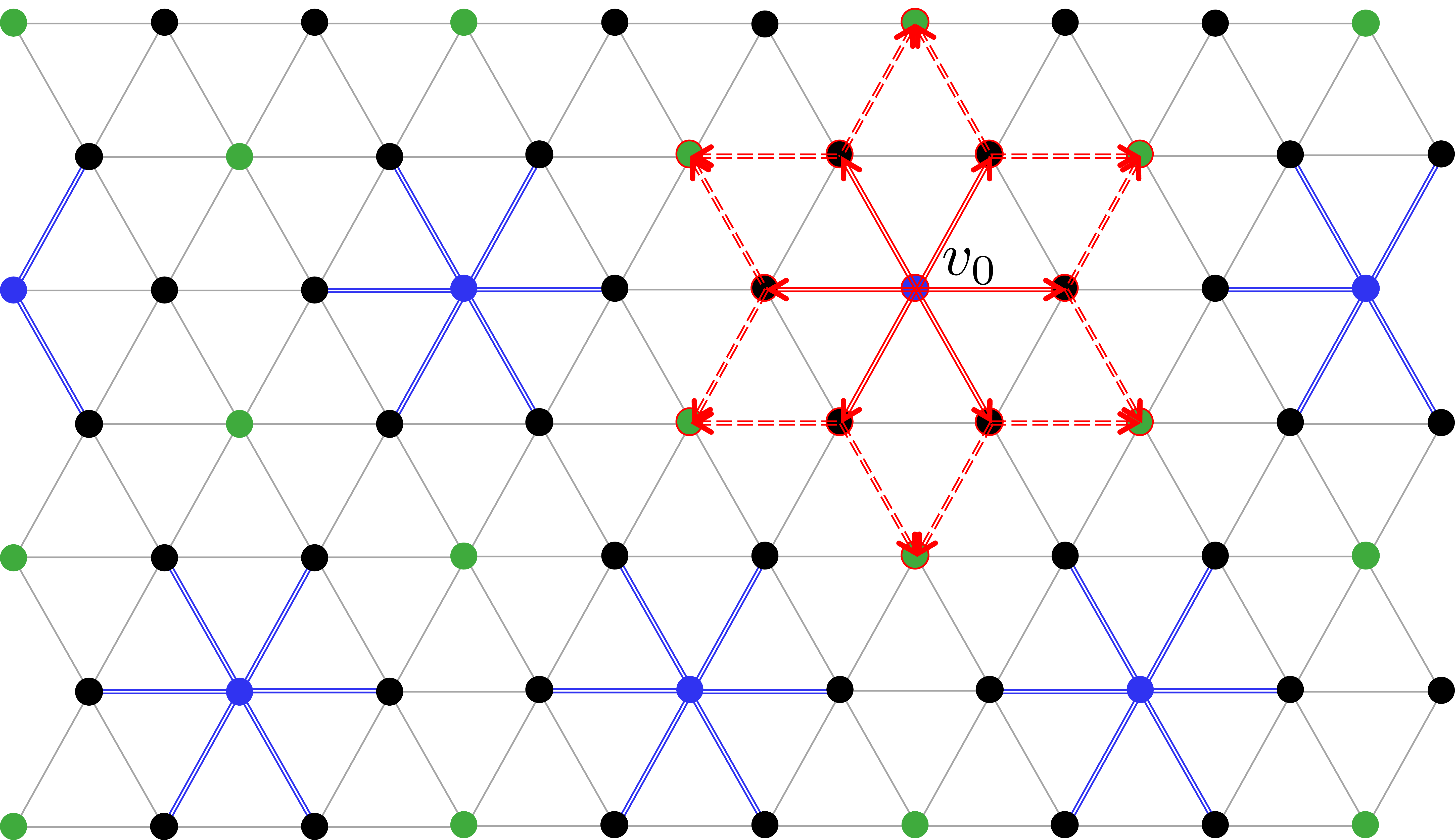}
\caption{DSN in \Cref{exam11}. Black and non-black nodes refer to nodes that are locally-recoverable and non-locally-recoverable, respectively. The decoding graphs at $v_0$ are marked with color red. Any solid blue line connects a black node $v_i$ and an non-black node $v_j$, where $v_i$ is a child of $v_j$, and all neighbors of $v_i$ except for $v_j$ are children of $v_i$, in the decoding graph at root $v_j$. The blue lines uniquely describe the decoding graphs at each non-black node.}
\label{fig11}
\end{figure}

The following two examples illustrate \Cref{theo: erausre pattern single}. They follow the notation in \Cref{cons: 1} and \Cref{theo: erausre pattern single}.

\begin{exam} \label{exam11} 

Fig.~\ref{fig11} presents two erasure patterns, $\bold{u}_1=(u_{1,1},u_{1,2},\dots,u_{1,p})$ (left) and $\bold{u}_2=(u_{2,1},u_{2,2},\dots,u_{2,p})$ (right), on the same DSN denoted by $G(V,E)$ with the EC solution characterized by $\Se{C}$ specified in \Cref{cons: 1}. Suppose there exists $\delta\in \Db{N}$ such that $\delta_i=\delta$ for all $i\in\MB{p}$. 

For any $i\in\MB{p}$, $v_i$ is marked in black if $u_{j,i}\leq r_i-\delta$, in green if $r_i-\delta<u_{j,i}\leq r_i$, and in blue otherwise, where $j\in\{0,1\}$. In the left panel, that specifies $\bold{u}_1$, any node $v_i$ marked in blue satisfies $r_i<u_{1,i}\leq r_i+\delta$. In the right panel, that specifies $\bold{u}_2$, any node $v_i$ marked in blue satisfies $r_i<u_{2,i}\leq r_i+6\delta$. 

Note that any non-black node is connected to exactly one black node by a blue edge, where the non-black node is the only parent of the black node in \Cref{defi decoding graph}. Then, for any node $v_i$ in $G(V,E)$, there exists a decoding graph at $v_i$, with the leaves being all marked in green. In Fig.~\ref{fig11}, the decoding graph at the node $v_0$ is marked in red on each graph of the two. Let $d_i$ be the maximum number of erasures node $v_i$ tolerates, for all $i\in\MB{p}$, and $\Delta_i=d_i-r_i$. Suppose $p\to\infty$ in $G(V,E)$. Denote the average of all $\Delta_i$'s by $\Delta$.

In the first subgraph, there will be approximately $2p/3$ nodes with any $v_i$ of them satisfying $\Delta_i=-\delta$, and approximately $p/3$ nodes with any $v_i$ of them satisfying $\Delta_i=\delta$. Thus, $\Delta=-\delta/3$. Similarly, in the second graph, there will be approximately $2p/3$, $2p/9$, and $p/9$ nodes with any $v_i$ of them satisfying $\Delta_i=-\delta$, $\Delta_i=0$, and $\Delta_i=6\delta$, respectively. Thus, $\Delta=0$. 
\end{exam}

\begin{figure}
\centering
\includegraphics[width=0.47\textwidth]{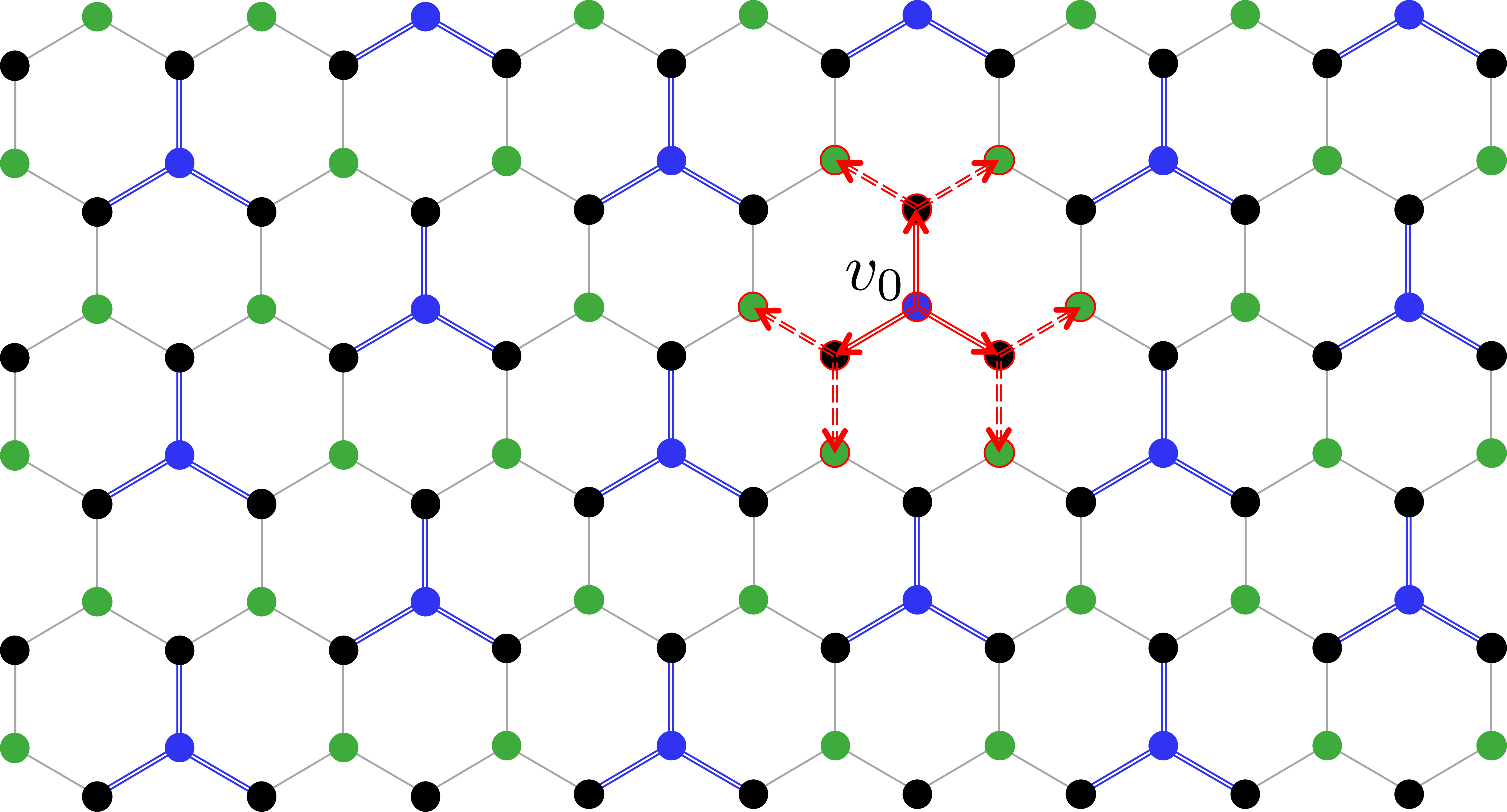}
\hspace{0.8cm}
\includegraphics[width=0.47\textwidth]{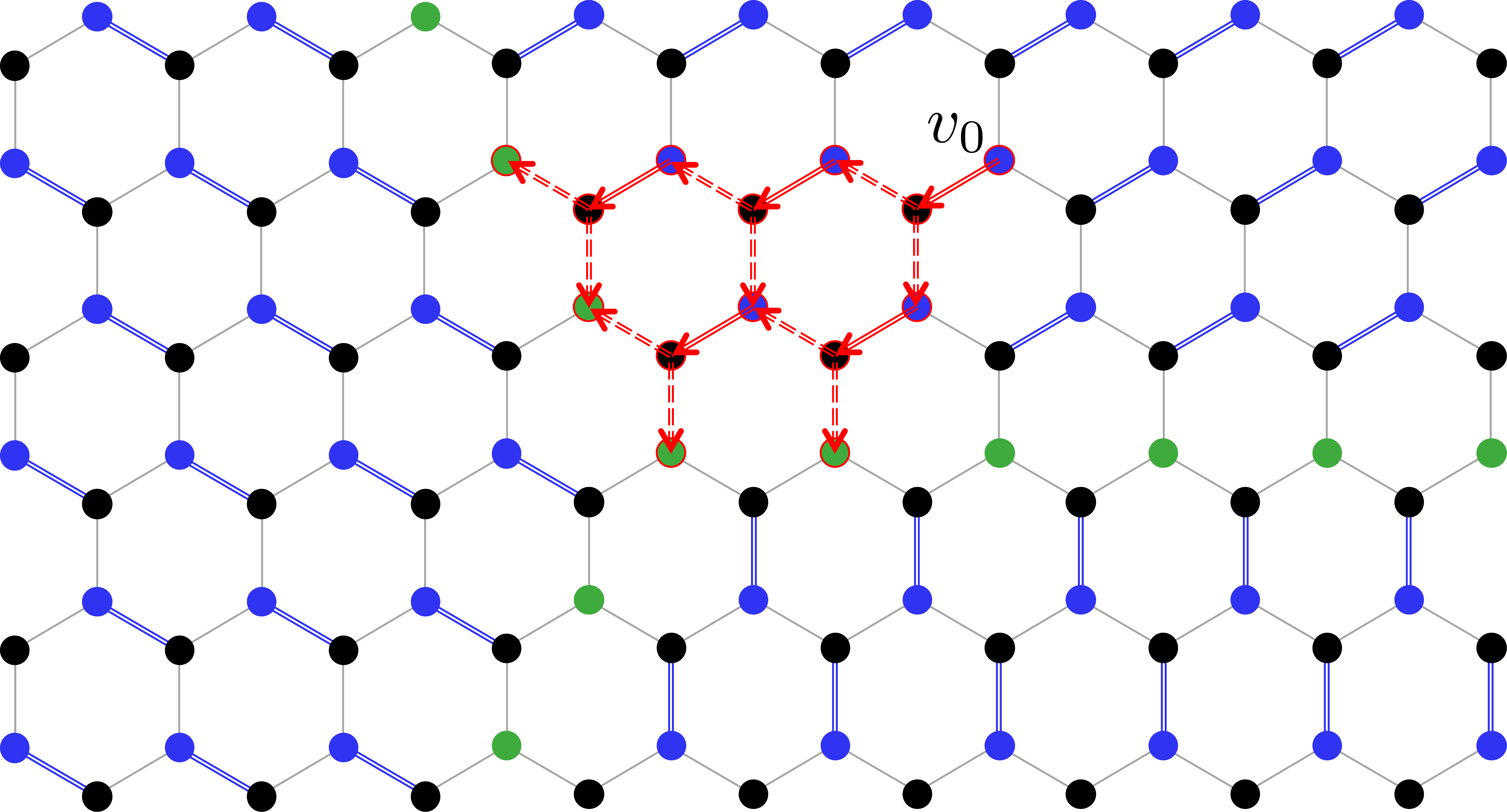}
\caption{DSN in \Cref{exam12}. Meanings of the components of the graphs are identical to those in Fig.~$\ref{fig11}$.}
\label{fig12}
\end{figure}

\begin{exam} \label{exam12} 

Similar to \Cref{exam11}, Fig.~\ref{fig12} also presents two erasure patterns on another DSN with an EC solution characterized by $\Se{C}$ specified in \Cref{cons: 1}. Suppose there exists $\delta\in \Db{N}$ such that $\delta_i=\delta$ for all $i\in\MB{p}$. 

For any $i\in\MB{p}$, $v_i$ is marked in black if $u_{j,i}\leq r_i-\delta$, in green if $r_i-\delta<u_{j,i}\leq r_i$, and in blue otherwise, where $j\in\{0,1\}$. In the left panel, that specifies $\bold{u}_1$, any node $v_i$ marked in blue satisfies $r_i<u_{1,i}\leq r_i+3\delta$. In the right panel, that specifies $\bold{u}_2$, any node $v_i$ marked in blue satisfies $r_i<u_{2,i}\leq r_i+\delta$. In Fig.~\ref{fig12}, the decoding graph at the node $v_0$ is marked in red on both graphs of the two.

We follow the definitions of $\Delta_i$'s, $i\in\MB{p}$, and $\Delta$, stated in \Cref{exam11}. In the left panel, there will be approximately $p/2$ nodes with any $v_i$ of them satisfying $\Delta_i=-\delta$, approximately $p/3$ nodes with any $v_i$ of them satisfying $\Delta_i=0$, and $p/6$ nodes with any $v_i$ of them satisfying $\Delta_i=3\delta$. Thus, $\Delta=0$. Similarly, in the right panel, there will be approximately $p/2$ and $p/2$ nodes with any $v_i$ of them satisfying $\Delta_i=-\delta$ and $\Delta_i=\delta$, respectively. Thus, $\Delta=0$. 
\end{exam}

\begin{figure}
\centering
\includegraphics[width=0.36\textwidth]{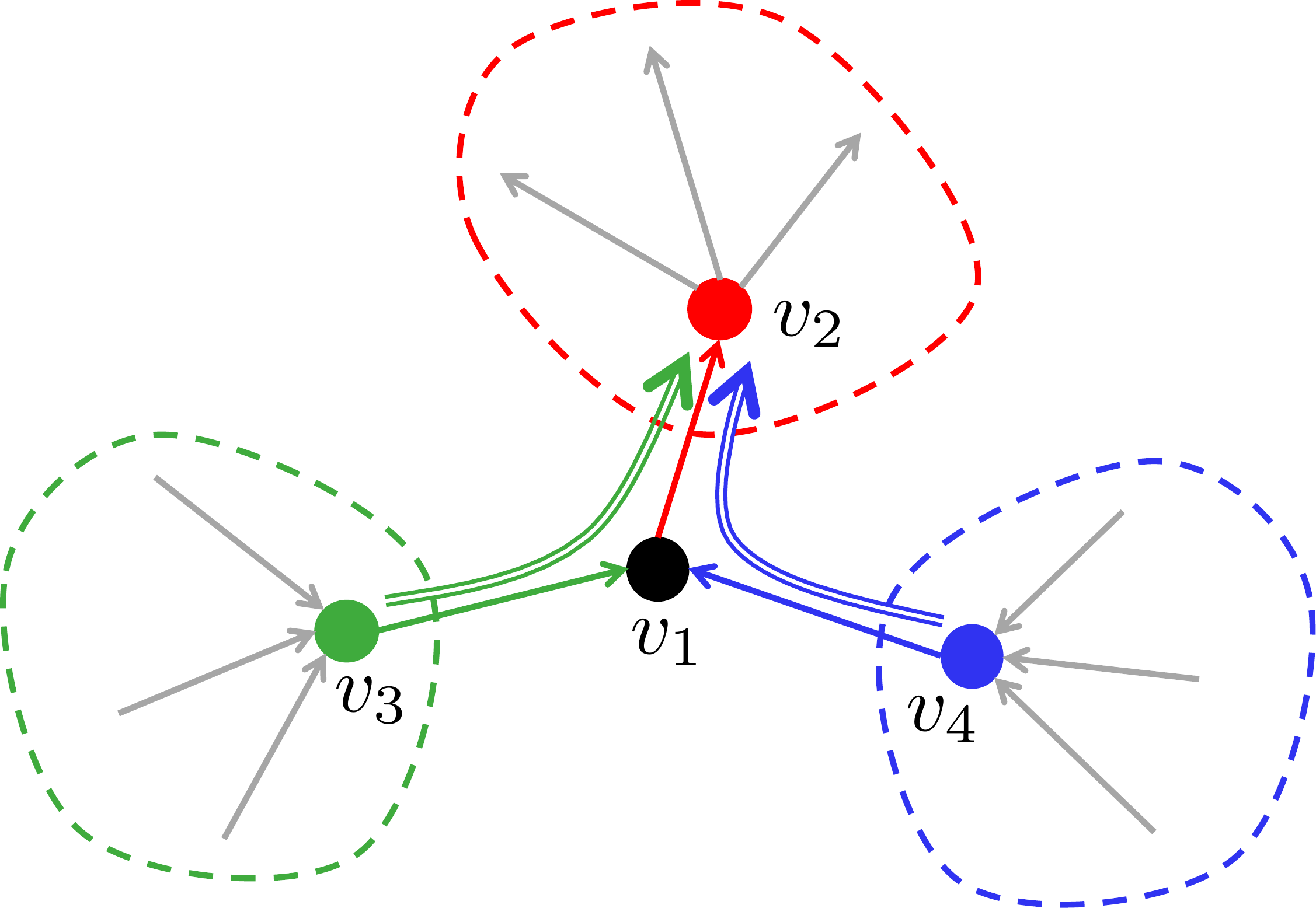}
\caption{Information flow in cooperative data protection. As a neighbor of $v_2$, $v_1$ helps in removing the cross parities from the parity part of $\bold{c}_2$. However, nodes $v_3$ and $v_4$ also indirectly help $v_1$ to provide extra cross parities to $v_2$ if they are recovered. This can be interpreted as information flow from $v_3$ and $v_4$ to $v_2$ through $v_1$.}
\label{fig: information flow}
\end{figure}

\begin{rem} \label{rem: infoflow} \emph{\textbf{(Information Flow in Coded DSN)}} Note that the values of $\Delta_i$'s, $i\in\MB{p}$, in \Cref{exam11} and \Cref{exam12} imply the unbalanced reliabilities of nodes in the coded DSN. In other words, any node (black) $v_i$ with $\Delta_i<0$ is of higher reliability than any node (blue) $v_i$ with $\Delta_i>0$. Therefore, any blue node utilizes extra information from non-black nodes in its neighborhood. The average $\Delta$ being nonnegative can be interpreted as a higher level of intrinsic information flow among nodes with different reliabilities in the coded DSN.

For example, as shown in Fig.~\ref{fig: information flow}, suppose $v_1$ is a locally-recoverable node with neighboring nodes $v_2$, $v_3$, and $v_4$. If $v_3$ and $v_4$ are recoverable, then $v_1$ provides $\delta_1$ extra parities to node $v_2$. Therefore, the information flows from $v_3$ and $v_4$ to $v_2$ through $v_1$, which is depicted in the figure. 
\end{rem}

\section{Multi-Level Cooperation}
\label{section multi-level cooperation}

In this section, we extend the construction presented in Subsection~\ref{subsec single-level cooperation} to codes with EC hierarchies of depth larger than $1$. As is shown in schemes with single-level cooperation, cooperation utilizes the redundant information from nodes with fewer erasures to help in decoding of nodes that cannot be decoded locally. However, each node only obtains additional parities from its neighbors in the single-level cooperation, which immediately motivates us to explore multi-level cooperation to further improve the global EC capability of each node. Although multi-level cooperation inevitably degrades the local EC capability of each node, it enables the DSN to tolerate erasure patterns where erasures are distributed non-uniformly among the nodes, such as bursty erasures in few sparsely scattered nodes. In this section, we investigate the EC hierarchy of multi-level cooperation schemes. We first define the so-called \textbf{cooperation graphs} that describe how the nodes are coupled to cooperatively transmit information, and then prove the existence of hierarchical codes over a special class of cooperation graphs: the so-called \textbf{compatible graphs}.

\begin{figure*}[!t]
\normalsize
\centering

\setcounter{equation}{5}
\setcounter{MaxMatrixCols}{12}
\begin{equation}\small
\scalebox{.7}{$
\begin{matrix}

\textcolor{robert}{1}&\textcolor{robert}{2}&\textcolor{robert}{3}&\textcolor{robert}{4}&\textcolor{robert}{5}&\textcolor{robert}{6}&\textcolor{robert}{7}&\textcolor{robert}{8}&\textcolor{robert}{9}&\textcolor{robert}{10}&\textcolor{robert}{11}&\textcolor{robert}{12}\\
\hline
\hline
\bold{A}_{1,1}&\bold{B}_{1,2}\bold{U}_2&\bold{0}&\bold{0}&\bold{0}&\bold{0}&\bold{0}&\bold{0}&\bold{0}&\bold{0}&\bold{0}&\bold{0}\\
\hline
\bold{B}_{2,1}\bold{U}_1&\bold{A}_{2,2}&\bold{B}_{2,3}\bold{U}_3&\bold{0}&\bold{B}_{2,5}\bold{U}_5&\bold{0}&\bold{0}&\textcolor{red}{\bold{B}_c\bold{V}_{8;2}}&\textcolor{red}{\bold{B}_c\bold{V}_{9;2}}&\textcolor{red}{\bold{B}_d\bold{V}_{10;3}}&\textcolor{red}{\bold{B}_d\bold{V}_{11;3}}&\bold{0}\\
\hline
\bold{0}&\bold{B}_{3,2}\bold{U}_2&\bold{A}_{3,3}&\bold{B}_{3,4}\bold{U}_4&\bold{0}&\bold{0}&\bold{0}&\textcolor{red}{\bold{B}_e\bold{V}_{8;2}}&\textcolor{red}{\bold{B}_e\bold{V}_{9;2}}&\textcolor{red}{\bold{B}_f\bold{V}_{10;3}}&\textcolor{red}{\bold{B}_f\bold{V}_{11;3}}&\bold{0}\\
\hline
\bold{0}&\bold{0}&\bold{B}_{4,3}\bold{U}_3&\bold{A}_{4,4}&\bold{B}_{4,5}\bold{U}_5&\bold{B}_{4,6}\bold{U}_6&\bold{0}&\textcolor{apple green}{\bold{B}_{\alpha}\bold{V}_{8;2}}&\textcolor{apple green}{\bold{B}_{\alpha}\bold{V}_{9;2}}&\textcolor{blue}{\bold{B}_g\bold{V}_{10;2}}&\textcolor{blue}{\bold{B}_g\bold{V}_{11;2}}&\bold{0}\\
\hline
\bold{0}&\bold{B}_{5,2}\bold{U}_2&\bold{0}&\bold{B}_{5,4}\bold{U}_4&\bold{A}_{5,5}&\bold{B}_{5,6}\bold{U}_6&\bold{0}&\bold{B}_{5,8}\bold{U}_8&\bold{0}&\bold{0}&\textcolor{blue}{\bold{B}_h\bold{V}_{11;2}}&\textcolor{blue}{\bold{B}_h\bold{V}_{12;2}}\\
\hline
\bold{0}&\bold{0}&\bold{0}&\bold{B}_{6,4}\bold{U}_4&\bold{B}_{6,5}\bold{U}_5&\bold{A}_{6,6}&\bold{B}_{6,7}\bold{U}_7&\textcolor{apple green}{\bold{B}_{\beta}\bold{V}_{8;2}}&\textcolor{apple green}{\bold{B}_{\beta}\bold{V}_{9;2}}&\textcolor{blue}{\bold{B}_j\bold{V}_{10;2}}&\bold{0}&\textcolor{blue}{\bold{B}_j\bold{V}_{12;2}}\\
\hline
\bold{0}&\bold{0}&\bold{0}&\bold{0}&\bold{0}&\bold{B}_{7,6}\bold{U}_6&\bold{A}_{7,7}&\bold{B}_{7,8}\bold{U}_8&\bold{B}_{7,9}\bold{U}_9&\textcolor{blue}{\bold{B}_l\bold{V}_{10;2}}&\bold{B}_{7,11}\bold{U}_{11}&\textcolor{blue}{\bold{B}_l\bold{V}_{12;2}}\\
\hline
\bold{0}&\textcolor{red}{\bold{B}_m\bold{V}_{2;2}}&\textcolor{red}{\bold{B}_m\bold{V}_{3;2}}&\textcolor{apple green}{\bold{B}_{y}\bold{V}_{4;2}}&\bold{B}_{8,5}\bold{U}_5&\textcolor{apple green}{\bold{B}_{y}\bold{V}_{6;2}}&\bold{B}_{8,7}\bold{U}_7&\bold{A}_{8,8}&\bold{B}_{8,9}\bold{U}_9&\textcolor{blue}{\bold{B}_n\bold{V}_{10;2}}&\textcolor{blue}{\bold{B}_n\bold{V}_{11;2}}&\bold{0}\\
\hline
\bold{0}&\textcolor{red}{\bold{B}_o\bold{V}_{2;2}}&\textcolor{red}{\bold{B}_o\bold{V}_{3;2}}&\textcolor{apple green}{\bold{B}_{z}\bold{V}_{4;2}}&\bold{0}&\textcolor{apple green}{\bold{B}_{z}\bold{V}_{6;2}}&\bold{B}_{9,7}\bold{U}_7&\bold{B}_{9,8}\bold{U}_8&\bold{A}_{9,9}&\bold{B}_{9,10}\bold{U}_{10}&\textcolor{blue}{\bold{B}_p\bold{V}_{11;2}}&\textcolor{blue}{\bold{B}_p\bold{V}_{12;2}}\\
\hline
\bold{0}&\textcolor{red}{\bold{B}_q\bold{V}_{2;3}}&\textcolor{red}{\bold{B}_q\bold{V}_{3;3}}&\textcolor{blue}{\bold{B}_r\bold{V}_{4;2}}&\bold{0}&\textcolor{blue}{\bold{B}_r\bold{V}_{6;3}}&\textcolor{blue}{\bold{B}_s\bold{V}_{7;3}}&\textcolor{blue}{\bold{B}_s\bold{V}_{8;3}}&\bold{B}_{10,9}\bold{U}_9&\bold{A}_{10,10}&\bold{B}_{10,11}\bold{U}_{11}&\bold{B}_{10,12}\bold{U}_{12}\\
\hline
\bold{0}&\textcolor{red}{\bold{B}_x\bold{V}_{2;3}}&\textcolor{red}{\bold{B}_x\bold{V}_{3;3}}&\textcolor{blue}{\bold{B}_t\bold{V}_{4;2}}&\textcolor{blue}{\bold{B}_t\bold{V}_{5;2}}&\bold{0}&\bold{B}_{11,7}\bold{U}_7&\textcolor{blue}{\bold{B}_u\bold{V}_{8;3}}&\textcolor{blue}{\bold{B}_u\bold{V}_{9;3}}&\bold{B}_{11,10}\bold{U}_{10}&\bold{A}_{11,11}&\bold{B}_{11,12}\bold{U}_{12}\\
\hline
\bold{0}&\bold{0}&\bold{0}&\bold{0}&\textcolor{blue}{\bold{B}_v\bold{V}_{5;2}}&\textcolor{blue}{\bold{B}_v\bold{V}_{6;3}}&\textcolor{blue}{\bold{B}_w\bold{V}_{7;3}}&\bold{0}&\textcolor{blue}{\bold{B}_w\bold{V}_{9;3}}&\bold{B}_{12,10}\bold{U}_{10}&\bold{B}_{12,11}\bold{U}_{11}&\bold{A}_{12,12}\\
\hline
\end{matrix}$}
\label{fig: example2}
\end{equation}
\hrulefill
\setcounter{equation}{6}
\end{figure*}

\subsection{Cooperation Graphs}
\label{subsec cooperation graphs}

Based on the aforementioned notation, for each $v_i\in V$ and ${\ell}\in\MB{L_i}$, let $\Se{I}^{\ell}_i=\Se{A}^{\ell}_i\setminus\Se{A}^{{\ell}-1}_i$ (with $\Se{A}^{0}_i = \varnothing$) and refer to it as the $\ell$-th \textbf{helper} of $v_i$. We next define the so-called \textbf{cooperation matrix}.

\begin{figure}
\centering
\includegraphics[width=0.75\textwidth]{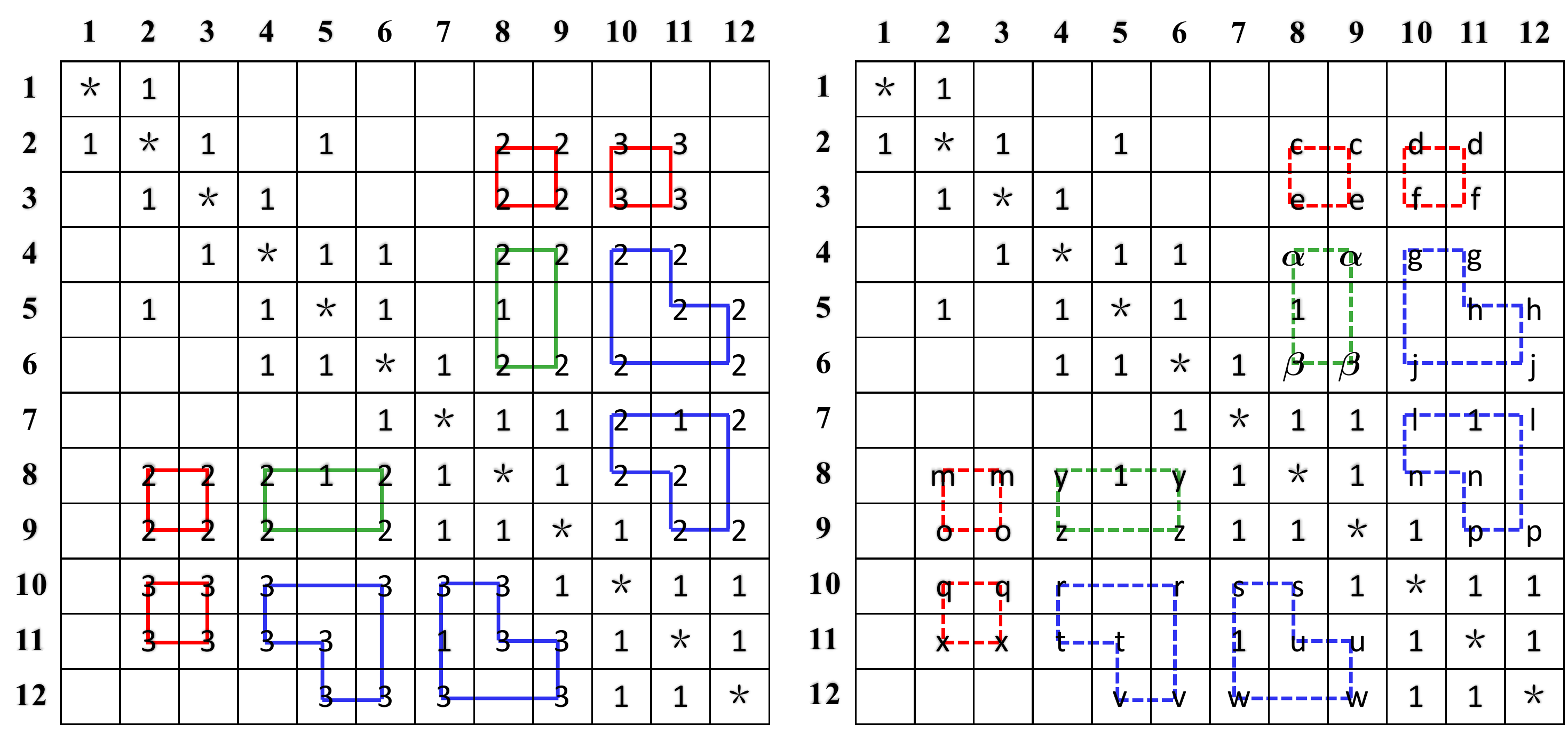}
\caption{Matrices $\bold{D}$ (left) and $\bold{X}$ (right) in \Cref{exam: exam3}. A numerical entry $\ell$ at position $(i,j)$ in the left panel implies that $v_j$ is adjacent to $v_i$ in the $\ell$-th level cooperation of $v_i$, while symbolic entries (letters) in the right panel represent the indices of the component matrices $\bold{A}_{i,j}$.}
\label{fig: CoopMatrix}
\end{figure}

\begin{defi} \label{defi: cooperation matrix} For a joint coding scheme $\Se{C}$ for a DSN represented by $G(V,E)$ with $|V|=p$, the matrix $\bold{D}\in\mathbb{N}^{p\times p}$, in which $\bold{D}_{i,j}$ equals to $\ell$ for all $i,j\in\MB{p}$ such that $j\in\mathcal{I}^{\ell}_i$, ${\ell}\in\MB{L_i}$, and zero otherwise, is called the \textbf{cooperation matrix}.
\end{defi}

As an example, the cooperation matrix in \Cref{exam: exam1} is exactly the adjacency matrix of the graph in Fig.~\ref{fig: example1}. Note that cooperation graphs corresponding to some joint coding schemes must satisfy certain properties. In Subsection~\ref{subsec construction over compatible graphs}, we prove the existence of codes if the cooperation matrix represents a so-called \textbf{compatible graph}. Before going into details of the construction, we present an example to provide some intuition.

\begin{exam}\label{exam: exam3} Recall the DSN in \Cref{exam: exam1}. We present a coding scheme with the cooperation matrix specified in the left panel of Fig.~\ref{fig: CoopMatrix}. The non-systematic part of the generator matrix is shown in (\ref{fig: example2}), which is obtained through the following process:
\begin{enumerate}
\item Partition all the non-zero, non-one elements into structured groups, each of which is marked in either a rectangle or a hexagon in $\bold{D}$, as indicated in the left panel of Fig.~\ref{fig: CoopMatrix}.
\item Replace the endpoints of each horizontal line segment in Step 1 with $s\in S$ ($S$ is a set of symbols), as indicated in the right panel of Fig.~\ref{fig: CoopMatrix}; denote the new matrix by $\bold{X}$.
\item Assign a parameter $\gamma_s\in\Db{N}$ to each $s\in S$, and a matrix $\bold{B}'_{s}\in \textup{GF}^{k_i\times\gamma_s}$ to any $(i,j)$ such that $\bold{X}_{i,j}=s$.
\item For each $i\in\MB{p}$, ${\ell}\in\MB{L}$, let $\eta_{i;{\ell}}=\max\nolimits_{s: k\in \Se{I}_i^{\ell},\bold{X}_{k,i}=s} \gamma_{s}$, assign $\bold{V}_{i;{\ell}}\in\textup{GF}(q)^{\eta_{i;{\ell}}\times r_i}$ to $v_i$; let $\bold{B}_{s}=\left[\bold{B}'_s,\bold{0}_{\eta_{i;{\ell}}-\gamma_s}\right]$; compute $\bold{A}_{i,j}=\bold{B}_s\bold{V}_{j;{\ell}}$ for $s=\bold{X}_{i,j}$, $l=\bold{D}_{i,j}$.
\item Compute $\bold{A}_{i,j}$ for $\bold{X}_{i,j}=1$ according to \Cref{cons: 1}.
\end{enumerate}

Note that the colors of submatrices in (\ref{fig: example2}) are consistent with the colors of cycles in Fig.~\ref{fig: CoopMatrix}.
 Let us again focus on node $v_2$. Let $\Se{I}_2^1=\{v_1,v_3,v_5\}$, $\Se{I}_2^2=\{v_8,v_9\}$, $\Se{I}_2^3=\{v_{10},v_{11}\}$. Then, $\Se{B}_2^1=\{v_4,v_6,v_8\}$, $\Se{B}_2^2=\{v_4,v_6\}$, $\Se{B}_2^3=\varnothing$, $d_{2,0}=r_2-\delta_2-\eta_{i;2}-\eta_{i;3}$, $d_{2,1}=r_2+\delta_1+\delta_3+\delta_5$, $d_{2,2}=d_{2,1}+\gamma_{c}$, $d_{2,3}=d_{2,2}+\gamma_{d}$. Note that for each $i\in \MB{p}$, $s=\bold{X}_{i,j}$, and $l=\bold{D}_{i,j}$, $\gamma_s$ denotes the maximum number of parity symbols $v_i$ can obtain from $v_j$ in the $\ell$-th level cooperation, and $\eta_{i;{\ell}}$ represents the reduction in the value of the local erasure correction capability needed at $v_i$ because of its $\ell$-th level cooperation.

We first show that knowing $\{\bold{m}_j\}_{v_j\in\Se{A}_2^1}$ is sufficient for removing $\bold{s}_2=\sum\nolimits_{j\in\Se{I}_2^1}\bold{m}_j\bold{B}_{j,2}\bold{U}_{2}+\sum\nolimits_{{\ell}=2}^{3}\sum\nolimits_{j\in\Se{I}_2^{\ell}}\bold{m}_j\bold{B}_{\bold{X}_{j,2}}\bold{V}_{2;{\ell}}$ from the parity part of $\bold{c}_2$. Note that if the rows of $\bold{A}_{i,i}$, $\bold{U}_{i}$, and $\{\bold{V}_{i;{\ell}}\}_{{\ell}\in\{2,3\}}$ are linearly independent, then for all $\ell$, $\sum\nolimits_{j\in\Se{I}_i^{\ell}}\bold{m}_j\bold{B}_{\bold{X}_{j,i}}$ is recoverable if $\bold{m}_i$ is recoverable. In our example, this means that $\{\bold{m}_j\bold{B}_{j,2}\}_{j=1,3,5}$, $\bold{m}_{8}\bold{B}_{m}+\bold{m}_9\bold{B}_o$, $\bold{m}_{10}\bold{B}_{o}+\bold{m}_{11}\bold{B}_x$ are known, which means $\bold{s}_2$ is also known. Therefore, $\bold{s}_2$ is removed from the parity part of $\bold{c}_2$ through the $1$-st level cooperation. We next show that additional parities are obtained through $\ell$-th level cooperations with \textcolor{lara}{$\ell \in \{2,3\}$}.

In the $2$-nd level cooperation, $\bold{m}_8,\bold{m}_9$ are known. Therefore, $\bold{m}_2\bold{B}_c+\bold{m}_3\bold{B}_e+\bold{m}_4\bold{B}_{\alpha}+\bold{m}_6\bold{B}_{\beta}$ is also known. We remove $\bold{m}_3\bold{B}_{e}$, that is obtained via $v_3$, from the parity part of $\bold{c}_2$. In order to obtain the $\gamma_c$ parities from $\bold{m}_2\bold{B}_c$, one needs $\bold{m}_4,\bold{m}_6$ to be recoverable. Therefore, $\Se{B}_2^2=\{v_4,v_6\}$, $d_{2,2}=d_{2,1}+\gamma_{c}$, $\lambda_{2,2;\varnothing}=d_{2,1}$. 
\end{exam}

The cooperation matrix adopts a partition of non-zero, non-one elements into groups where each of them forms a cycle (see \Cref{exam: exam3}). Suppose there are $T$ cycles. Represent each cycle with index $t\in \MB{T}$ by a tuple $C_t=(X_t,Y_t,\{X_{t;j}\}_{j\in Y_t},\{Y_{t;i}\}_{i\in X_t},g_t,({\ell}_{t;j})_{j\in Y_{t}})$, where $X_t$ and $Y_t$ denote the sets containing indices of the rows and the columns of the cycle, respectively. Let $X_{t;j}=\{i_1,i_2\}$ for $j\in Y_t$, where $(i_1,j),(i_2,j)$ are the vertices of the cycle $C_t$ with column index $j$. Let $Y_{t;i}=\{j_1,j_2\}$ for $i\in X_t$, where $(i,j_1),(i,j_2)$ are the vertices of the cycle $C_t$ with row index $i$. Let $g_t$ denote a group number assigned to the cycle $C_t$, which will be explained shortly. Observe that any two vertices of a cycle that share the same column have the same cooperation level. Let ${\ell}_{t;j}$ denote the number representing the cooperation level assigned to the vertices $(i_1,j),(i_2,j)$ of the cycle $C_t$ where $X_{t;j}=\{i_1,i_2\}$. Suppose values in $(g_t)_{t\in\MB{T}}$ span all the values in $\MB{A}$, for some $A\in\Db{N}$. For any $g\in \MB{A}$, denote the set containing all $t$ such that $g_t=g$ by $T_g$.

For example, let $t=1$ for the blue cycle at the bottom left panel of the matrices in Fig.~\ref{fig: CoopMatrix}. Then, the cycle $C_1$ is represented by $(\{10,11,12\},\{4,5,6\},\{X_{1;j}\}_{j=4}^6,\{Y_{1;i}\}_{i=10}^{12},1,({\ell}_{1;j})_{j=4}^6)$, where $X_{1;4}=\{10,11\}$, $X_{1;5}=\{11,12\}$, $X_{1;6}=\{10,12\}$, $Y_{1;10}=\{4,6\}$, $Y_{1;11}=\{4,5\}$, $Y_{1;12}=\{5,6\}$, and ${\ell}_{1;4}={\ell}_{1;5}={\ell}_{1;6}=3$.

\begin{figure}
\centering
\includegraphics[width=0.5\textwidth]{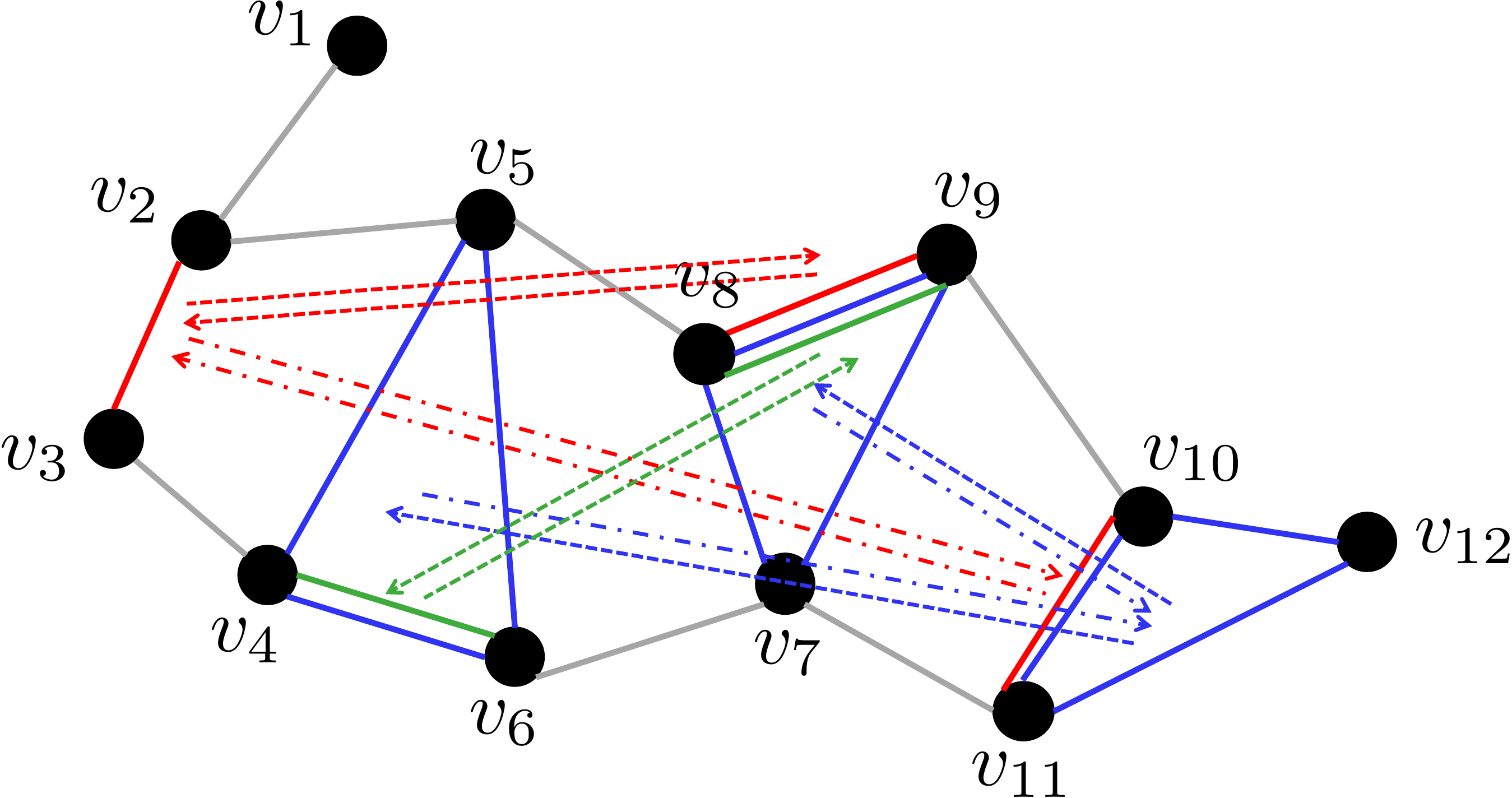}
\caption{Cooperation graph of \Cref{exam: exam3}. Dashed directed edges are used for cooperation level $2$, and dash-dotted directed edges are used for cooperation level $3$.}
\label{fig: figMHP}
\end{figure}

Observe that cycle $C_t$, $t\in\MB{T}$, in Fig.~\ref{fig: CoopMatrix} essentially represents a cycle in the complementary graph\footnote{The complementary graph of a graph $G(V,E)$ consists of all nodes in $V$ and all edges that are not in $E$.}. $\bar{G}$ of $G$ on $V$, since each vertex $(i,j)$ on the cycle implies that $\bold{D}_{i,j}\neq 1$, i.e., there is no edge connecting $v_i$ and $v_j$ in $G$, and there is an edge $(i,j)$ in the complementary graph $\bar{G}$. Cycle $C_t$ can also be interpreted as a pair of non-adjacent edges or non-overlapping triangles in $G$ with vertices from $X_t$ and $Y_t$. We mark $X_t$, $Y_t$, add an edge with arrow and \textcolor{lara}{the} label $g_t$ that points from $X_t$ to $Y_t$ for each $t\in \MB{T}$, and obtain the so-called \textbf{cooperation graph}\footnote{Note that there are two types of edges in the graph. The edges of one type connect nodes that are directly connected in their $1$-st level cooperation, and these are solid edges in the graph. The edges of the other type point between two groups of nodes (edges or triangles) that are adjacent in their higher level cooperation, and these are dashed arrows in the graph. For simplicity, ``connected'' here is used with the meaning of ``adjacent''.}. The cooperation graph for the coding scheme in \Cref{exam: exam3} is shown in Fig.~\ref{fig: figMHP}.

From the aforementioned description, a cooperation graph does not necessarily lead to a unique cooperation matrix, since the latter requires not only to further specify the associated set of cooperation levels, but also to identify the associated local matching graph, to be defined soon, for each cycle $C_t$ if $X_t$ has more than two nodes. For example, Fig.~\ref{fig: matching} presents three out of a total of six possible ways to specify $\{X_{t;j}\}_{j\in Y_t}$ and $\{Y_{t;i}\}_{i\in X_t}$ for a cooperation cycle $C_t$ with $X_t=\{i_1,i_2,i_3\}$ and $Y_t=\{j_1,j_2,j_3\}$. As indicated by Fig.~\ref{fig: matching}, each specified cycle of the three is uniquely represented by the set $\{(i,j)\}_{i\in X_t,Y_t\setminus Y_{t;i}=\{j\}}$ of edges that are not on the cycle. \textcolor{lara}{In these three graphs, if a double-dashed blue line with bidirectional arrows connects node $v_{i_x}$ to node $v_{j_y}$, this means that $v_{i_x}$ cooperates with the two nodes in $Y_t\setminus\{v_{j_y}\}$ in the cooperation involving vertices on the two triangles containing $v_{j_x}$, $v_{i_y}$, respectively. We} refer to the resulting graph as a \textbf{matching graph} corresponding to a cycle. If there is more than one cycle involved in the matching graph, we call that graph a \textbf{local matching graph}. If cycles in the local matching graph are all the cycles of a cycle group, we call that graph an \textbf{isolated local matching graph} and will discuss it in detail in Subsection~\ref{subsec construction over compatible graphs}.

\begin{figure}
\centering
\includegraphics[width=0.95\textwidth]{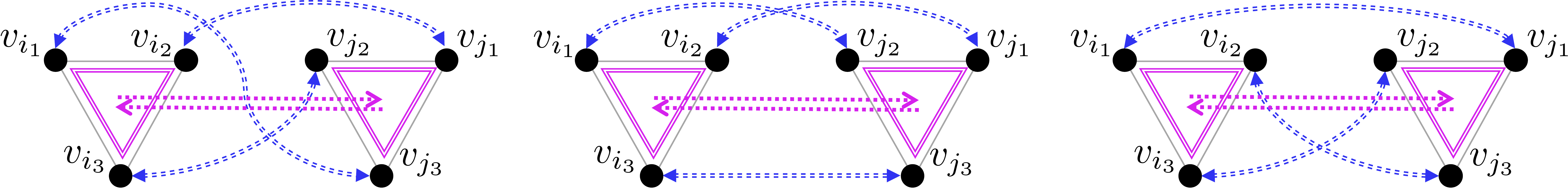}
\includegraphics[width=0.94\textwidth]{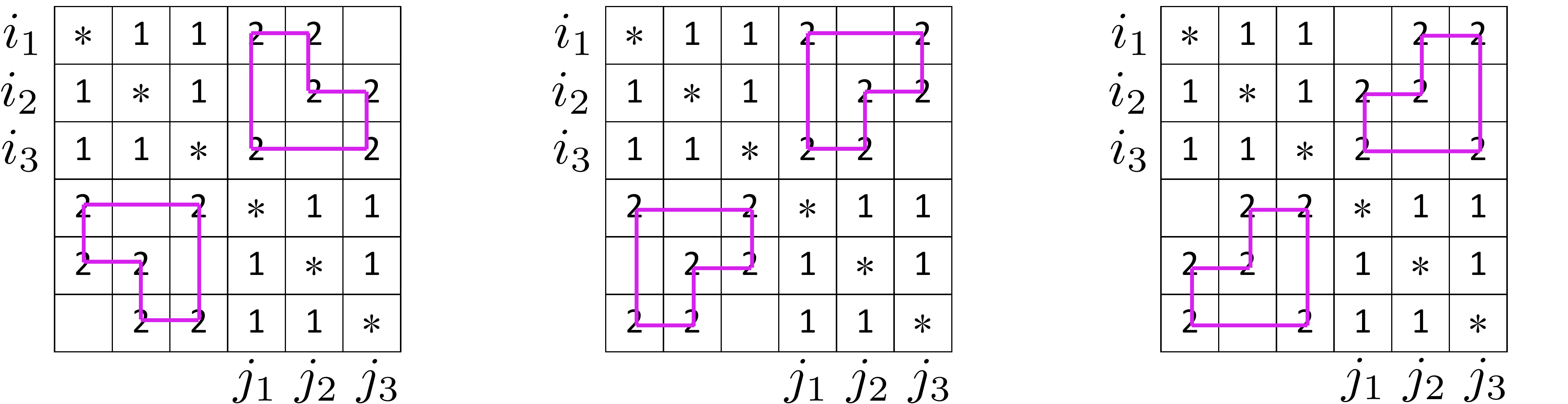}
\caption{Possible local matching graphs and their corresponding local cooperation matrices contained in a multi-level cooperation graph between $6$ nodes.}
\label{fig: matching}
\end{figure}

Observe that although in \Cref{exam: exam3}, the cooperation levels $({\ell}_{t;j})_{t\in\MB{T},j\in Y_t}$ specified for all the nodes on any cycle $C_t$, $t\in\MB{T}$, are identical, \textcolor{lara}{this case} is not a necessary condition. We present \textcolor{lara}{an} example in \Cref{exam: symcoop}, in which cooperation levels for nodes on the same cycle can be different. This example provides intuition both in deciding conditions that ensure a graph to be a cooperation graph, and in how to assign cooperation levels to such a cooperation graph if it exists.

In Subsection~\ref{subsec construction over compatible graphs}, we introduce the method of assigning cooperation levels over a given cooperation graph to obtain a so-called \textbf{compatible graph}. The algorithm to find a cooperation graph over a DSN $G(V,E)$ with a given topology is described in Subsection~\ref{subsec topology adaptivity}. 

\begin{figure}
\centering
\includegraphics[width=0.95\textwidth]{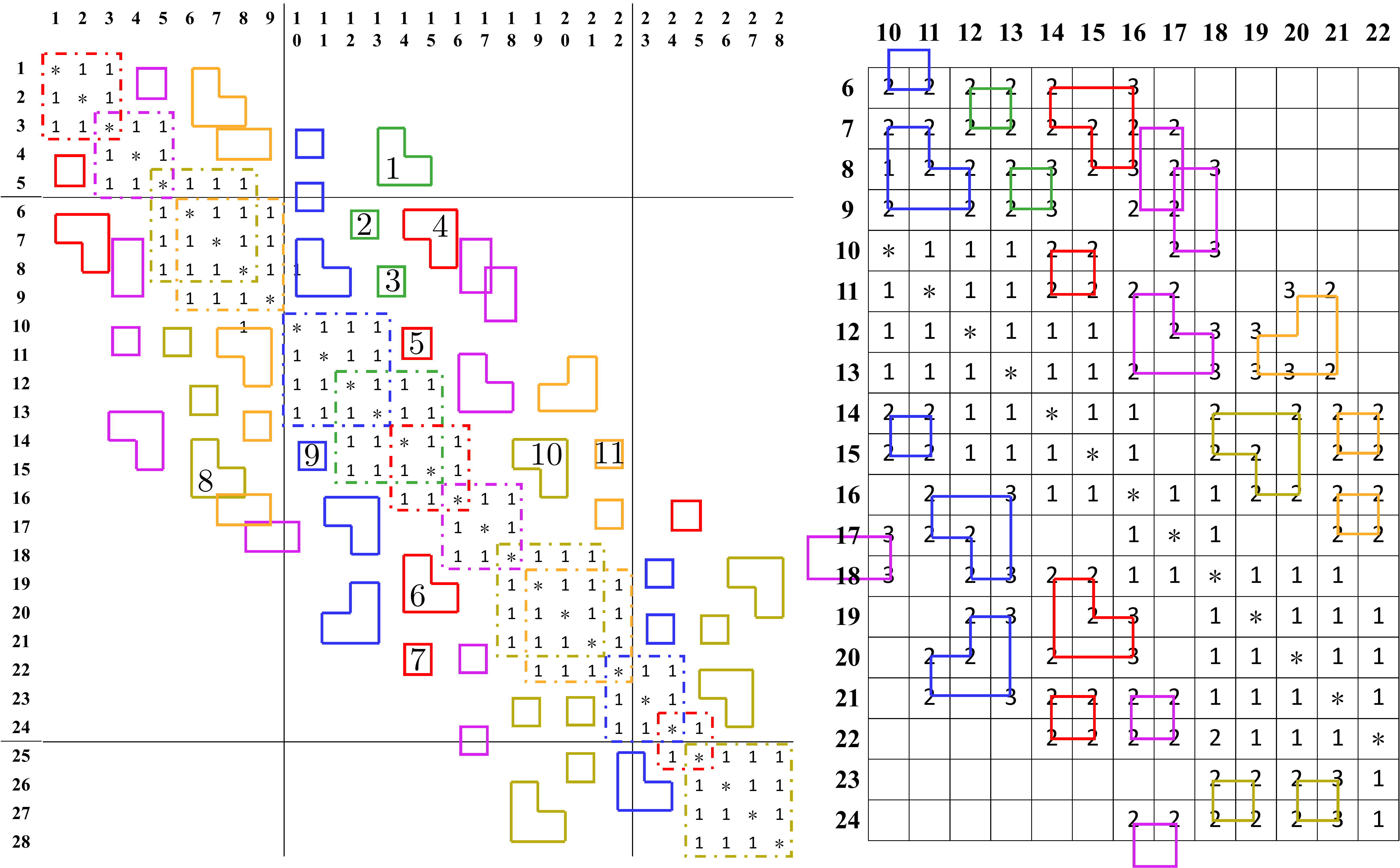}
\caption{Cycles and the assignment for part of the cooperation matrix on the compatible graph of a symmetric cooperation.}
\label{fig: symcompatiblegraph}
\end{figure}

\begin{exam} \label{exam: symcoop} The left panel of Fig.~\ref{fig: symcompatiblegraph} presents the cycle representation of a subgraph of a compatible graph on a DSN. Denote nodes associated with the left-most column to the right-most column by $v_1$ to $v_{28}$ in order, and let Fig.~\ref{fig: coopgraph} represent the subgraph containing nodes $\{v_i\}_{5\leq i\leq 22}$ of the cooperation graph. The right panel of Fig.~\ref{fig: symcompatiblegraph} denotes the cooperation matrix of $\{v_i\}_{10\leq i\leq 22}$. 

Note that some sub-matrices of the cooperation matrix are marked in dashed colored rectangles; these sub-matrices are all square matrices and have all non-diagonal entries being ones. For any  such rectangle, there are cycles marked in the same color (as the rectangle) that are totally contained within the columns spanned by this rectangle; these cycles are assigned a unique group number to form a group as specified in the previous subsection. In Fig.~\ref{fig: coopgraph}, instead of writing the group number assigned to each cycle, we mark the arrow connecting nodes representing the row and column indices of the cycle with a specified color for simplicity. Moreover, each one of those dashed rectangles corresponds to a maximum clique in $G(V,E)$ that denotes the DSN.

One can easily observe that cooperation levels assigned to entries in different columns within the same cycle are not always identical. Cooperation levels assigned to two nodes within the same column are identical if and only if they are on the same cycle or they are on different cycles from the same group.

Given all the aforementioned conditions, repeat steps 2)--5) specified in \Cref{exam: exam3} to obtain a generator matrix of a cooperative coding scheme on the DSN in this example (\Cref{exam: symcoop}). Then, for any node, the cross parities resulting from each cycle group can be derived from accessing other nodes in the maximum clique that contains this group. That is to say, by communicating with all the single-level neighbors, the cross parities for each cooperation level of any node $v_i\in V$ are computable and can then be subtracted from the parity part of codeword $\mathbf{c}_i$. 

Take node $v_{14}$ as an example. The column representing $v_{14}$ intersects with green cycles $C_1$ and $C_3$, and red cycles $C_4$, $C_5$, $C_6$, and $C_7$, corresponding to the $3$-rd and the $2$-nd level cooperation, respectively. \textcolor{lara}{All} the nodes in the green clique $\{v_{12},v_{13},v_{14},v_{15}\}$ except for $v_{14}$ itself are locally-recoverable. Thus, their cross parities, resulting from the cooperation the green cycles represent, are computable, and they sum up to the $3$-rd level cross parity of $\mathbf{c}_{14}$. Similarly, the $2$-nd level cross parity of $\mathbf{c}_{14}$ can also be derived if the other two nodes in the red clique $\{v_{14},v_{15},v_{16}\}$ are locally-recoverable.

More details, including the code construction, are given in Subsection~\ref{subsec construction over compatible graphs} and Subsection~\ref{subsec topology adaptivity}.
\end{exam}

\begin{figure}
\centering
\includegraphics[width=0.85\textwidth]{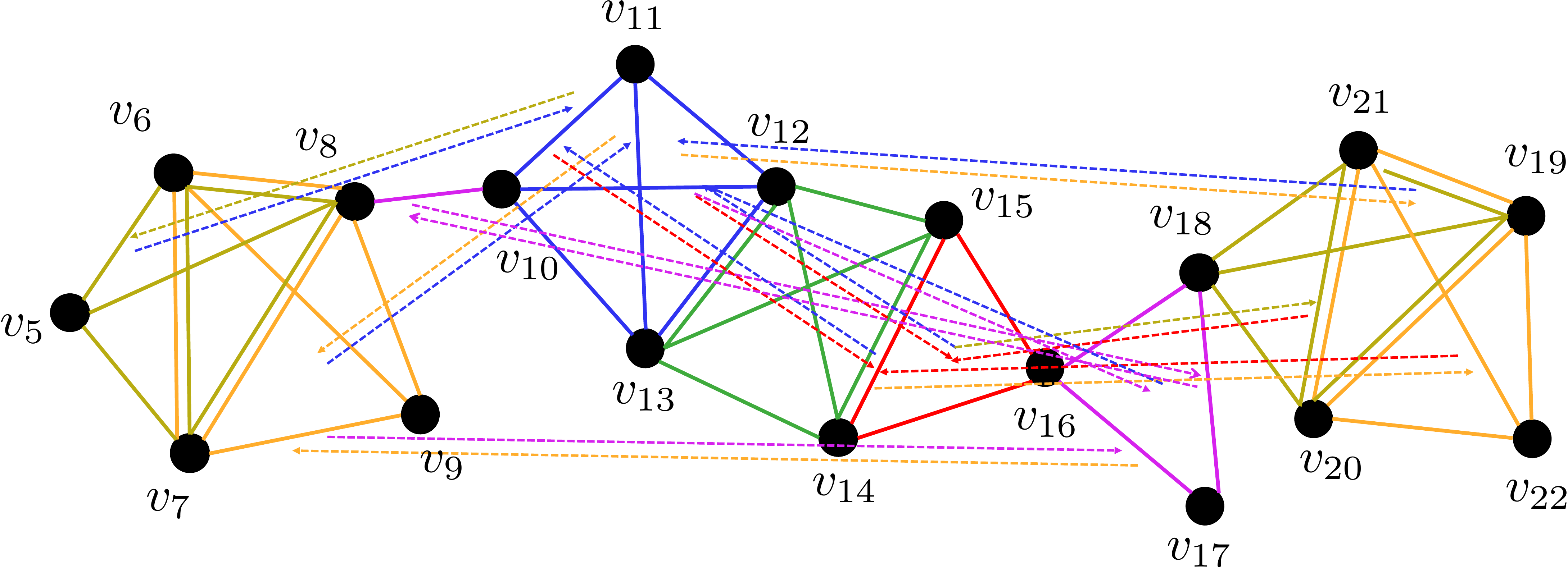}
\caption{Cooperation graph of \Cref{exam: symcoop}. For each node, all maximum cliques containing it are marked with different colors. For any dashed arrow pointing from an edge (or a triangle) to another edge (or another triangle), with the color identical to that of the maximum clique containing the latter one: it represents a cycle in the cooperation graph that is contained in the columns spanned by this maximum clique.}
\label{fig: coopgraph}
\end{figure}

Moreover, although \Cref{exam: exam3} has a topologically symmetric cooperation graph and also a topologically symmetric compatible graph, it is not necessary \textcolor{lara}{in principle} to constrain them to be symmetric. In the case where asymmetric cooperation is allowed, the basic components of the cooperation graph are edges instead of cycles, which allows more flexibility in choosing the cooperation graph. However, this asymmetry increases the complexity of defining the decoding graph for a node. Therefore, for simplicity, we only discuss topologically symmetric cooperations in this paper.

\subsection{Construction over Compatible Graphs}
\label{subsec construction over compatible graphs}

We have defined the notion of cooperation graphs in \Cref{subsec cooperation graphs}. Observe that the cooperation graphs shown in Fig.~\ref{fig: figMHP} and Fig.~\ref{fig: coopgraph} satisfy a set of conditions that define the so-called \textbf{compatible graph}. We show in \Cref{theo: ECcons2} the existence of a hierarchical coding scheme with cooperation graph $\Se{G}$ if $\Se{G}$ is a compatible graph. The coding scheme is presented in \Cref{cons: 2}.

\begin{defi}\label{defi: maximum clique} For any graph $G(V,E)$ with $|V|=p$, a subgraph $G'(V',E')$ is called a \textbf{maximum clique} of $G$ if any two nodes in $V'$ are connected, and there does not exist any node in $V\setminus V'$ that is connected to all nodes in $V'$. The set of all maximum cliques of $G$ is referred to as the \textbf{collection of maximum cliques} over $G$ and is denoted by $\Se{S}(V,E)$. Each maximum clique in $\Se{S}$ is represented by a subset $S$ of $\MB{p}$, where $S$ consists of the indices of all nodes in the maximum clique.
\end{defi}

\Cref{tab: notationtable} summarizes some notation associated with cooperation graphs that are used throughout the remainder of the paper. Take the DSN and its cooperation matrix shown in Fig.~\ref{fig: symcompatiblegraph} as an example. Observe that the green cycles in the columns spanned by the maximum clique $\{12,13,14,15\}$ are indexed by $1$, $2$, and $3$, and the red cycles in the columns spanned by the maximum clique $\{14,15,16\}$ are indexed by $4$, $5$, $6$, and $7$. These green and red cycles have group numbers $1$ and $2$, respectively. Suppose those cycles corresponding to the $2$-nd level cooperation of $v_{14}$ with top edges in the row representing $v_{14}$ are labeled with $8$, $9$, $10$, and $11$. Note that these cycles and cycles $C_4$, $C_5$, $C_6$, and $C_7$ are symmetric with respect to the diagonal. Then, $T_1=\{1,2,3\}$, $T_2=\{4,5,6,7\}$; $A_{14}=\{1,2\}$; $S(1)=\{12,13,14,15\}$, $S(2)=\{14,15,16\}$; $U_{14;1}=V_{14;3}=\{3,4,8,9\}$; $U_{14;2}=V_{14;2}=\{6,7,10,11,18,20,21,22\}$; $R_{14;2}=\{4,5,6,7\}$; $T_{14;2}=\{8,9,10,11\}$; ${\ell}_{5;14}=2$; $\Se{M}_{14}=\Se{N}_{14}=\{v_{12},v_{13},v_{15},v_{16}\}$. \Cref{defi: consistent graph} formally defines the sufficient conditions that result in a compatible graph, which were discussed informally in \Cref{exam: symcoop}. Note that this definition of compatible graphs is more general than that presented in the short version of the paper \cite{Yang2020DSN}, since the cooperation levels of nodes on different columns of cycles are allowed to be different here.

\begin{table}
\centering
\caption{Notation associated with cooperation graphs.}
\begin{tabular}{|c|c|}
\hline
Notation & Physical Interpretation\\
\hline
$A$ & The total number of different cycle groups\\
\hline
$T$ & The total number of different cycles\\
\hline
$T_{g}$ & The set consisting of indices of cycles with group number $g$\\
\hline
$A_i$ & The set of group numbers of those cycle groups that intersect with the column representing $v_i$\\
\hline
$S(g)$ & The set consisting of vertices of the maximum clique that contains all columns spanned by cycles in group $g$\\
\hline
$U_{j;g}$ & The intersection of all cycles contained in group $g$ with the column representing $v_j$\\
\hline
$R_{j;{\ell}}$ & The set consisting of indices of cycles that intersect with the column representing node $v_j$ at its $\ell$-th level cooperation \\
\hline
$V_{j;{\ell}}$ & The intersection of all cycles contained in $R_{j;{\ell}}$ with the column representing $v_j$\\
\hline
$T_{i;{\ell}}$ & The set consisting of indices of cycles that intersect with the row representing node $v_i$ at its $\ell$-th level cooperation\\
& (cycles with labels in $R_{j;{\ell}}$ and $T_{i;{\ell}}$ are symmetric with respective to the diagonal) \\
\hline
${\ell}_{t;j}$ & The cooperation level of node $v_j$ in cycle $C_t$\\
\hline
$\gamma_{i;t}$ & The maximum number of parity symbols that nodes in cycle $C_t$ can provide to node $v_i$\\
 \hline
 $\Se{M}_i$ & The set of nodes in the $1$-st level cooperation graph of $v_i$\\
 \hline
 $\Se{N}_i$ & The set of nodes in the neighborhood of $v_i$\\
\hline
\end{tabular}
\label{tab: notationtable}
\end{table}

\begin{defi}\label{defi: consistent graph} Let $\Se{G}$ be a cooperation graph on $G(V,E)$, where $\Se{G}$ is represented by $\{C_t=(X_t,Y_t,\{X_{t;j}\}_{j\in Y_t},\{Y_{t;i}\}_{i\in X_t},g_t,({\ell}_{t;j})_{j\in Y_{t}})\}_{t\in \MB{T}}$. Suppose node $v_i\in V$ has $L_i$ cooperation levels. We call $\Se{G}$ a \textbf{compatible graph} on $G$ if the following conditions are satisfied: 
\begin{enumerate}
\item For any $v_i\in V$, $\Se{M}_i\subseteq \Se{N}_i$.
\item All cycles $C_t$ with $t\in \MB{T}$ are disjoint.
\item For any $g\in\MB{A}$, there exists a maximum clique $S(g)\in\Se{S}(V,E)$ such that for all $t \in T_g$, $Y_t\subseteq S(g)$.
\item For each $v_i\in V$, ${\ell}\in\MB{L_i}$, there exists a unique $g\in A_i$, such that $V_{i;{\ell}}=U_{i;g}$; denote $g$ by $g(i;{\ell})$. 
\end{enumerate}

\end{defi}

\begin{cons}\label{cons: 2} Let $G(V,E)$ represent a DSN with parameters $(\bold{n},\bold{k},\bold{r})$. Suppose $\Se{G}$ is a compatible graph on $G$, with parameters $\{C_t=(X_t,Y_t,\{X_{t;j}\}_{j\in Y_t},\{Y_{t;i}\}_{i\in X_t},g_t,({\ell}_{t;j})_{j\in Y_{t}})\}_{t\in \MB{T}}$. Suppose node $v_i\in V$ has $L_i$ cooperation levels. 

Let $\boldsymbol{\delta}$ be the $1$-st level cooperation parameter. For any $v_i\in V$ and $g\in T_{i;{\ell}}$, assign a cooperation parameter $\gamma_{i;t}\in\Db{N}$ to the cooperation between node $v_i$ and nodes in $Y_{t;i}$. Let $\eta_{j;{\ell}}=\max\nolimits_{i\in V_{j;{\ell}}} \gamma_{i;t}$, for ${\ell}\in \MB{L_i}$.

Let $u_{i}=k_{i}+\delta_{i}+\sum\nolimits_{{\ell}=2}^{L_i}\eta_{i;{\ell}}$, $v_{i}=r_{i}+\sum\nolimits_{v_j\in \Se{M}_i}\delta_{j}+\sum\nolimits_{2\leq {\ell} \leq L_i,t\in T_{i;{\ell}}} \gamma_{i;t}$, for $i\in\MB{p}$. For each $i\in\MB{p}$, let $a_{i,s}$, $s\in \MB{u_{i}}$, and $b_{i,t}$, $t\in\MB{v_{i}}$, be distinct elements of $\textup{GF}(q)$, where $q\geq \max\nolimits_{i\in\MB{p}}\LB{u_i+v_i}$.

Matrix $\bold{G}$ in (\ref{eqn: GenMatDL}) is assembled as follows. Consider the Cauchy matrix $\bold{T}_{i}$ on $\textup{GF}(q)^{u_{i}\times v_{i}}$ such that $\bold{T}_{i}=\bold{Y}(a_{i,1},\dots,a_{i,u_{i}}; b_{i,1},\dots,b_{i,v_{i}})$, for $i\in\MB{p}$. Then, we obtain $\bold{A}_{i,i}$, $\bold{B}_{i,j}$, $\bold{E}_{i;{\ell}}$, $\bold{U}_{i}$, $\bold{V}_{i;{\ell}}$, for $i\in\MB{p}$, $j\in\MB{p}\setminus\LB{i}$, ${\ell}\in\MB{L_i}$, according to the following partition of $\bold{T}_{i}$:

\begin{equation}\label{eqn: CRSHL}
\bold{T}_{i}=\left[
\begin{array}{c|c}
\bold{A}_{i,i} & \begin{array}{c|c|c|c}
\bold{B}_{i} & \bold{E}_{i;2} & \dots & \bold{E}_{i;L_i}
\end{array}
\\
\hline
\begin{array}{c}\bold{U}_{i}\\
\hline 
\bold{V}_{i;2}\\
\hline
\vdots\\
\hline
\bold{V}_{i;L_i}\end{array}& \bold{Z}_{i}
\end{array}\right],
\end{equation} 
\begin{equation}
\textit{where } \text{ } \bold{B}_{i}=\left[\begin{array}{c|c|c}\bold{B}_{i,j_1} & \dots & \bold{B}_{i,j_{|\Se{M}_i|}}
\end{array}
\right],
\end{equation}
\begin{equation}
\textit{and  } \text{ } \bold{E}_{i;{\ell}}=\left[\begin{array}{c|c|c}\bold{E}_{i;{\ell};t_1} & \dots & \bold{E}_{i;{\ell};t_{|B_{i;{\ell}}|}}
\end{array}
\right], 
\end{equation}
such that $\Se{M}_i=\{v_{j_{1}},v_{j_{2}},\dots,v_{j_{|\Se{M}_i|}}\}$, $T_{i;{\ell}}=\{t_{1},t_{2},\dots,t_{|T_{i;{\ell}}|}\}$, $\bold{A}_{i,i}\in \textup{GF}(q)^{k_{i}\times r_{i}}$, $\bold{U}_{i}\in \textup{GF}(q)^{\delta_{i}\times r_i}$, $\bold{V}_{i;{\ell}}\in \textup{GF}(q)^{\eta_{i;{\ell}}\times r_i}$, $\bold{B}_{i,j}\in \textup{GF}(q)^{k_{i}\times \delta_{j}}$ for all $v_{j}\in\Se{M}^1_i$, and $\bold{E}_{i;{\ell};t}\in \textup{GF}(q)^{k_{i}\times \gamma_{i;t}}$. Let $\bold{B}_{i,j}=\left[ \bold{E}_{i;{\ell};t},\bold{0}_{k_i\times(\eta_{j;{\ell}}-\gamma_{i;t})}\right]$, and $\bold{A}_{i,j}=\bold{B}_{i,j}\bold{V}_{j;{\ell}}$, for all $j\in Y_{t;i}$, $t\in T_{i;{\ell}}$. Let $\bold{A}_{i,j}=\bold{B}_{i,j}\bold{U}_j$, for $v_j\in \Se{M}_i$; otherwise $\bold{A}_{i,j}=\bold{0}_{k_i\times r_i}$. Substitute the components of $\bold{G}$ in (\ref{eqn: GenMatDL}).

Let $\Se{C}_2$ represent the code with generator matrix $\bold{G}$.

\end{cons}

\begin{theo} \label{theo: ECcons2} The code $\Se{C}_2$ has EC hierarchies $\bold{d}_i=(d_{i,0},d_{i,1},\dots,d_{i,L_i})$, for all $v_i\in V$, where $d_{i,0}=r_i-\delta_i-\sum\nolimits_{{\ell}=2}^{L_i}\eta_{i;{\ell}}$, $d_{i,1}=r_i+\sum\nolimits_{v_j\in \Se{M}_i}\delta_j$, and $d_{i,\ell}=r_{i}+\sum\nolimits_{v_j\in \Se{M}_i}\delta_{j}+\sum\nolimits_{2\leq {\ell}' \leq {\ell} ,t\in T_{i;{\ell}'}} \gamma_{i;t}$. Moreover, $\Se{I}_i^1=\Se{M}_i$, $\Se{B}_i^1=\bigcup\nolimits_{v_j\in\Se{M}_i}\left(\Se{M}_j\setminus(\{v_i\}\cup \Se{M}_i)\right)$. For $2\leq {\ell} \leq L_i$, $\Se{I}^{\ell}_i=\bigcup\nolimits_{t\in R_{i;{\ell}}} \{v_j:j\in X_{t;i}\}{=\{v_j:j\in V_{i;{\ell}}\}}$, $\Se{B}_i^{\ell}=\bigcup\nolimits_{v_j\in\Se{I}^{\ell}_i} \left(\Se{I}^{\ell}_j\setminus (\{v_i\}\cup\Se{A}_i^{\ell})\right)$ (recall $\Se{A}_i^{\ell}=\bigcup\nolimits_{{\ell}'\leq l}\Se{I}^{{\ell}'}_i$), $\lambda_{i,{\ell};\Se{W}}=r_i+\sum\nolimits_{j:v_j\in \Se{M}_i,(\Se{M}_j\setminus\{v_i\})\subseteq (\Se{M}_i\cup\Se{W})}\delta_j+\sum\nolimits_{\substack{(i,t):2\leq {\ell}'\leq {\ell},t\in T_{i;{\ell}}, Y_{t;i}=\{j,j'\}, \\ \Se{I}_{j}^{{\ell}_{t;j}}\setminus\Se{A}_{i}^{{\ell}'}\subseteq (\{v_i\}\cup\Se{W})\allowbreak\text{ or }\Se{I}_{j'}^{{\ell}_{t;j'}}\setminus\Se{A}_{i}^{{\ell}'}\subseteq (\{v_i\}\cup\Se{W})}}\gamma_{i;t}$, $\varnothing\subseteq \Se{W}\subseteq \Se{B}_i^{\ell}$.
\end{theo}


\begin{proof}
For any node $v_j\in V$, denote the cross parities of $v_j$ due to cooperation with nodes in $\Se{I}_j^{\ell}$ by $\bold{s}_{j;{\ell}}$, {${\ell}\in \MB{L_j}$.} \textcolor{lara}{The cross parities are} given by the following equation:
\begin{equation}\label{eqn: crossparity}
\bold{s}_{j;{\ell}}=\begin{cases}
\sum\nolimits_{v_k\in \Se{M}_j} \bold{m}_k\bold{B}_{k,j}, &{\ell}=1,\\
\sum\nolimits_{k\in V_{j;{\ell}}} \bold{m}_k\bold{B}_{k,j}, &2\leq {\ell}\leq L_j.
\end{cases}
\end{equation} 
Thus, the codeword stored at $v_j\in V$ can be expressed in the following form:
\begin{equation}\label{eqn: codeword}
\begin{split}
\bold{c}_{j}&=\bold{m}_{j}\bold{A}_{j,j}+\sum_{v_k\in\Se{M}_j} \bold{m}_k\bold{B}_{k,j}\bold{U}_{j}+\sum_{{\ell}=2}^{L_j}\sum_{k\in V_{j;{\ell}}} \bold{m}_k\bold{B}_{k,j}\bold{V}_{j;{\ell}}\\
&=\bold{m}_{j}\bold{A}_{j,j}+\bold{s}_{j;1}\bold{U}_j+\sum\nolimits_{{\ell}=2}^{L_j}\bold{s}_{j;{\ell}}\bold{V}_{j;{\ell}}.\\
\end{split}
\end{equation}
Provided that the rows of $\bold{A}_{j,j}$, $\bold{U}_j$, and $\{\bold{V}_{j;{\ell}}\}_{{\ell}=2}^{L_j}$  are linearly independent, $\bold{m}_j$ and $\{\bold{s}_{j;{\ell}}\}_{{\ell}\in\MB{L_j}}$ are all computable if $\bold{c}_j$ is locally-recoverable.

We first show that by communicating with all the neighboring nodes in the $1$-st level cooperation, the cross parities {$\{\bold{s}_{i;\ell}\}_{\ell\in\MB{L_i}}$} of any node $v_i\in V$ can be computed and removed from the parity part of the codeword stored at this node if all its neighbors are locally-recoverable. {Under this condition on the neighbors of $v_i$, calculating $\bold{s}_{i;1}$ is trivial. Next, we prove for $2\leq\ell\leq {L_i}$ that $\{\bold{s}_{i;\ell}\}$ can also be computed.}

Condition 4) in \Cref{defi: consistent graph} indicates that there exists a unique $g=g(i;{\ell})\in A_i$, such that the following equation holds:
\begin{equation}\label{eqn: crossparitygroup}
\bold{s}_{i;{\ell}}=\sum\nolimits_{j\in U_{i;g}} \bold{m}_j\bold{B}_{j,i}.
\end{equation}
Moreover, Condition 3) guarantees the existence of a maximum clique $S(g)\in \Se{S}(V,E)$ such that for all $t\in T_g$, $Y_t\subseteq S(g)$. Let $\beta_{i;g}=\max_{\{(i,\ell):g(i;{\ell})=g\}} \eta_{i;{\ell}}$, and $\bold{u}_{i;{\ell}}=\MB{\bold{s}_{i;{\ell}},\bold{0}_{\beta_{i;g}-\eta_{i;{\ell}}}}$, for all $i\in \MB{p}$, ${\ell}\in\MB{L_i}$. We now consider:
\begin{equation*}
\begin{split}
\sum_{(i,\ell):g(i;{\ell})=g} \bold{u}_{i;{\ell}}&=\sum_{t\in T_{g},i\in X_t, j\in Y_{t;i}}\MB{\bold{m}_i\bold{B}_{i,j},\bold{0}_{\beta_{i;g}-\eta_{i;{\ell}_{t;i}}}}\\&=\sum_{t\in T_{g},i\in X_t, j\in Y_{t;i}}\MB{\bold{m}_i\bold{E}_{i;{\ell}_{t;i};t},\bold{0}_{\beta_{i;g}-\gamma_{i;t}}}\\&=\sum_{t\in T_{g},i\in X_t}\bold{0}_{\beta_{i;g}}=\bold{0}_{\beta_{i;g}}.\\
\end{split}
\end{equation*}
It follows that $\bold{s}_{i;{\ell}}$ can be derived from $\bold{u}_{i;{\ell}}$ if all $\bold{s}_{j;{\ell}'}$ such that $g(j;{\ell}')=g$ are known. Condition 3) in \Cref{defi: consistent graph} implies that all these $j$'s belong to $S(g)$, and the set of nodes indexed by $S(g)$ is a subset of $\Se{M}_i$, which means that all the aforementioned $\bold{s}_{j;{\ell}'}$'s are computable given that the neighbors of $v_i$ are locally-recoverable.

We have proved that all the $\ell$-th level cross parities, ${\ell}\in\MB{L_i}$, of any node $v_i\in V$ can be computed if the neighboring nodes are locally-recoverable. Now, we move forward to calculate the EC hierarchies of each node. Observe that the local and the $1$-st level cooperation erasure correction capabilities are proved the same way they are proved for Theorem~\ref{theo: ECcons1}. Thus, we only consider cases where $2\leq {\ell} \leq L_i$ in the following graph.

Moreover, for $2\leq {\ell} \leq L_i$, any cycle $C_t$ with index $t\in T_{i;{\ell}}$ has \textcolor{lara}{a} potential to provide additional $\gamma_{i;t}$ in \textcolor{lara}{the} $\ell$-th level cross parities to $v_i$.  Therefore, $d_{i,\ell}=r_{i}+\sum\nolimits_{v_j\in \Se{M}_i}\delta_{j}+\sum\nolimits_{2\leq {\ell}' \leq {\ell} ,t\in T_{i;{\ell}'}} \gamma_{i;t}$. The term $\gamma_{i;t}$ is added to the EC capability if any one of the two nodes $v_j$ and $v_{j'}$, $Y_{t;i}=\{j,j'\}$, obtains its $\gamma_{i;t}$ cross parities at $v_i$ through its ${\ell}_{t;j}$-th or ${\ell}_{t;j'}$-th level cooperation. Namely, any one of $\bold{m}_i\bold{B}_{i,j}$ and $\bold{m}_i\bold{B}_{i,j'}$ provides the value of $\bold{m}_i\bold{E}_{i;{\ell}_{t;i},t}$, and thus provides extra $\gamma_{i;t}$ parity symbols to $v_i$. Provided that $\bold{s}_{j;{\ell}_{t;j}}$ can be computed if $v_{j}$ is locally-recoverable, one needs to know the cross parities from all the nodes in the set $\Se{I}_{j}^{{\ell}_{t;j}}\setminus\{v_i\}$ to obtain those extra $\gamma_{i;t}$ parity symbols, i.e., those nodes are locally-recoverable, which means $\Se{I}_j^{{\ell}_{t;j}}\setminus \Se{A}_i^{\ell} \subseteq \{v_i\}\cup\Se{W}$. Similarly, the condition of $v_{j'}$ successfully calculating these $\gamma_{i;t}$ cross parities at $v_i$ is described as $\Se{I}_{j'}^{{\ell}_{t;j'}}\setminus \Se{A}_i^{\ell} \subseteq \{v_i\}\cup\Se{W}$. Therefore, the overall requirement is stated as ``$\Se{I}_j^{{\ell}_{t;j}}\setminus \Se{A}_i^{\ell} \subseteq \{v_i\}\cup\Se{W}$ or $\Se{I}_{j'}^{{\ell}_{t;j'}}\setminus \Se{A}_i^{\ell} \subseteq \{v_i\}\cup\Se{W}$''. From this discussion, we reach that $\Se{B}_i^{\ell}=\bigcup\nolimits_{t\in T_{i;{\ell}},j\in Y_{t;i}}(\Se{I}^{{\ell}_{t;j}}_j\setminus (\{v_i\}\cup\Se{A}_i^{\ell}))$ and $\lambda_{i,{\ell};\Se{W}}=r_i+\sum\nolimits_{j:v_j\in \Se{M}_i,(\Se{M}_j\setminus\{v_i\})\subseteq (\Se{M}_i\cup\Se{W})}\delta_j+\sum\nolimits_{\substack{(i,t):2\leq {\ell}'\leq {\ell},t\in T_{i;{\ell}'}, Y_{t;i}=\{j,j'\}, \\\Se{I}_{j}^{{\ell}_{t;j}}\setminus\Se{A}_{i}^{{\ell}'}\subseteq (\{v_i\}\cup\Se{W})\allowbreak\text{ or }\Se{I}_{j'}^{{\ell}_{t;j'}}\setminus\Se{A}_{i}^{{\ell}'}\subseteq (\{v_i\}\cup\Se{W})}}\gamma_{i;t}$, for $\varnothing\subseteq \Se{W}\subseteq \Se{B}_i^{\ell}$. \end{proof}


Note that although in a DSN represented by $G(V,E)$, node $v_i\in V$ cooperates with all nodes in $\Se{I}_i^{\ell}$ in the $\ell$-th level cooperation, $2\leq {\ell} \leq L_i$, it is not necessary that all codewords stored in nodes from $\Se{I}_i^{\ell}$ need to be recovered. The reason is that these nodes are partitioned into node pairs where the two nodes in the same pair provide exactly the same group of parity symbols and only one of them needs to be recovered for node $v_i$ to recover its codeword, as we discussed in \Cref{theo: ECcons2}. 

For example, suppose Fig.~\ref{fig: matching} corresponds to a subgraph of a DSN with a recoverable erasure pattern. The pink triangles and the dashed arrows represent the $\ell$-th level cooperation at each node such that no other nodes are involved in the $\ell$-th level cooperation of these nodes, i.e., $C_t$ is the only cycle in the cycle group containing it. As discussed in \Cref{theo: ECcons2}, to remove the local cross parities of each node, neighbors of any non-locally-recoverable node should all be locally-recoverable. Therefore, there exists at most one non-locally-recoverable node in each one of the two triangles. For any $i\in X_t$, previous conditions indicate that at least one of the two nodes with indices in $Y_{t;i}$ is locally-recoverable; let it be node $v_j$, where $j\in Y_{t;i}$. We know that $Y_{t;j}$ consists of $i$ and $i'$ for some $i'\in X_t$ and $i'\neq i$. Since the codeword stored at $v_j$ is locally-recoverable and $C_t$ forms an isolated cycle group, the cross parity $\bold{m}_{i}\bold{B}_{i,j}+\bold{m}_{i'}\bold{B}_{i',j}$ can be derived at node $v_j$. Since the codeword at $v_{i'}$ is locally-recoverable, $\bold{m}_{i'}\bold{B}_{i',j}$ can be further subtracted from the cross parity to obtain $\bold{m}_i\bold{B}_{i,j}$. \textcolor{lara}{This observation} indicates that regardless of the way the matching graph is specified and the indices of the nodes that are not recovered, the non-recovered nodes are able to obtain their $\ell$-th level cross parities. For example, suppose then $v_{i_1}$ and $v_{j_1}$ in Fig.~\ref{fig: matching} are not locally-recoverable. Since $V_{i_1;{\ell}}=\{j_1,j_3\}$ and $\Se{I}^{\ell}_i=\{v_{j_1},v_{j_3}\}$, the aforementioned discussion demonstrates that recovering $v_{i_1}$ does not require $v_{j_1}$ to be recovered.

In Fig.~\ref{fig: matching}, the additional erasure correction capabilities offered to the nodes are identical regardless of the structure of the local matching graph. However, this may not be true in general, if more than one cycle gets involved. In general, different local matching graphs are likely to result in non-identical erasure correction capabilities. In particular, although the EC hierarchies are defined by $(\lambda_{i,l;\Se{W}})_{\varnothing \subset \Se{W} \subset\Se{B}_i^{\ell}}$ for each individual node at each cooperation level, this can be more elaborately defined since accessing a different subset of nodes in $\Se{I}_i^l$ may result in different EC capabilities even if $\Se{W}$ are the same, and $\lambda_{i,l;\Se{W}}$ only specifies the largest one. We show it in details by \Cref{exam: examlocal}.

Given an isolated matching graph $G'$, if none of the non-locally-recoverable nodes are able to derive any additional cross parities solely from the local cooperation specified by $G'$, i.e., all cycles in $G'$ constitute a cycle group, we call it an \textbf{absorbing matching graph}.

In \Cref{exam: examlocal}, EC capabilities of compatible graphs resulting from the same cooperation graph associated with different local matching graphs, as shown in Fig.~\ref{fig: localmatchinggraph}, are discussed. We prove that the left two panels are two absorbing matching graphs, while the right two panels are not, which also demonstrates that the EC capability is not uniquely determined by the cooperation graph. Instead, the local matching graph also matters.

We focus on the local cooperation graph between $9$ nodes in a DSN $G(V,E)$, where $V=\{v_i\}_{i\in\MB{9}}$. Let $V_1=\{v_1,v_4,v_7\}$, $V_2=\{v_2,v_5,v_8\}$, and $V_3=\{v_3,v_6,v_9\}$, and suppose nodes in each one of these sets mutually cooperate with nodes in each of the remaining two sets. According to the definition of cooperation graphs, each one of the graphs represents $3$ cycles, $\{C_i\}_{i\in\MB{3}}$, where $X_i=\{3j+i\}_{0\leq j\leq 2} $, $Y_i=X_{i+1}$, and $X_4=X_1$. Suppose $\{C_i\}_{i\in\MB{3}}$ form an isolated cycle group in the cooperation graph on $G$. Represent each one of the cycles by a specified matching graph and refer to the resulting local matching graph as an isolated local matching graph. Fig.~\ref{fig: localmatchinggraph} presents four different isolated local matching graphs on these nodes. Let colors blue and black refer to nodes that are non-locally-recoverable and locally-recoverable nodes, respectively. Then, each $V_i$, $i\in\MB{3}$, contains at most one blue node if it is contained in a recoverable erasure pattern. 

\begin{exam}\label{exam: examlocal}
\begin{figure}
\centering
\includegraphics[width=0.95\textwidth]{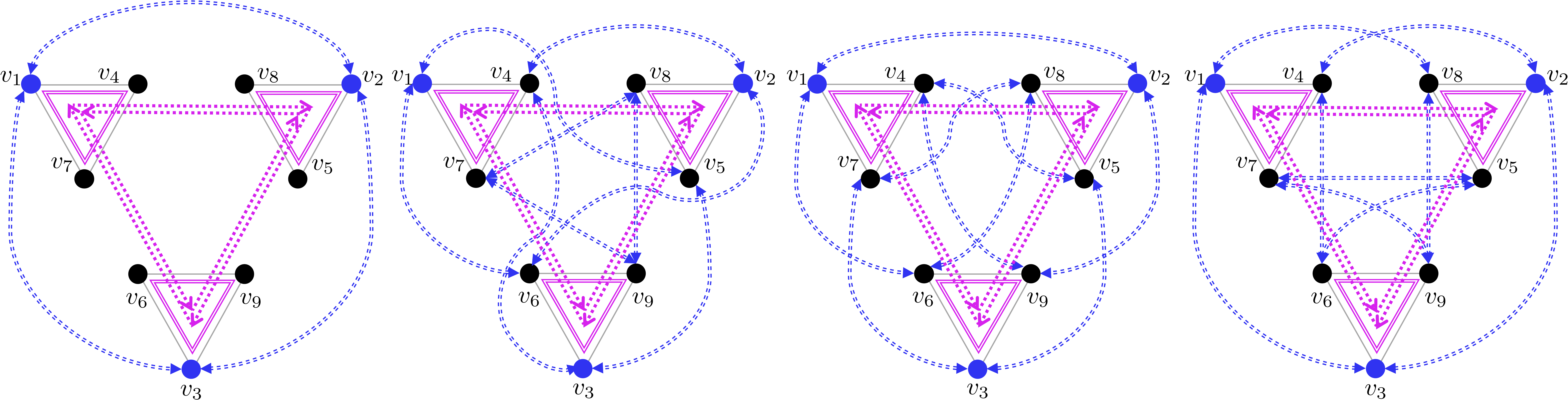}
\caption{Possible local matching graphs contained in a multi-level cooperation graph between $9$ nodes. The nodes are partitioned into three groups, where nodes within each one of them are pairwisely connected. Each dashed double-sided arrow represent a cycle.}
\label{fig: localmatchinggraph}
\end{figure} 

\textbf{\emph{(Absorbing Matching Graphs)}} 
Consider the left-most local matching graph in Fig.~\ref{fig: localmatchinggraph}, we notice that $v_1$, $v_2$, $v_3$ are mutually connected (and we will show that the connections between the rest of the nodes actually do not matter, thus we omit them in the figure). Without loss of generality, it is sufficient to prove that $v_1$ is not able to obtain any extra parity symbols from cycle $C_1$. Since $v_1$ and $v_2$ are connected, we know that $Y_{1;1}=\{5,8\}$. Thus, $v_1$ needs to either obtain $\bold{m}_1\bold{B}_{1,5}$ from $v_5$ or $\bold{m}_1\bold{B}_{1,8}$ from $v_8$. Since $v_2$ and $v_3$ are connected, both $v_5$ and $v_8$ have cross parities at node $v_3$. Therefore, the parities $v_i$, $i\in\{5,8\}$, are of the form $\bold{s}_i=\bold{m}_1\bold{B}_{1,i}+\bold{m}_3\bold{B}_{3,i}+\bold{m}_{j}\bold{B}_{j,i}+\bold{m}_{j'}\bold{B}_{j',i}$, where $j\in\{4,7\}$, $j'\in\{6,9\}$. If codewords stored at node $v_j$ and $v_{j'}$ are all locally-recoverable, their parities can be subtracted from $\bold{s}_i$ to obtain the remainder $\bold{s}'_i=\bold{m}_1\bold{B}_{1,i}+\bold{m}_3\bold{B}_{3,i}$. Observe that in order to obtain $\bold{m}_1\bold{B}_{1,i}$, $\bold{m}_3\bold{B}_{3,i}$ needs to be obtained first. That is to say, $v_1$ can only obtain additional parities unless the message in $v_3$ is recovered. However, following a similar process, we will need $v_2$ to be recovered for $v_3$, and $v_1$ to be recovered for $v_2$. This cyclic requirement indicates that $v_1$, $v_2$, and $v_3$ are ``absorbed'' into a balanced situation where none of them can be recovered first, which cannot be broken unless information from the rest of the graphs \textcolor{lara}{is} provided.

Similarly, we can prove that the second-to-the-left panel is also an absorbing matching graph. Since the connections between the three triangles are symmetric, it is still sufficient to prove that $v_1$ is not able to obtain any extra parity symbols from cycle $C_1$. Since $v_1$ and $v_5$ are connected, we know that $Y_{1;1}=\{2,8\}$. Given that $v_2$ is not locally-recoverable, the only path for $v_1$ to obtain extra parities is to obtain $\bold{m}_1\bold{B}_{1,8}$ from $v_8$. Observe that $v_8$ is connected to $v_7$ and $v_9$ in the matching graph, which means that $v_8$ has cross parities at $v_1$, $v_3$, $v_4$, and $v_6$. Therefore, $v_8$ needs $v_3$, $v_4$, and $v_6$ to be all recovered in order to subtract $\bold{m}_3\bold{B}_{3,8}$, $\bold{m}_4\bold{B}_{4,8}$ and $\bold{m}_6\bold{B}_{6,8}$ from $\bold{m}_3\bold{B}_{3,8}+\bold{m}_4\bold{B}_{4,8}+\bold{m}_6\bold{B}_{6,8}$ to obtain $\bold{m}_1\bold{B}_{1,8}$. This requires $v_3$ to be recovered. Following a similar argument to that of the left-most panel, this graph is also an absorbing matching graph.

Moving on to the second-from-the-right panel, $v_5$ has cross parities at $v_7$, $v_6$, $v_9$, and $v_1$. Therefore, $v_1$ is able to obtain the cross parities resulting from cooperation cycle $C_3$ through $v_5$. Similarly, $v_2$ is able to obtain the cross parities resulting from cooperation cycle $C_2$ through $v_4$. Finally, $v_3$ is able to obtain its cross parities from $C_3$ and $C_2$. 

Now we look at the right-most panel, $v_4$ has cross parities at $v_5$, $v_8$, $v_9$, and $v_3$. Therefore, $v_3$ is able to obtain the cross parities resulting from cooperation cycle $C_3$ through $v_4$. Similarly, $v_3$ is also able to obtain the cross parities resulting from cooperation cycle $C_2$ through $v_5$. After that, one of $v_1$ and $v_2$ is able to obtain additional parities from $C_1$, and the other one can obtain additional parities from both $C_1$ and its cooperation with $v_6$ and $v_9$. 

Note that the major difference between the second-to-the-right and the right-most panel is that the decoding of $v_1$ provides no additional parities on $v_2$ and $v_3$ in the third one. Therefore, if each cycle $C_i$, $i\in\MB{3}$, provides $\eta$ cross parities at each of its nodes, the third and the fourth graphs allow up to additional $4\eta$ and $5\eta$ cross parities, respectively.
\end{exam}

We have shown in the previous example that the left-most two panels in \Cref{fig: matching} are absorbing matching graphs, and they become non-absorbing matching graphs if any of the blue nodes turns to be recovered from cooperation with the rest of the graph. In this case, without loss of generality, suppose $v_3$ is recovered, then $v_1$ in these two graphs is also recoverable according to discussion in \Cref{exam: examlocal}. While in the left-most panel, nodes in $\Se{I}_1^l$ are indeed all locally recovered, those in the second graph from the left are not. We also notice that in the right-most panel, nodes in $\Se{I}_1^l$ are locally recovered to recover $v_1$, while those in the second-to-the right panel are not. Given that any isolated matching graph corresponding to a $\lambda_{1,l;\Se{W}}$ with $\Se{W}=\varnothing$, this example also demonstrates that a different set $\Se{A}_i^l$ will also provide different EC capabilities. While in \Cref{exam: examlocal} we already subtly discussed such a scenario, we leave more detailed analysis for future work.

Moreover, we state without proof here that the left-most two panels of Fig.~\ref{fig: localmatchinggraph} discussed in \Cref{exam: examlocal} are all the possible structures of an absorbing matching graph for this specified local cooperation graph (subject to the graph isomorphism). Since these graphs also are all the possible structures where the nine matching edges form disconnected cycles, those edges in other matching graphs all form a cycle of length $9$ and are mutually isomorphic according to permutations of $v_1$ to $v_9$. However, different permutations of the nodes do result in different erasure correction capabilities. For example, the second-to-the-right panel is isomorphic to the right-most panel if $v_3$, $v_5$, $v_7$ are blue instead, as shown in \Cref{fig: information coupling}. This has no impact on the average erasure correction capability while looking into the local matching graphs individually, but the permutation matters while taking the connection to the rest of the graphs into consideration.

In \Cref{rem: infoflow}, we discuss the information flow between neighboring nodes, and the information flow between nodes with distance two through their common neighbors. Observe that nodes cooperating with any given node in its higher-level cooperations are not necessarily all within its two-hop neighborhood. However, these nodes actually provide additional parities to the original nodes. This scenario is not covered by the previous definition of information flow, instead of it, we proposed the notion of \textbf{information coupling} to describe it, as discussed in \Cref{rem: infocoupling}.

\begin{rem} \emph{\textbf{(Information Coupling in Multi-Level Coded DSN)}} \label{rem: infocoupling}
Take the local matching graph shown in the right-most panel in Fig.~\ref{fig: matching} as an example. Consider the case where non-locally-recoverable nodes are $v_1$, $v_2$ and $v_3$, as shown in the left panel in Fig.~\ref{fig: information coupling}. Node $v_3$ is able to obtain additional parities from $v_4$ since $v_5$, $v_8$ and $v_9$ are locally-recoverable. This \textcolor{lara}{case} can be regarded as information flow from the cooperation between $v_4$, $v_5$ and $v_8$ to the cooperation between $v_4$, $v_3$ and $v_9$ through $v_4$. 

Consider another case where the non-locally-recoverable nodes are $v_3$, $v_5$ and $v_7$ instead, as shown in the right panel in Fig.~\ref{fig: information coupling}. Node $v_3$ is no longer able to decode its codeword first. Instead, node $v_7$ is able to obtain additional parities from $v_2$ since $v_1$, $v_6$ and $v_9$ are locally-recoverable. This \textcolor{lara}{case} can be regarded as information flow from the cooperation between $v_1$, $v_2$, and $v_7$ to the cooperation between $v_2$, $v_6$, and $v_9$ through $v_2$. 

The aforementioned cases indicate that for any node, nodes cooperating with it in its higher-level cooperation do not have impact on it individually, \textcolor{lara}{but rather} collectively. Moreover, as discussed in \Cref{exam: examlocal}, this impact is not only dependent on the local matching graphs, but also dependent on the erasure patterns. Therefore, instead of discussing information flow between two cycles, it is more appropriate to treat all the cycles contained in any cycle group collaboratively. This can be interpreted as information coupling resulted from the cooperation between $\{v_1,v_4,v_7\}$, $\{v_2,v_5,v_8\}$, and $\{v_3,v_6,v_9\}$, as an analogy to information coupling in network navigation. 
\end{rem}

\begin{figure}
\centering
\includegraphics[width=0.5\textwidth]{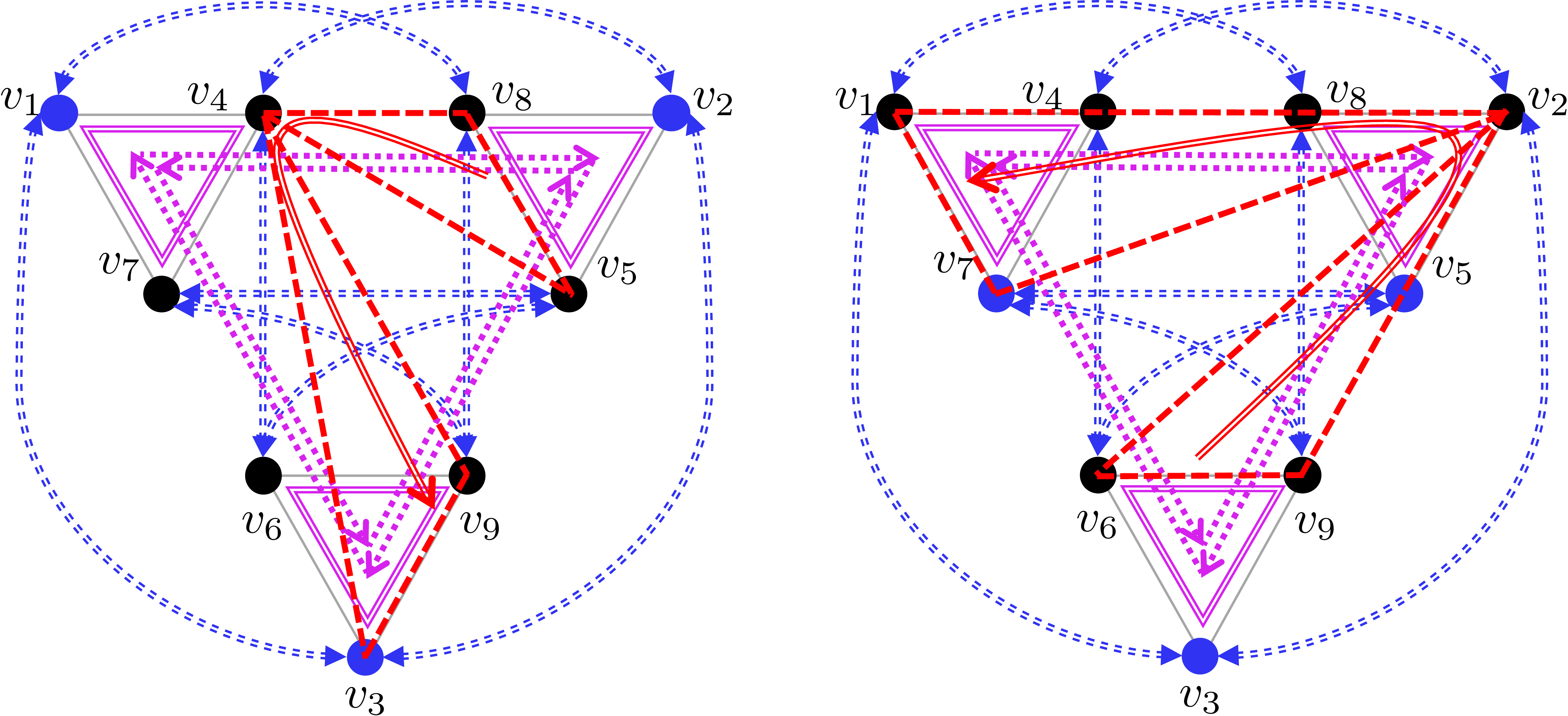}
\caption{Information coupling. The two graphs represent two different erasure patterns for the non-absorbing local matching graphs in Fig.~\ref{fig: matching}. While the information flow depicts the communication between any two nodes separated by a distance of $1$ or $2$ in their $1$-st level cooperation, it is not able to fully describe higher level cooperations. Two cycles in a local matching graph also help the decoding of nodes on each other through their shared nodes, and we call this ``information coupling''.}
\label{fig: information coupling}
\end{figure}

\subsection{Recoverable Erasure Patterns}
\label{subsec recoverable erasure patterns 2}
Recall the notion of ``decoding graph'' in \Cref{defi decoding graph}, under which recoverable erasure patterns of the single-level cooperative codes are described. However, in cases where higher-level cooperations are involved, \Cref{defi decoding graph} is not enough to define and enumerate all associated recoverable erasure patterns. In this section, we extend \Cref{defi decoding graph} into \Cref{defi HL decoding graph} to \textcolor{lara}{allow for the} multi-level cooperation. Recoverable erasure patterns of hierarchical codes are specified in \Cref{theo:erausre pattern hierarhical }.

\begin{defi} \label{defi HL decoding graph} \emph{\textbf{(Decoding Graph in Multi-Level Cooperation)}} Let $G(V,E)$ represent a DSN with $|V|=p$. Let $\mathcal{T}(\mathcal{V},\mathcal{E})$ denote a directed subgraph of $G$ associated with $v_j\in \Se{V}$. For all $v_i\in\mathcal{V}$, denote the set containing the children of $v_i$ by $\Se{V}_i^{\textup{C}}$, and the set containing all parents of $v_i$ by $\Se{V}_i^{\textup{P}}$. Suppose $v_j$ is the only node without parents, we call it the \textbf{root} of $\Se{T}$. We call any node without children a \textbf{leaf}. Suppose that all the leaves of $\Se{T}$ are not locally-recoverable, and any other $v_i\in \mathcal{V}$ satisfies either one of the following conditions.
\begin{enumerate}
\item The codeword stored at $v_i$ is locally-recoverable: there exists a set $\Se{L}\subseteq\{2,\dots,L_i\}$, with $|\Se{V}_i^{\textup{P}}\cap \Se{M}_i|\in\{0,1\}$, $|\Se{V}_i^{\textup{P}}\cap V_{i;{\ell}}|=1$, and $\Se{V}_i^{\textup{P}}\cup\Se{V}_i^{\textup{C}}$, that consists of all nodes with indices in $V_{i;{\ell}}$ for ${\ell}\in \Se{L}$ (and $\Se{M}_i$ if $|\Se{V}_i^{\textup{P}}\cap \Se{M}_i|=1$), where codewords stored at them are not locally-recoverable.
\item The codeword stored at $v_i$ is not locally-recoverable: codewords stored at nodes from $\Se{V}_i^{\textup{P}}\cup\Se{V}_i^{\textup{C}}$ are locally-recoverable.
\end{enumerate}
We call $\Se{T}$ a decoding graph at its root node $v_j$ over $G(V,E)$.
\end{defi}


\begin{theo} \emph{\textbf{(Flexible Erasure Patterns)}} \label{theo:erausre pattern hierarhical } Let $\Se{C}$ be a code with hierarchical cooperation on a DSN represented by $G(V,E)$, where $\Se{C}$ and all related parameters are specified according to \Cref{cons: 2}. Let $\bold{u}\in\mathbb{N}^p$ such that $\bold{u}\preceq\bold{n}$. Suppose $\Se{C}$ and $\bold{u}$ satisfy the following conditions:
\begin{enumerate}
\item Let $V^{\textup{NL}}$ represent the set contains all the nodes $v_i$, $i\in\MB{p}$ such that $u_i>r_i-\delta_i$. Let $V^{\textup{L}}=V\setminus V^{\textup{NL}}$. Then, for any $v_i\in V^{\textup{NL}}$, $\Se{M}_i\subset V^{\textup{L}}$.
\item For any $v_i\in V^{\textup{NL}}$, there exists a decoding graph $\mathcal{T}_i(\Se{V}_i,\Se{E}_i)$ at root $v_i$ over $G$. Moreover, for any leaf $v_j$ of $\mathcal{T}_i$, $u_j\leq r_j$; for any node $v_j\in \Se{V}_i\cap V^{\textup{NL}}$, $u_j\leq r_j+\sum\nolimits_{v_k\in\Se{M}_i\cap\Se{V}_j^{\textup{C}}}\delta_k+\sum\nolimits_{{\ell}=2}^{L_i}\sum\nolimits_{k\in V_{j;{\ell}}\cap X_{t;j},v_k\in\Se{V}_j^{\textup{C}}}\gamma_{k;t}$. 
\end{enumerate} 
Then, $\bold{u}$ is a \emph{\textbf{recoverable erasure pattern}} of $\Se{C}$ over $G(V,E)$.
\end{theo}

\textcolor{lara}{Consider} the DSN with the $1$-st level cooperation graph presented in \Cref{exam11}. \textcolor{lara}{We} add the $2$-nd level and the $3$-rd level cooperation graphs to the DSN and mark them in pink and olive, respectively, as shown in \Cref{exam: exam21} and \Cref{fig: DSNexam22}. Black and blue/green still refer to nodes where the stored codewords are locally-recoverable and non-locally-recoverable, respectively. Components marked in red represent local decoding graphs, which is the subgraph of the decoding graph corresponding to the local matching graph.

\begin{exam}\label{exam: exam21}
\begin{figure}
\centering
\includegraphics[width=0.47\textwidth]{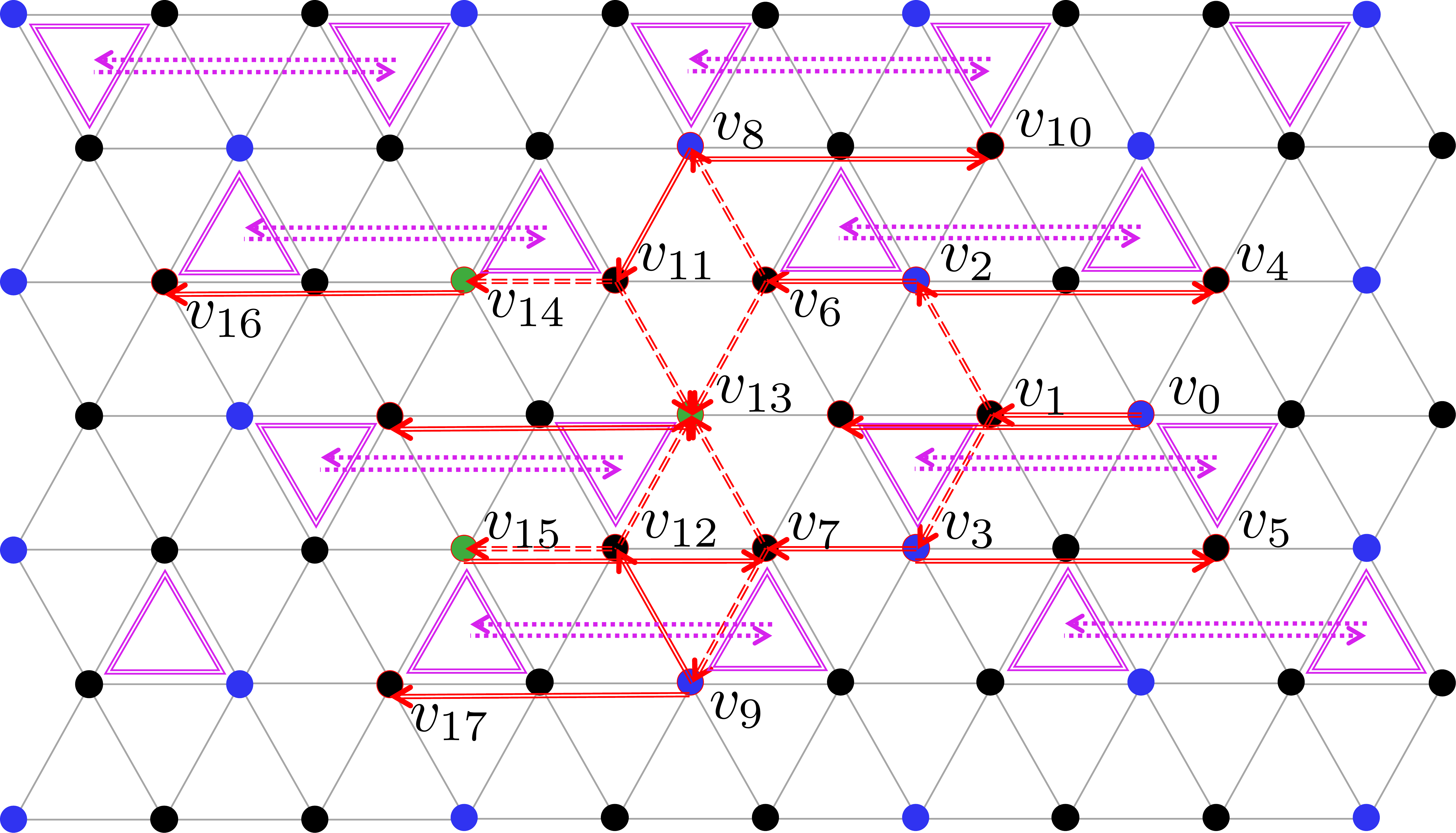}
\includegraphics[width=0.49\textwidth]{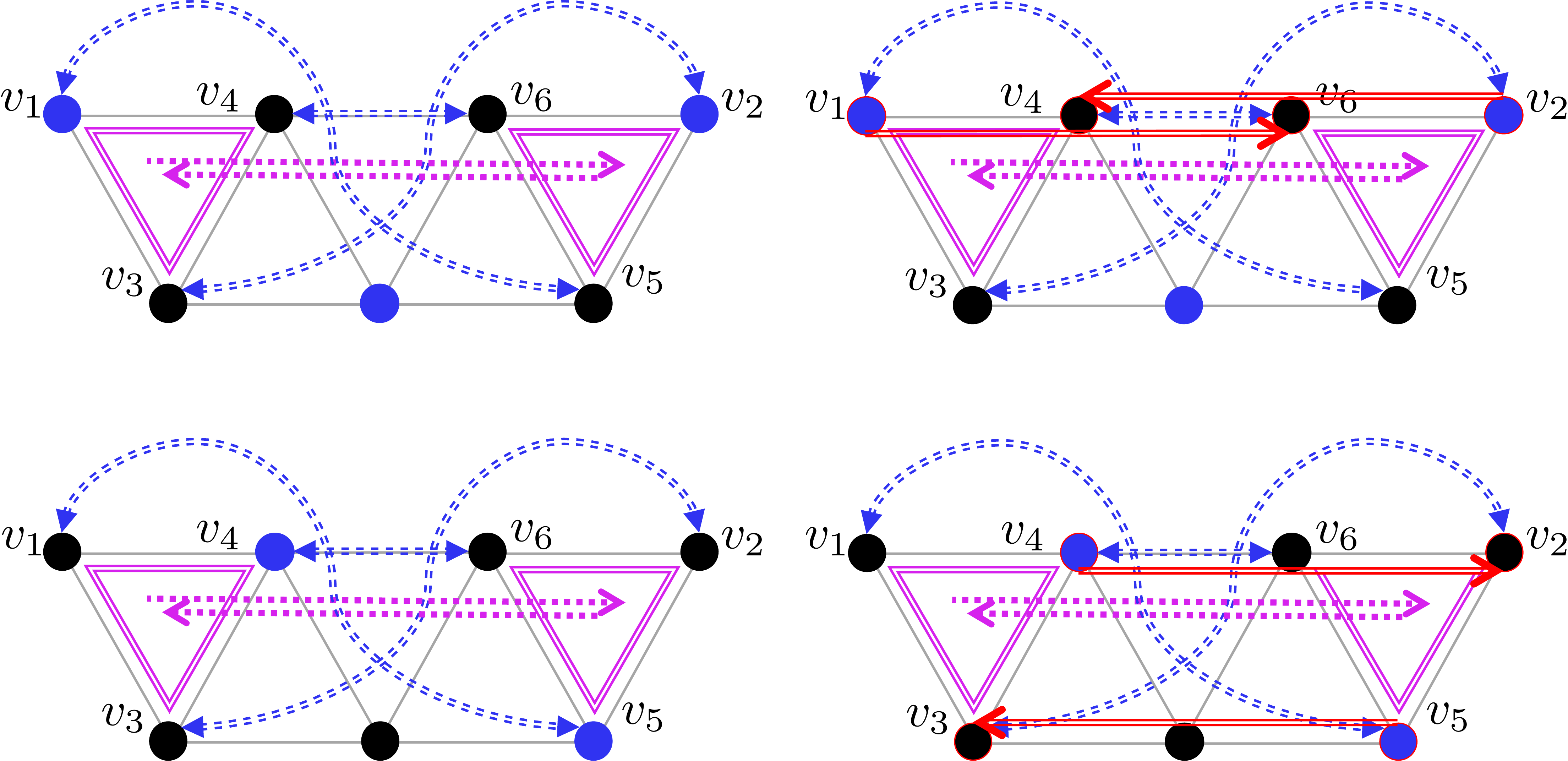}
\caption{DSN (left-most) and the local matching graph specified for $2$-nd level cooperation in \Cref{exam: exam21}.}
\label{fig: DSNexam21}
\end{figure}

Fig.~\ref{fig: DSNexam21} has five graphs. The left-most panel describes the cooperation graph resulting from adding the $2$-nd level cooperation among nodes to the DSN in \Cref{exam11}, where there exist two possible local matching graphs that are specified by the two graphs in the center (in the central panel). The right-most two panels present the subgraphs in local decoding graphs corresponding to the two possible local matching graphs. Let $\bold{u}=\SB{u_{1},u_{2},\dots,u_{p}}$ be an erasure pattern on this DSN. Under the EC solution specified in \Cref{theo: ECcons2}, suppose there exists $\gamma\in\Db{N}$ such that $\gamma_{i;t}=\gamma$, for all $t\in\MB{T}$, and $i\in X_t$. 

In the specified cooperation graph, for $i\in\MB{p}$, if $ v_i$ is black, then $0\leq u_{i}\leq r_i-\delta-\gamma$; else if $v_i$ is green, then $r_i-\delta-\gamma\leq u_i\leq r_i+\gamma$; otherwise $r_i+\gamma<u_{1,i}\leq r_i+\delta+\gamma$. Since each blue node is contained in an isolated local matching graph, it can obtain additional $\gamma$ cross parity symbols from its $2$-nd level cooperation according to the previous discussion about Fig.~\ref{fig: matching}. Therefore, all the non-locally-recoverable nodes are able to tolerate extra $\gamma$ erasures, which means that $\bold{u}$ is a recoverable erasure pattern of this graph but not a recoverable pattern of the left panel in Fig.~\ref{fig11}.

\end{exam}

\begin{exam}\label{exam: exam22}
\begin{figure}
\centering
\includegraphics[width=0.47\textwidth]{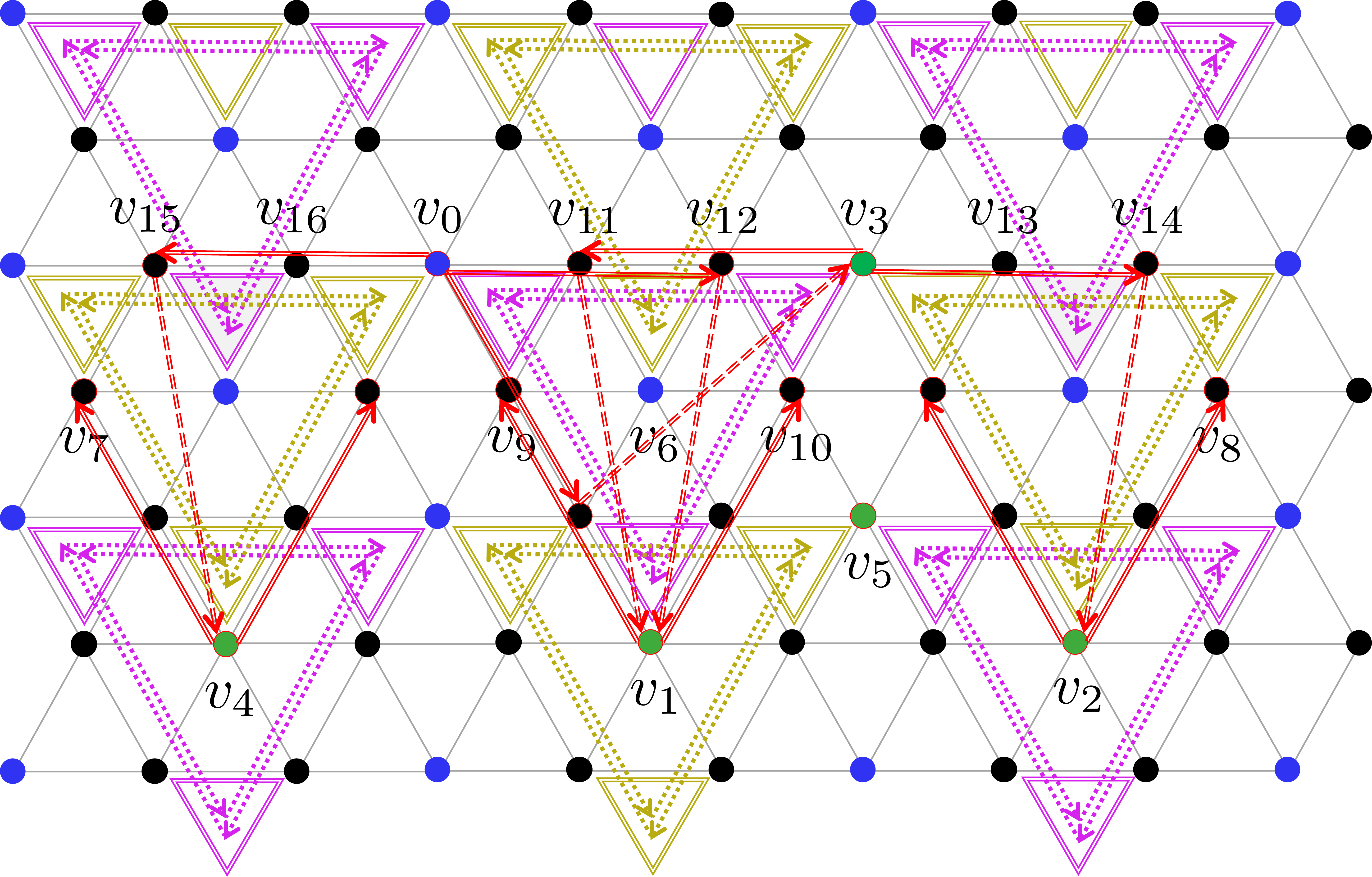}
\includegraphics[width=0.49\textwidth]{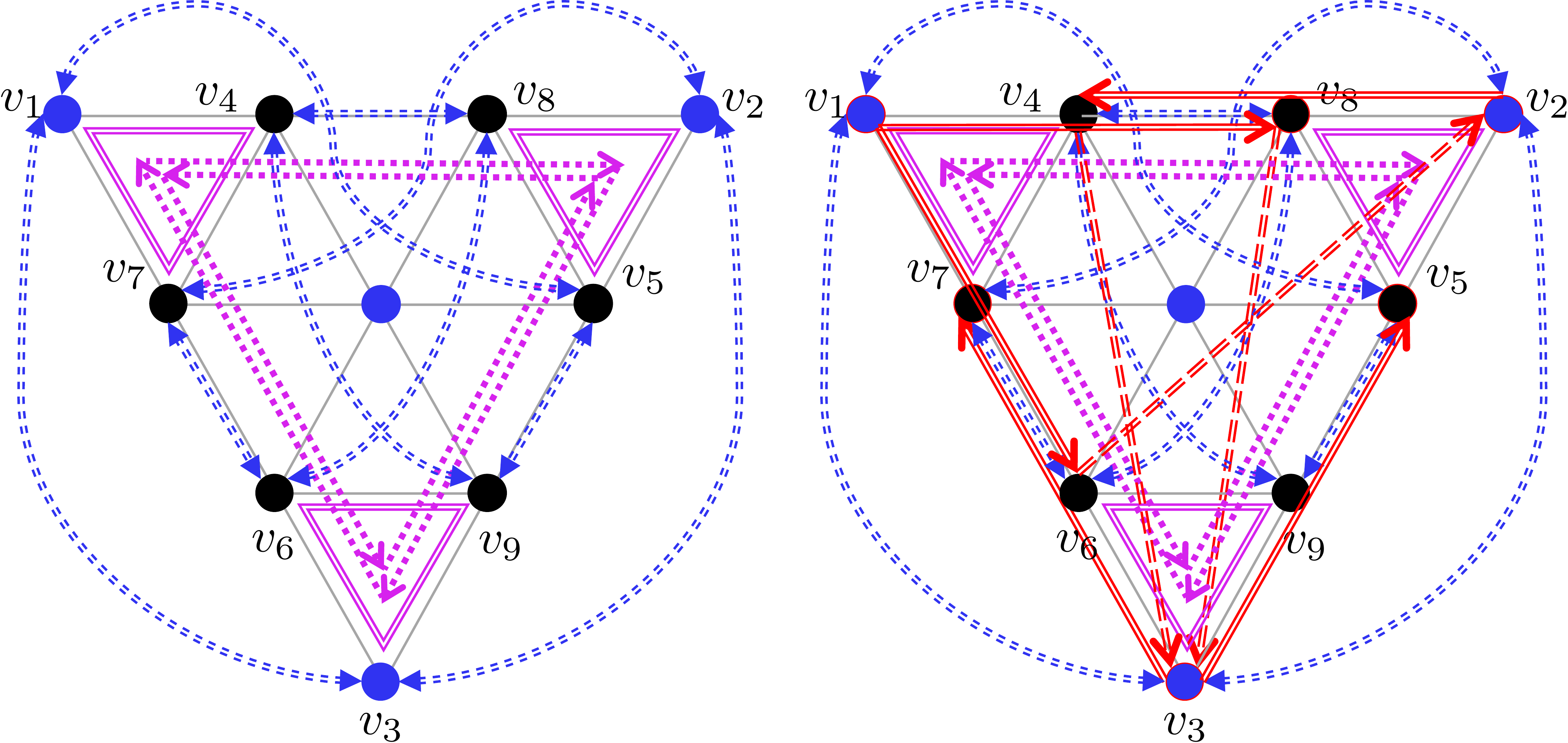}
\caption{DSN (left-most) and the local matching graph specified for $2$-nd and $3$-rd level cooperation in \Cref{exam: exam22}.}
\label{fig: DSNexam22}
\end{figure}

Fig.~\ref{fig: DSNexam22} has three graphs. The left-most one describes the cooperation graph resulting from adding the $2$-nd and the $3$-rd level cooperations among nodes to the DSN in \Cref{exam11}. We adopt the right-most local matching graph in Fig.~\ref{fig: localmatchinggraph} to specify local matching graphs in this example, and it is shown in the central panel. Note that we have exchanged the indices of $v_4$ and $v_7$, and those of $v_5$ and $v_8$ in the original graph to obtain the graph in the center. The right-most panel presents the subgraph in local decoding graphs corresponding to the local matching graph. Let $\bold{u}=\SB{u_{1},u_{2},\dots,u_{p}}$ be an erasure pattern on this DSN. Under the EC solution specified in \Cref{theo: ECcons2}, suppose there exists $\gamma\in\Db{N}$ such that $\gamma_{i;t}=\gamma$, for all $t\in\MB{T}$, and $i\in X_t$.

In the specified cooperation graph, for $i\in\MB{p}$, if $v_i$ is black and is connected to two triangles, then $0\leq u_{i}\leq r_i-\delta-2\gamma$; else if $v_i$ is black and is connected to only one triangle, then $0\leq u_{i}\leq r_i-\delta-\gamma$; else if $v_i$ is blue and is connected to only one triangle, $r_i-\delta-\gamma<u_{1,i}\leq r_i+2\gamma$; else if $v_i$ is green, then $r_i-\delta-2\gamma< u_i\leq r_i+2\gamma$; otherwise $r_i-\delta-2\gamma<u_{1,i}\leq r_i+3\gamma$. Since each non-locally-recoverable node, e.g., $v_1$, $v_2$, and $v_4$, is contained in an isolated local matching graph, i.e., a triangle, laying at the bottom of this triangle, it can obtain additional $2\gamma$ cross parity symbols from it according to the previous discussion in \Cref{exam: examlocal}. Then, $v_3$ and $v_5$ can also obtain extra $2\gamma$ parity symbols, where $\gamma$ of them are from the $2$-nd level cooperation, and the remaining $\gamma$ of them are from the $3$-rd level cooperation, respectively, according to \Cref{exam: examlocal}. After that $v_{0}$ is able to obtain $2\gamma$ cross parity symbols from the $2$-nd level cooperation (the pink triangle), and $\gamma$ cross parity symbols from the $3$-rd level cooperation (the olive triangle). Following a similar logic, all codewords in the non-locally-recoverable nodes are able to be recovered, which means that $\bold{u}$ is a recoverable erasure pattern of this graph but not a recoverable pattern of any of the graphs in Fig.~\ref{fig11}.
\end{exam}

\section{Topology Adaptivity, Scalability, and Flexibility}
\label{section desired properties}

In \Cref{section multi-level cooperation}, we have presented a construction of codes with hierarchical locality for a DSN with a given cooperation graph, which enables the system to offer multi-level access at each node while simultaneously reducing the latency by taking into account the communication cost between different nodes. However, multi-level accessibility is not the only property that is desirable in practical cloud storage applications. In this section, we therefore discuss topology adaptivity, scalability, and flexibility of our construction, which are especially critical in dynamic cloud storage. 

\subsection{Topology Adaptivity}
\label{subsec topology adaptivity}

As discussed in \Cref{sectoin: introduction}, varying topology is a critical property of DSNs because of the dynamic nature of practical networks. While discussing EC solutions for DSNs with a specific topology, the time cost in each communication link and the erasure statistics of each node should also be taken into consideration to have a good trade-off between low latency and high EC capability. Although hierarchical coding schemes over a DSN with a specified cooperation graph has been discussed in Subsection~\ref{subsec cooperation graphs}, the method of finding a cooperation graph over DSNs with arbitrary topology has not yet been discussed. \Cref{algo: search cooperation graph} searches for a cooperation graph over a given network; the existence of such a graph is implicitly proved in the algorithm. Here $G(V,E)$ denotes a DSN with the collection $S(V,E)$ of maximum cliques.

\begin{algorithm}
\caption{Cooperation Graph Search}\label{algo: search cooperation graph}
\begin{algorithmic}[1]
\Require
\Statex $G(V,E)$: existing DSN;

\Statex $a(S)$: the number of different cycle groups associated with the maximum clique $S\in S(V,E)$;
\Statex $b(g)$: the number of cycles within the cycle group $g$;
\Ensure
\Statex $\Se{G}(\Se{V},\Se{E})$: a cooperation graph over $G$;
\Statex //\emph{Find a cooperation graph}
\State $\Se{V}\gets V$, $\Se{E}\gets \varnothing$;
\For{$v_i\in V$}
\State Assign a subset of $\Se{N}_i$ to $\Se{M}_i$;
\State $  \Se{E}\gets \Se{E}\cup \{e_{i,j}:v_j\in \Se{M}_i\}$;
\EndFor
\State Find the collection $\Se{S}(\Se{V},\Se{E})$ of maximum cliques over $\Se{G}$;
\State $t \gets 1$, $g \gets 1$;
\For{$S\in \Se{S}(\Se{V},\Se{E})$}
\For{$1\leq i\leq a(S)$}
\For{$1\leq b\leq b(g)$}
\State Find an edge or a triangle contained in $S$ and denote the set consisting of indices of its vertices by $Y_t$;
\State Find another edge in $G$ if $|{Y}_t|=2$; else find a triangle such that there exists a bijection $f$ from $X_t$ to $Y_t$ and $\{e_{i,j}\}_{(i,j)\in X_t\times Y_t\setminus\{(i,f(i)):i\in X_t\}}\subseteq \bar{\Se{E}}$, where $X_t$ denotes the set consisting of indices of its vertices;
\State $\Se{E}\gets\Se{E}\cup \{e_{i,j}\}_{(i,j)\in X_t\times Y_t\setminus\{(i,f(i)):i\in X_t\}}$;
\State $X_{t;j}\gets X_t\setminus\{f^{-1}(j)\}$, $Y_{t;i}\gets Y_t\setminus\{f(i)\}$, $g_t\gets g$, $l_{t;i}\gets 0$, for $i\in X_t$, $j\in Y_t$; 
\State Add $C_t(X_t,Y_t,\{X_{t;j}\}_{j\in Y_t},\{Y_{t;i}\}_{i\in X_t},g_t,(l_{t;j})_{j\in Y_t})$ to $\Se{G}$;
\State $t\gets t+1$;
\EndFor
\State $g\gets g+1$;
\EndFor
\EndFor
\Statex //\emph{Assign associated cooperation levels to $\Se{G}$, following the notation specified in \Cref{tab: notationtable}}
\For{$v_i\in V$}
\State Find the set $A_i$ consisting of group numbers of those cycle groups that contain at least a cycle $C_t$ with $i\in Y_t$; 
\For{$g\in A_i$}
\State Find the set $R_{i;g}$ consisting of all cycles $C_t$ such that $g_t=g$ and $i\in Y_{t}$;
\State Denote the average distance of all nodes $v_j$, $j\in X_{t;i}$, $t\in R_{i;g}$, by $z_i$; 
\EndFor
\State Order $(z_i)_{i\in \MB{a(S)}}$ from the smallest to the largest, and obtain $(z_{\pi'(i)})_{i\in \MB{a(S)}}$, where $\pi'(i)$ is a permutation of elements from $\MB{a(S)}$;
\For{$g\in A_i$ and $t\in R_{i;g}$}
\State $l_{t;j}\gets \pi'(i)+1 $, for $j\in X_{t;i}$;
\EndFor
\EndFor
\end{algorithmic}
\end{algorithm}

\begin{rem} (Latency Optimization in Cooperation Graphs) One might observe that although \Cref{algo: search cooperation graph} presents a general method to search for a cooperation graph over a given DSN described by $G(V,E)$, the resulting code is not guaranteed to possess optimized latency. Optimization of the construction with the lowest latency is left for future work.
\end{rem}

\subsection{Scalability}
\label{subsec scalability}

As discussed in \Cref{sectoin: introduction}, scalability refers to the capability of expanding the backbone network to accommodate additional workload without rebuilding the entire infrastructure. More specifically, when a new cloud is added to the existing configuration, computing a completely different generator matrix results in changing all the encoding-decoding components in the system, and is very costly. The preferred scenario is that adding a new cloud does not change the encoding-decoding components of the existing clouds.

\begin{figure}[H]
\centering
\includegraphics[width=0.5\textwidth]{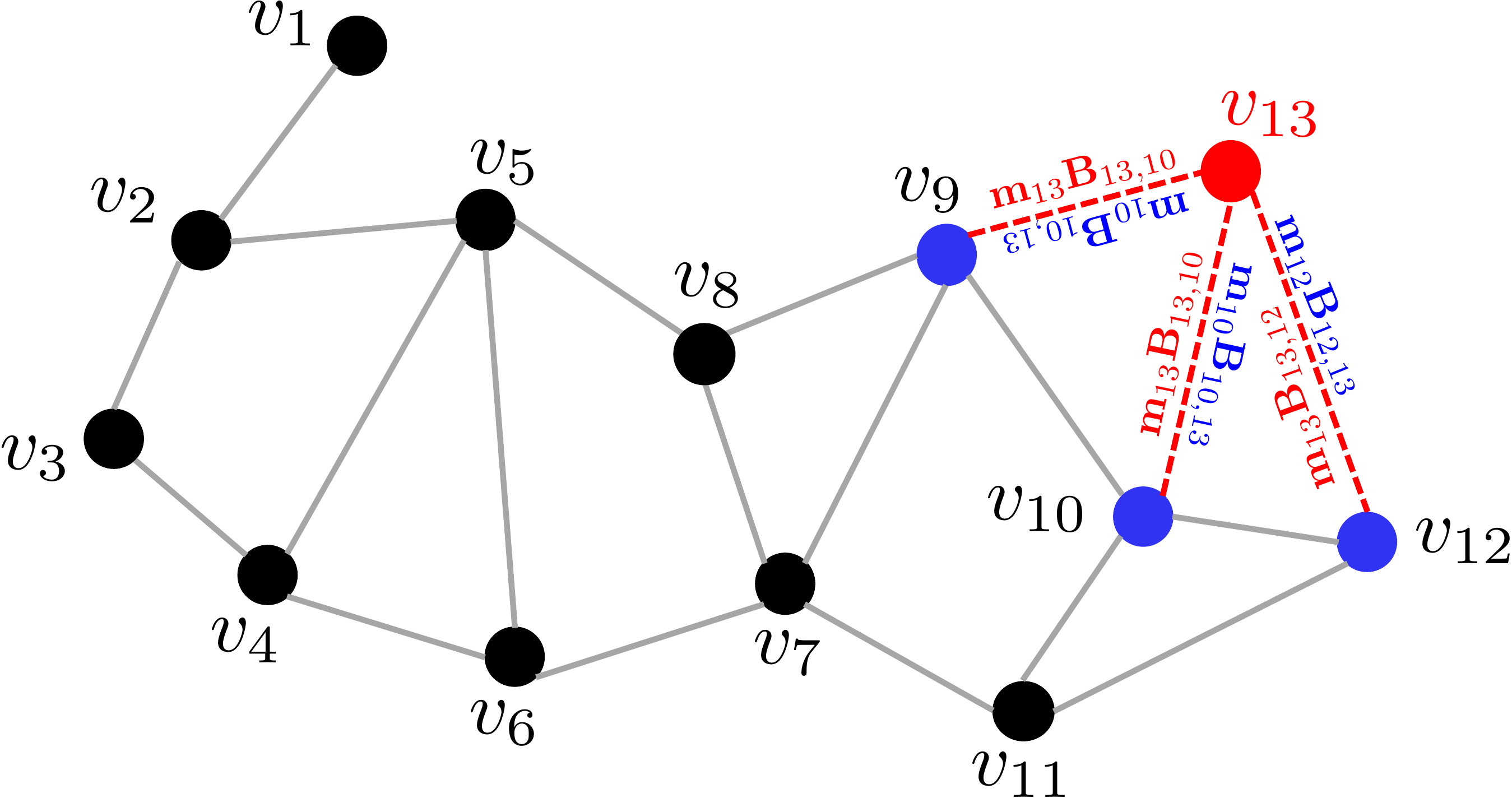}
\caption{Scalability: Add a node to existing DSN.}
\label{fig: addnodeDSN}
\end{figure}

We show that our construction naturally achieves this goal. For simplicity, we only discuss the scalability over constructions with single-level cooperation here. Observe that in \Cref{cons: CRScons}, the components $\bold{A}_{x,x}$, $\bold{U}_x$, and $\bold{B}_{x,i}$, $i\in\MB{p}\setminus\LB{x}$, are built locally. Suppose cloud $p+1$ is added into a double-level configuration adopting \Cref{cons: CRScons}. \Cref{algo: add node} presents a procedure for adding this cloud, which only results in adding some columns and rows to the original generator matrix without changing the existing ones. Thus, the existing infrastructure does not need to be changed; each node only needs to add cross parities it receives from the newly added node to its current parities. Moreover, with this algorithm, the erasure correction capabilities $\{d_{1,i}\}_{v_i\in \Se{N}_{p+1}}$ of neighboring nodes of $v_{p+1}$ are increased by $\delta_{p+1}$.

\begin{algorithm}
\caption{Node Addition}\label{algo: add node}
\begin{algorithmic}[1]
\Require
\Statex $G(V,E)$: existing DSN;
\Statex $p$: number of nodes in $G(V,E)$;
\Statex $v_{p+1}$: the newly added node;
\Statex $r_{p+1}$: the message length of $v_{p+1}$;
\Statex $k_{p+1}$: the number of parity symbols of the $v_{p+1}$;
\Statex $\delta_{p+1}$: the number of additional parities $v_{p+1}$ provides globally to the DSN;
\Statex $\Se{N}_{p+1}$: the set of nodes to be connected to $v_{p+1}$; 
\Statex $\bold{G}$: the original generator matrix;
\Ensure
\Statex $\bold{G}$: the updated generator matrix;
\State Node $v_{p+1}$ chooses its local parameters $\bold{A}_{p+1,p+1}$, $\bold{U}_{p+1}$, and $\bold{B}_{p+1,i}$, $v_i\in\Se{N}_{p+1}$;
\State Node $v_i$ chooses additional cross parity matrices $\bold{B}_{i,p+1}$, $v_i\in\Se{N}_{p+1}$; 
\State Node $v_{p+1}$ sends $\bold{m}_{p+1}\bold{B}_{p+1,i}$ to node $v_i$, $v_i\in \Se{N}_{p+1}$;
\State Node $v_i$ sends $\bold{m}_i\bold{B}_{i,p+1}$ to $v_{p+1}$, $v_i\in \Se{N}_{p+1}$;
\State Node $v_{p+1}$ computes the stored codeword $\bold{c}_{p+1}=\bold{m}_{p+1}\bold{A}_{p+1,p+1}+\sum\nolimits_{i\in\Se{N}} \bold{m}_i\bold{B}_{i,p+1}\bold{U}_{p+1}$;
\State Node $v_i$ adds $\bold{m}_{p+1}\bold{B}_{p+1,i}$ to its current parity symbols, $v_i\in\Se{N}_{p+1}$;
\State Update $\bold{G}$ accordingly;
\end{algorithmic}
\end{algorithm}

\begin{exam} Consider again the set up in \Cref{exam: exam1}. Suppose a node $v_{13}$ is to be added to the existing DSN and is to be connected to nodes $v_{9}$, $v_{10}$, and $v_{12}$, as shown in Fig.~\ref{fig: addnodeDSN}. 
The messages near the edges marked in red are sent from $v_{13}$, while those marked in blue are sent from the neighboring nodes $v_9$, $v_{10}$, and $v_{12}$, to $v_{13}$. 
Note that the new node $v_{13}$ has coding parameters chosen independently from the existing nodes according to \Cref{theo: ECcons1}, which means that it naturally achieves scalability. 
\end{exam}

\subsection{Flexibility}
\label{subsec flexibility}

The concept of flexibility was proposed and investigated for dynamic cloud storage in \cite{martnez2018universal}. In a dynamic cloud storage system, the rate of which a given piece of data is accessed is likely to change. When the data stored at a cloud become hot, i.e., of higher demand, splitting the cloud into two smaller clouds effectively reduces the latency. However, this action should be done without reducing the erasure correction capability of the rest of the system \textcolor{lara}{or} changing the remaining components. 

Specifically, if the data stored at a cloud $v_i\in V$ become unexpectedly hot, the DSN needs to split $v_i$ into two separate smaller clouds $v_{i^{\Tx{a}}}$ and $v_{i^{\Tx{b}}}$ to maintain relatively low latency; \Cref{algo: split node} presents the procedure to do this. For simplicity, we focus here on the case where only the $1$-st level cooperation is involved in, as presented in \Cref{cons: 1}.

\begin{figure}[H]
\centering
\includegraphics[width=0.5\textwidth]{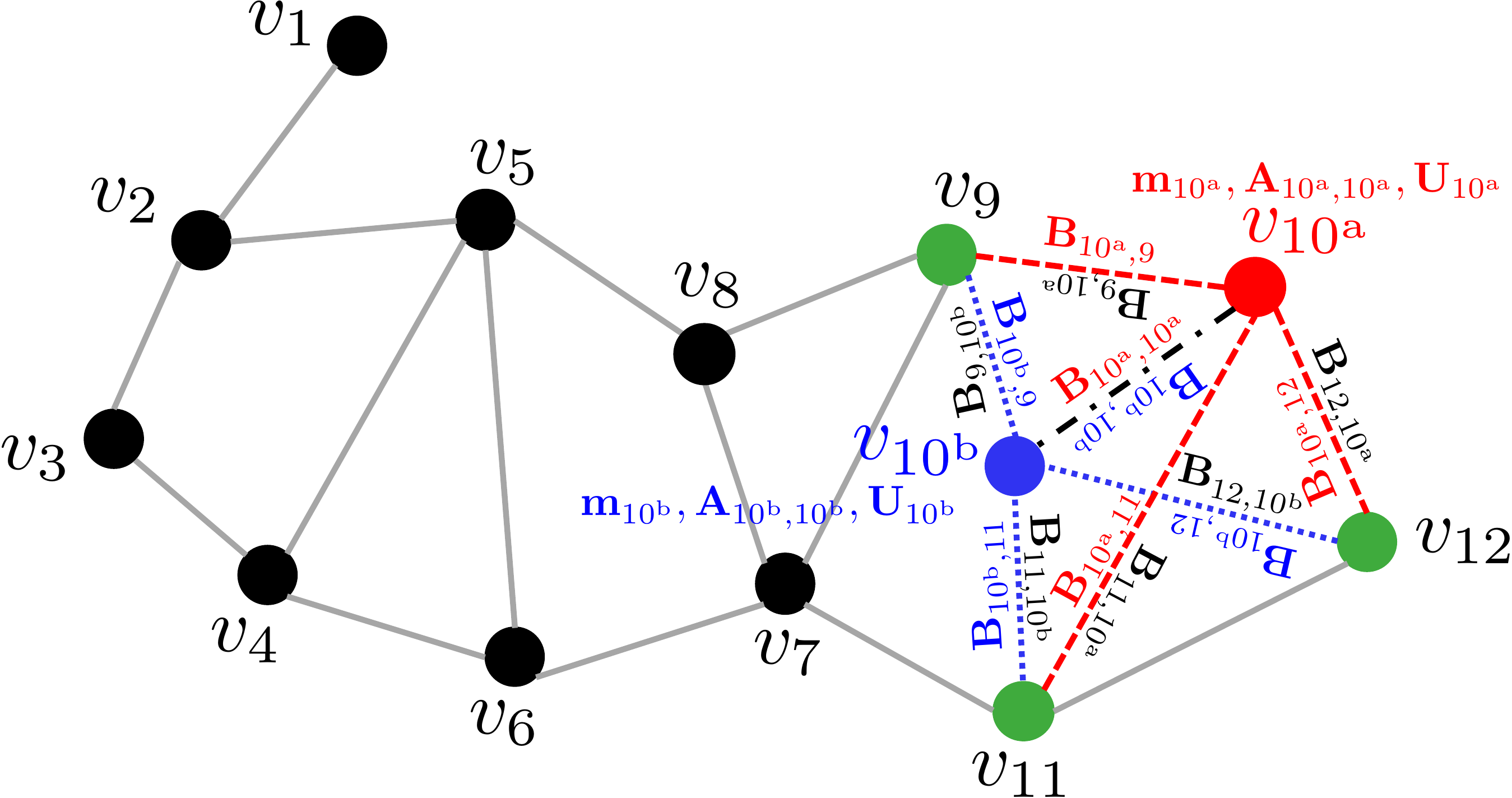}
\caption{Flexibility: Split a node to two nodes in a DSN when it gets hot such that accessing the codeword stored at each one of them achieves low latency.}
\label{fig: splitnodeDSN}
\end{figure}

\begin{algorithm}
\caption{Node Splitting}\label{algo: split node}
\begin{algorithmic}[1]
\Require
\Statex $G(V,E)$: existing DSN;
\Statex $v_i$: the node to be split in $G(V,E)$;
\Statex $\Se{N}_{i}$: the set of neighboring nodes of $v_i$;
\Statex $v_{i^{\textup{a}}}$, $v_{i^{\textup{b}}}$: nodes $v_i$ is split into;
\Statex $k_{i}$, $k_{i^\textup{a}}$, $k_{i^\textup{b}}$: the message lengths of $v_{i}$, $v_{i^{\textup{a}}}$, $v_{i^{\textup{b}}}$, respectively; $k_{i}=k_{i^\textup{a}}+k_{i^\textup{b}}$;
\Statex $r_{i}$, $r_{i^\textup{a}}$, $r_{i^\textup{b}}$: the number of parity symbols of $v_{i}$, $v_{i^{\textup{a}}}$, $v_{i^{\textup{b}}}$, respectively; $r_{i}=r_{i^\textup{a}}+r_{i^\textup{b}}$;
\Statex $\delta_{i}$, $\delta_{i^\textup{a}}$, $\delta_{i^\textup{b}}$: the number of additional parities $v_{i}$, $v_{i^\textup{a}}$, $v_{i^\textup{b}}$ provides globally to the DSN, respectively; $\delta_{i}=\delta_{i^\textup{a}}+\delta_{i^\textup{b}}$;
\Statex $\bold{G}$: the original generator matrix;
\Ensure
\Statex $\bold{G}$: the updated generator matrix;
\State Node $v_i$ splits $\bold{A}_{i,i}$ into $\bold{A}_{i^{\Tx{a}},i^{\Tx{a}}}$, $\bold{B}_{i^{\Tx{b}},i^{\Tx{a}}}$, $\bold{A}_{i^{\Tx{b}},i^{\Tx{b}}}$, $\bold{B}_{i^{\Tx{a}},i^{\Tx{b}}}$ as follows:
\Statex $\bold{A}_{i^{\Tx{a}},i^{\Tx{a}}}=\bold{A}_{i,i}\MB{1:k_{i^{\Tx{a}}}, 1: r_{i^{\Tx{a}}}}$, 
\Statex $\bold{B}_{i^{\Tx{b}},i^{\Tx{a}}}=\bold{A}_{i,i}\MB{k_{i^{\Tx{a}}}+1:k_i, 1: \delta_{i^{\Tx{a}}}}$, 
\Statex $\bold{A}_{i^{\Tx{b}},i^{\Tx{b}}}=\bold{A}_{i,i}\MB{k_{i^{\Tx{a}}}+1:k_i, r_{i^{\Tx{a}}}+1:r_i}$, 
\Statex $\bold{B}_{i^{\Tx{a}},i^{\Tx{b}}}=\bold{A}_{i,i}\MB{1: k_{i^{\Tx{a}}},r_{i^{\Tx{a}}}+1:r_{i^{\Tx{a}}}+\delta_{i^{\Tx{b}}}}$;
\State Node $v_i$ splits $\bold{B}_{i,j}$, $\forall v_j\in \Se{N}_i$, into $\bold{B}_{i^{\Tx{a}},j}$ and $\bold{B}_{i^{\Tx{b}},j}$ as follows:
\Statex $\bold{B}_{i^{\Tx{a}},j}=\bold{B}_{i,j}\MB{1: k_{i^{\Tx{a}}},1:\delta_j}$, 
\Statex $\bold{B}_{i^{\Tx{b}},j}=\bold{B}_{i,j}\MB{k_{i^{\Tx{a}}}+1:k_i,1:\delta_j}$;
\State Node $v_j$, $\forall v_j\in \Se{N}_i$, splits $\bold{B}_{j,i}$ into $\bold{B}_{j,i^{\Tx{a}}}$ and $\bold{B}_{j,i^{\Tx{b}}}$ as follows: 
\Statex $\bold{B}_{j,i^{\Tx{a}}}=\bold{B}_{j,i}\MB{1:k_j,1: \delta_{i^{\Tx{a}}}}$, 
\Statex $\bold{B}_{j,i^{\Tx{b}}}=\bold{B}_{j,i}\MB{1:k_j,\delta_{i^{\Tx{a}}}+1:\delta_i}$, $\forall v_j\in \Se{N}_i$;
\State Node $v_i$ splits $\bold{U}_i$ into $\bold{U}_{i^{\Tx{a}}}$ and $\bold{U}_i^{\Tx{b}}$ as follows:
\Statex $\bold{U}_{i^{\Tx{a}}}=\bold{U}_i\MB{1:\delta_{i^{\Tx{a}}},1:r_{i^{\Tx{a}}}}$,
\Statex $\bold{U}_{i^{\Tx{b}}}=\bold{U}_i\MB{\delta_{i^{\Tx{a}}}+1:\delta_1,r_{i^{\Tx{a}}}+1:r_{i^{\Tx{b}}}}$;
\State Compute the additional cross parities $\bold{s}_i$'s by solving the equation $\bold{s}_i\bold{U}_i=\bold{c}_i-\bold{m}_i\bold{A}_{i,i}$, where $i\in\MB{p}$. Find $\bold{s}_{i^{\Tx{a}}}\in \textup{GF}(q)^{\delta_{i^{\Tx{a}}}}$, $\bold{s}_{i^{\Tx{b}}}\in \textup{GF}(q)^{\delta_{1^{\Tx{b}}}}$ such that $\bold{s}_i=\MB{\bold{s}_{i^{\Tx{a}}},\bold{s}_{i^{\Tx{b}}}}$;
\State Compute the message stored at the node $v_{i^{\textup{a}}}$ and $v_{i^{\textup{a}}}$ as follows: 
\Statex $\bold{c}_{i^{\Tx{a}}}=\MB{\bold{m}_{i^{\Tx{a}}},\bold{m}_{i^{\Tx{a}}}\bold{A}_{i^{\Tx{a}},i^{\Tx{a}}}+\SB{\bold{m}_{i^{\Tx{b}}}\bold{B}_{i^{\Tx{b}},i^{\Tx{a}}}+\bold{y}_{i^{\Tx{a}}}}\bold{U}_{i^{\Tx{a}}}}$,
\Statex $\bold{c}_{i^{\Tx{b}}}=\MB{\bold{m}_{i^{\Tx{b}}},\bold{m}_{i^{\Tx{b}}}\bold{A}_{i^{\Tx{b}},i^{\Tx{b}}}+\SB{\bold{m}_{i^{\Tx{a}}}\bold{B}_{i^{\Tx{a}},i^{\Tx{b}}}+\bold{y}_{i^{\Tx{b}}}}\bold{U}_{i^{\Tx{b}}}}$;
\State Update $\bold{G}$ accordingly;
\end{algorithmic}
\end{algorithm}

Note that the matrix $\bold{B}_{i,j}$ is vertically split into $\bold{B}_{i^{\textup{a}},j}$ and $\bold{B}_{i^{\textup{b}},j}$, while $\bold{B}_{j,i}$ is horizontally split into $\bold{B}_{j,i^{\textup{a}}}$ and $\bold{B}_{j,i^{\textup{b}}}$, for all $v_j$ that are neighboring nodes of $v_i$. Therefore, it is obvious that $\bold{m}_i\bold{B}_{i,j}=\bold{m}_{i^{\textup{a}}}\bold{B}_{i^{\textup{a}},j}+\bold{m}_{i^{\textup{b}}}\bold{B}_{i^{\textup{b}},j}$ and one can prove that the local codeword $\bold{c}_j$ doesn't change for $v_j$ that is a neighboring node of $v_i$. Moreover, since both the local and the global parity-check matrices for each non-split cloud remain unchanged, the local and global erasure capabilities of them are not affected according to \Cref{lemma: DLcodedis}. Furthermore, one can prove that the local codewords stored at the new clouds $i^{\Tx{a}}$ and $i^{\Tx{b}}$ tolerate $(r_{i^{\Tx{a}}}-\delta_{i^{\Tx{a}}})$ and $(r_{i^{\Tx{b}}}-\delta_{i^{\Tx{b}}})$ local erasures, respectively.

\begin{exam} Consider again \Cref{exam: exam1}. If the data stored at node $v_{10}$ become unexpectedly hot, then we split $v_{10}$ into two separate nodes $v_{10^{\Tx{a}}}$ and $v_{10^{\Tx{b}}}$ following \Cref{algo: split node}, as shown in Fig.~\ref{fig: splitnodeDSN}. 

Originally, node $v_{10}$ needs to access all $n_{10}=k_{10}+r_{10}$ symbols to obtain message $\bold{m}_{10}$, which results in high latency when one has to access any set of symbols from $\bold{m}_{10}$ frequently. This \textcolor{lara}{operation} results in unnecessary cost in terms of data processing times, which can be solved by splitting $v_{10}$ into two nodes $v_{10^{\Tx{a}}}$ and $v_{10^{\Tx{b}}}$ that store $n_{10^\Tx{a}}=k_{10^\Tx{a}}+r_{10^\Tx{a}}$ and $n_{10^\Tx{b}}=k_{10^\Tx{b}}+r_{10^\Tx{b}}$ symbols, which contain the information of $\bold{m}_{10^\Tx{a}}$ and $\bold{m}_{10^\Tx{b}}$, respectively. Local access to each one of the two nodes will require significantly lower latency compared with a full access of the original node $v_{10}$. This approach improves the latency especially if the erasures are bursty, i.e., concentrated within any one of $\bold{c}_{10^\Tx{a}}$ or $\bold{c}_{10^\Tx{b}}$. Even if the erasures are distributed more evenly among $\bold{c}_{10^\Tx{a}}$ and $\bold{c}_{10^\Tx{b}}$, the total processing time to obtain $\bold{m}_{10}$ will be the maximum of their individual processing times, which is still much shorter than the original time.
\end{exam}

\section{Conclusion}
\label{section conclusion}

Hierarchical locally accessible codes in the context of centralized cloud networks have been discussed in various prior works, whereas those of DSNs (no prespecified topology) have not been explored. In this paper, we proposed a topology-adaptive cooperative data protection scheme for DSNs, which significantly extends our previous work on hierarchical coding for centralized distributed storage. We discussed the recoverable erasure patterns of our proposed scheme, demonstrating that our scheme corrects patterns pertaining to dynamic DSNs. Our scheme achieves faster recovery speed compared with existing network coding methods, and enables an intrinsic information flow from nodes with higher reliability to nodes with lower reliability that are close to them on the network. Moreover, our constructions are also proved to be scalable and flexible, making them a construction with great potential to be employed in dynamic DSNs. 

\IEEEpeerreviewmaketitle

\section*{Acknowledgment}
This work was supported in part by NSF under the grants CCF-BSF 1718389 and CCF 1717602, and in part by AFOSR under the grant FA 9550-17-1-0291.

\balance
\bibliographystyle{IEEEtran}
\bibliography{ref}

\end{document}